\documentclass[12pt,titlepage]{utarticle}
\usepackage{amsmath}
\usepackage{mathrsfs}
\usepackage{amsfonts}
\usepackage{amssymb}
\usepackage{amsthm}
\usepackage{cite}
\usepackage{mathtools}
\usepackage{graphicx}
\usepackage{color}
\usepackage{xcolor}
\usepackage{ucs}
\usepackage[utf8x]{inputenc}
\usepackage{xparse}
\usepackage{tikz}
\usepackage{subcaption}
\usepackage{csquotes}
\usepackage{booktabs}
\usepackage{enumitem}
\usepackage{pgffor}
\usepackage{textgreek}

\makeatletter
\newcommand{\greekalpha}[1]{\c@greekalpha{#1}}
\newcommand{\c@greekalpha}[1]{%
  {%
    \ifcase\number\value{#1} %
    \or
    \textalpha
    \or
    \textbeta
    \or
    \textgamma
    \or
    \textdelta
    \or
    \textepsilon
    \or
    \textzeta
    \or
    \texteta
    \or
    \texttheta 
    \or
    \textiota
    \or
    \textkappa
    \or
    \textlambda
    \or
    \textmu
    \or
    \textnu
    \or
    \textxi
    \or
    \textomikron
    \or
    \textrho
    \or
    \textpi
    \or
    \textsigma
    \or
    \texttau
    \or
    \textupsilon
    \or
    \textphi
    \or
    \textchi
    \or
    \textpsi
    \or
    \textomega
    \fi
  }%
}
\AddEnumerateCounter*{\greekalpha}{\c@greekalpha}{5}
\makeatother

\usepackage{ytableau}

\definecolor{armygreen}{rgb}{0.29, 0.73, 0.13}

\usepackage{kpfonts}
\usepackage[scaled=0.90]{helvet} 
\usepackage{courier} 
\normalfont
\usepackage[T1]{fontenc}

\usepackage{mathbbol}
\newcommand{\mathds}[1]{\mathbb{#1}}

\definecolor{auburn}{rgb}{0.43, 0.21, 0.1}
\usepackage{hyperref}
\hypersetup{
    colorlinks=true,
    linkcolor=auburn,
    filecolor=magenta,  
    citecolor=magenta,
    urlcolor=blue,
}

\definecolor{aqua}{rgb}{0, 1.0, 1.0}
\definecolor{fuschia}{rgb}{1.0, 0, 1.0}
\definecolor{gray}{rgb}{0.502, 0.502, 0.502}
\definecolor{lime}{rgb}{0, 1.0, 0}
\definecolor{maroon}{rgb}{0.502, 0, 0}
\definecolor{navy}{rgb}{0, 0, 0.502}
\definecolor{olive}{rgb}{0.502, 0.502, 0}
\definecolor{purple}{rgb}{0.502, 0, 0.502}
\definecolor{silver}{rgb}{0.753, 0.753, 0.753}
\definecolor{teal}{rgb}{0, 0.502, 0.502}

\makeatletter
\newdimen\itex@wd%
\newdimen\itex@dp%
\newdimen\itex@thd%
\def\itexspace#1#2#3{\itex@wd=#3em%
\itex@wd=0.1\itex@wd%
\itex@dp=#2ex%
\itex@dp=0.1\itex@dp%
\itex@thd=#1ex%
\itex@thd=0.1\itex@thd%
\advance\itex@thd\the\itex@dp%
\makebox[\the\itex@wd]{\rule[-\the\itex@dp]{0cm}{\the\itex@thd}}}
\makeatother

\makeatletter
\newif\if@sup
\newtoks\@sups
\def\append@sup#1{\edef\act{\noexpand\@sups={\the\@sups #1}}\act}%
\def\reset@sup{\@supfalse\@sups={}}%
\def\mk@scripts#1#2{\if #2/ \if@sup ^{\the\@sups}\fi \else%
  \ifx #1_ \if@sup ^{\the\@sups}\reset@sup \fi {}_{#2}%
  \else \append@sup#2 \@suptrue \fi%
  \expandafter\mk@scripts\fi}
\def\tensor#1#2{\reset@sup#1\mk@scripts#2_/}
\def\multiscripts#1#2#3{\reset@sup{}\mk@scripts#1_/#2%
  \reset@sup\mk@scripts#3_/}
\makeatother

\makeatletter
\newbox\slashbox \setbox\slashbox=\hbox{$/$}
\def\itex@pslash#1{\setbox\@tempboxa=\hbox{$#1$}
  \@tempdima=0.5\wd\slashbox \advance\@tempdima 0.5\wd\@tempboxa
  \copy\slashbox \kern-\@tempdima \box\@tempboxa}
\def\slash{\protect\itex@pslash}
\makeatother

\def\clap#1{\hbox to 0pt{\hss#1\hss}}

\let\oldroot\root
\def\root#1#2{\oldroot #1 \of{#2}}
\renewcommand{\sqrt}[2][]{\oldroot #1 \of{#2}}

\DeclareSymbolFont{symbolsC}{U}{txsyc}{m}{n}
\SetSymbolFont{symbolsC}{bold}{U}{txsyc}{bx}{n}
\DeclareFontSubstitution{U}{txsyc}{m}{n}

\DeclareSymbolFont{stmry}{U}{stmry}{m}{n}
\SetSymbolFont{stmry}{bold}{U}{stmry}{b}{n}

\DeclareFontFamily{OMX}{MnSymbolE}{}
\DeclareSymbolFont{mnomx}{OMX}{MnSymbolE}{m}{n}
\SetSymbolFont{mnomx}{bold}{OMX}{MnSymbolE}{b}{n}
\DeclareFontShape{OMX}{MnSymbolE}{m}{n}{
    <-6>  MnSymbolE5
   <6-7>  MnSymbolE6
   <7-8>  MnSymbolE7
   <8-9>  MnSymbolE8
   <9-10> MnSymbolE9
  <10-12> MnSymbolE10
  <12->   MnSymbolE12}{}

\makeatletter
\def\re@DeclareMathSymbol#1#2#3#4{%
    \let#1=\undefined
    \DeclareMathSymbol{#1}{#2}{#3}{#4}}
\re@DeclareMathSymbol{\neArrow}{\mathrel}{symbolsC}{116}
\re@DeclareMathSymbol{\neArr}{\mathrel}{symbolsC}{116}
\re@DeclareMathSymbol{\seArrow}{\mathrel}{symbolsC}{117}
\re@DeclareMathSymbol{\seArr}{\mathrel}{symbolsC}{117}
\re@DeclareMathSymbol{\nwArrow}{\mathrel}{symbolsC}{118}
\re@DeclareMathSymbol{\nwArr}{\mathrel}{symbolsC}{118}
\re@DeclareMathSymbol{\swArrow}{\mathrel}{symbolsC}{119}
\re@DeclareMathSymbol{\swArr}{\mathrel}{symbolsC}{119}
\re@DeclareMathSymbol{\nequiv}{\mathrel}{symbolsC}{46}
\re@DeclareMathSymbol{\Perp}{\mathrel}{symbolsC}{121}
\re@DeclareMathSymbol{\Vbar}{\mathrel}{symbolsC}{121}
\re@DeclareMathSymbol{\sslash}{\mathrel}{stmry}{12}
\re@DeclareMathSymbol{\bigsqcap}{\mathop}{stmry}{"64}
\re@DeclareMathSymbol{\biginterleave}{\mathop}{stmry}{"6}
\re@DeclareMathSymbol{\invamp}{\mathrel}{symbolsC}{77}
\re@DeclareMathSymbol{\parr}{\mathrel}{symbolsC}{77}
\makeatother

\makeatletter
\def\Decl@Mn@Delim#1#2#3#4{%
  \if\relax\noexpand#1%
    \let#1\undefined
  \fi
  \DeclareMathDelimiter{#1}{#2}{#3}{#4}{#3}{#4}}
\def\Decl@Mn@Open#1#2#3{\Decl@Mn@Delim{#1}{\mathopen}{#2}{#3}}
\def\Decl@Mn@Close#1#2#3{\Decl@Mn@Delim{#1}{\mathclose}{#2}{#3}}
\Decl@Mn@Open{\llangle}{mnomx}{'164}
\Decl@Mn@Close{\rrangle}{mnomx}{'171}
\Decl@Mn@Open{\lmoustache}{mnomx}{'245}
\Decl@Mn@Close{\rmoustache}{mnomx}{'244}
\Decl@Mn@Open{\llbracket}{stmry}{'112}
\Decl@Mn@Close{\rrbracket}{stmry}{'113}
\makeatother

\makeatletter
\DeclareRobustCommand\widecheck[1]{{\mathpalette\@widecheck{#1}}}
\def\@widecheck#1#2{%
    \setbox\z@\hbox{\m@th$#1#2$}%
    \setbox\tw@\hbox{\m@th$#1%
       \widehat{%
          \vrule\@width\z@\@height\ht\z@
          \vrule\@height\z@\@width\wd\z@}$}%
    \dp\tw@-\ht\z@
    \@tempdima\ht\z@ \advance\@tempdima2\ht\tw@ \divide\@tempdima\thr@@
    \setbox\tw@\hbox{%
       \raise\@tempdima\hbox{\scalebox{1}[-1]{\lower\@tempdima\box
\tw@}}}%
    {\ooalign{\box\tw@ \cr \box\z@}}}
\makeatother

\makeatletter
\NewDocumentCommand\mathraisebox{moom}{%
\IfNoValueTF{#2}{\def\@temp##1##2{\raisebox{#1}{$\m@th##1##2$}}}{%
\IfNoValueTF{#3}{\def\@temp##1##2{\raisebox{#1}[#2]{$\m@th##1##2$}}%
}{\def\@temp##1##2{\raisebox{#1}[#2][#3]{$\m@th##1##2$}}}}%
\mathpalette\@temp{#4}}
\makeatletter

\makeatletter
\def\udots{\mathinner{\mkern2mu\raise\p@\hbox{.}
\mkern2mu\raise4\p@\hbox{.}\mkern1mu
\raise7\p@\vbox{\kern7\p@\hbox{.}}\mkern1mu}}
\makeatother

\newcommand{\gt}{>}
\newcommand{\lt}{<}

\theoremstyle{plain}
\newtheorem{theorem}{Theorem}[section]
\newtheorem{lemma}{Lemma}
\newtheorem{prop}{Proposition}
\newtheorem{cor}{Corollary}

\theoremstyle{definition}
\newtheorem{defn}{Definition}

\newtheorem{remark}{Remark}

\usetikzlibrary{,arrows,shapes,decorations.markings}

\makeatletter
\renewcommand*\env@cases[1][1.2]{%
  \let\@ifnextchar\new@ifnextchar
  \left\lbrace
  \def\arraystretch{#1}%
  \array{@{}l@{\quad}l@{}}%
}
\makeatother

\begin{document}

\preprint{UTTG--08--2020\\}

\title{\hspace{-0.1in}Families of Hitchin systems and $\mathcal{N}=2$ theories}

\author{\hspace{-0.1in}Aswin Balasubramanian\footnote{Rutgers University} \hspace{0.15in} Jacques Distler\footnote{University of Texas at Austin} \hspace{0.15in} Ron Donagi \footnote{University of Pennsylvania}}

\Abstract{Motivated by the connection to 4d $\mathcal{N}=2$ theories, we study the global behavior of families of tamely-ramified $SL_N$ Hitchin integrable systems as the underlying curve varies over the Deligne-Mumford moduli space of stable pointed curves. In particular, we describe a flat degeneration of the Hitchin system to a nodal base curve and show that the behaviour of the integrable system at the node is partially encoded in a pair $(O,H)$ where $O$ is a nilpotent orbit and $H$ is a simple Lie subgroup of $F_{O}$, the flavour symmetry group associated to $O$.  The family of Hitchin systems is nontrivially-fibered over the Deligne-Mumford moduli space. We prove a non-obvious result that the Hitchin bases fit together to form a vector bundle over the compactified moduli space. For the particular case of $\overline{\mathcal{M}}_{0,4}$, we compute this vector bundle explicitly. Finally, we give a classification of the allowed pairs $(O,H)$ that can arise for any given $N$.}
%

\maketitle 
\setcounter{tocdepth}{2}
\tableofcontents
 

\section{Introduction}

\subsection{The setup}\label{11_the_setup}

Four dimensional superconformal theories have been a subject of study for many years. Since the work of Seiberg-Witten \cite{Seiberg:1994rs,Seiberg:1994aj}, it has been understood that the low energy physics at a generic point of the Coulomb branch of 4d $\mathcal{N}=2$ theories is succinctly encoded in the geometry of a complex integrable system \cite{Martinec:1995by,Donagi:1995cf,Freed:1997dp}. Integrable systems that arise in this fashion are sometimes called Seiberg-Witten integrable systems. 

 Recently, a class of $\mathcal{N}=2$ theories that admit a uniform geometric construction from six dimensions have received greater attention  \cite{witten1997solutions,Gaiotto:2009we,gaiotto2013wall}. One of the important features shared by all such theories is the fact that their associated Seiberg-Witten integrable systems are isomorphic to particular instances of Hitchin's integrable system\cite{hitchin1987stable}. This includes several familiar theories with UV Lagrangians and more mysterious theories for which there is no known UV Lagrangian.

In this realization from six dimensions, the Hitchin system plays an important role. Specifically, the Coulomb branch associated to the four dimensional theory can be described as the base $B$ of Hitchin's integrable system associated to a simply laced Lie algebra $\mathfrak{j}$ and the UV curve $C_{g,n}$. The choice of the Lie algebra $\mathfrak{j}$ parameterizes the available 6d $(2,0)$ theories and the choice of $C_{g,n}$ determines the 2d surface on which we compactify the 6d theory (together with a partial twist). At the locations of the $n$ punctures, we insert four dimensional defects of the 6d $(2,0)$ theory. The insertions of these defects affects the behaviour of the Hitchin system at these punctures. For the present discussion, we are interested in ``tame defects''. These are the defects that induce a simple pole for the Higgs field in the Hitchin system at the location of the punctures,

\begin{equation}
\Phi = \frac{a}{z} + (\text{regular terms}).
\end{equation}
In order to obtain superconformal field theories (SCFT) using tame defects, we additionally require that $\text{Res}(\Phi) = a$ be a nilpotent element in the Lie algebra $\mathfrak{j}$. What really matters is the $\mathfrak{j}$-conjugacy class to which the element $a$ belongs. So, it is helpful label the Hitchin boundary condition by the nilpotent orbit $O_a$ to which the element $a$ belongs. We will sometimes call this nilpotent orbit the Hitchin orbit $O_H$ associated to the defect.

When $\mathfrak{j}$ is not of type $A$, one needs to further enhance this using some discrete data associated to the defect. The absence of such discrete data for defects in type $A$ is related to the fact that component groups of centralizers of nilpotent orbits are always trivial in type $A$. Let us define

\begin{equation}\label{AOdef}
A(O_a) = C_{J_{ad}} (a)/ C_{J_{ad}}^0 (a)
\end{equation}
to be the group of components of the centralizer of nilpotent orbit $O_a$. Here, $C_{J_{ad}} (a)$ is the centralizer of $\exp(a)$, the unipotent element associated to $a$, in the adjoint group $J_{ad}$ and $C_{J_{ad}}^0 (a)$ is its connected component. The above statement is equivalent to saying that

\begin{displaymath}
A(O_a) = 1
\end{displaymath}
for every nilpotent orbit in type $A$ (see \cite{Chacaltana:2012zy,Balasubramanian:2014jca} for more details). In the discussion below, we confine ourselves to examples from type $A$ Hitchin systems.

\subsection{Weakly coupled gauge groups}\label{weakly}

An important feature of this geometric realization from six dimensions is that the space of marginal parameters associated to the SCFT is identified with the Deligne-Mumford moduli space $\overline{\mathcal{M}}_{g,n}$ of complex structures on $C_{g,n}$. Moving/restricting to a (complex) codimension-one irreducible component of the boundary of $\mathcal{M}_{g,n}$ in $\overline{\mathcal{M}}_{g,n}$ corresponds to the appearance of a weakly coupled gauge group with an associated gauge coupling that is related to plumbing fixture parameter $q$ by

\begin{displaymath}
q \sim e^{2 \pi i \tau}
\end{displaymath}
Further $l$-fold intersections of the boundary correspond to loci where $l$ simple factors in the gauge group become weak. The $(3g-3+n)$-fold intersection of the boundary is zero-dimensional. Each point corresponds to a choice of pants-decomposition of $C_{g,n}$. Each such pants decomposition furnishes a presentation of the class-S theory as a ``gauge theory'' with semi-simple gauge group with $3g-3+n$ simple factors, coupled to $2g-2+n$ (free or interacting) isolated SCFTs corresponding to 3-punctured spheres. Different pants decompositions furnish different (``S-dual'') presentations of the same family of SCFTs as such a gauge theory.

``Generically,'' each factor in the gauge group is just $J$, the compact group\footnote{At the moment, we do not fix the global form of the group $J$. Much of the discussion below will not be sensitive to the exact global form of $J$.}  associated to the complex ADE Lie algebra $\mathfrak{j}$. Interestingly, there are cases where the weakly coupled gauge group that arises on an irreducible component of the boundary is a proper subgroup $H$ of $J$. This striking feature was first noticed in \cite{Argyres:2007cn} where S-duality of the $SU(3),N_f=6$ theory was studied. More general examples were studied using the six dimensional framework in \cite{Gaiotto:2009we,Chacaltana:2010ks}.

The emergence of a proper subgroup $H$ as the weakly coupled gauge group has manifestations for all aspects of the theory. Take, for instance, the Higgs branch of a Class $S[\mathfrak{j},C_{g,n}]$ theory. It can be described as a hyperK\"ahler quotient of the product of Higgs branches of the $2g-2+n$ SCFTs by the action of the $3g-3+n$ simple factors in the gauge group. The reduction of the gauge group from $J^{3g-3+n}$ to a subgroup has an obvious manifestation here. We will not be studying the Higgs branch in this paper, though we will return to this subject briefly in \S\ref{flavour}. 

We will instead study the appearance of the proper subgroups $H$ from the point of view of the Coulomb branch. The Coulomb branch of a Class $S[\mathfrak{j},C_{g,n}]$ theory is the base of a Hitchin integrable system of type $\mathfrak{j}$ on the punctured curved $C_{g,n}$. We would like to understand the implication of the reduction from $J$ to $H$ for the corresponding Hitchin systems. To do this, we will need to elaborate a theory of Hitchin systems on nodal curves which behaves well in families.
In this paper, we will study only the $\mathfrak{j}=\mathfrak{sl}(N)$ case. There are additional complications that arise beyond type-A, which we will leave for a future work. 


\subsection{From Hitchin systems to Good, Ugly and  Bad theories}\label{13_hitchin_systems_corresponding_to_good_theories}

In studying connections between tame Hitchin systems and 4d $\mathcal{N}=2$ theories, it is useful to understand which Hitchin systems correspond to good, ugly or bad\cite{goodbadugly} 4d $\mathcal{N}=2$ theories. This trichotomy of 4d $\mathcal{N}=2$ theories was suggested by \cite{Gaiotto:2012uq} and it can be thought of as a 4d analog of a similar trichotomy arising in 3d $\mathcal{N}=4$ theories \cite{Gaiotto:2008ak}. In \cite{Gaiotto:2012uq}, this trichotomy was proposed using the properties of the Higgs branch. However, for 4d $\mathcal{N}=2$ theories with Coulomb branches described by the tame Hitchin system, this trichotomy can also be understood purely in terms of a Hitchin system with base $B= \bigoplus_{k} H^{0}(C,\mathcal{L}_{k})$ for some line bundles, $\mathcal{L}_{k}$, to be defined below.  With this goal in mind, we introduce the following definitions for any tame Hitchin system on a smooth Riemann surface $C_{g,n}$ : 

\begin{itemize}
\item \emph{Bad}: These are Hitchin systems with $h^1 (C,\mathcal{L}_{k})  > 0$, for some $k$, on the smooth curve C. These are precisely the Hitchin systems where the graded Coulomb branch dimension $h^0 (C,\mathcal{L}_{k}) $ is \text{not} given by $\deg(\mathcal{L}_k) + 1 - g$ for some values of $k$.  When we have $h^1 (C,\mathcal{L}_{k}) =0$ for all values of $k$, we say that the corresponding Hitchin system is \textit{OK} \cite{gunfight}.
\item \emph{Ugly}:  Consider the space of all local mass deformations of the Hitchin system along the lines of \cite{Balasubramanian:2018pbp}. From this space of the deformations, one can define a map $\kappa$ to the space of mass deformations of the global Hitchin system : 
 \begin{equation}
 \kappa : \{ m_i \} _{\text{local}} \rightarrow  \{ m_i \}_{\text{global}} ,
 \end{equation} 
where $\{ m_i \}_{\text{global}}$ is the space of mass deformations of the spectral curve.

Ugly Hitchin systems are those which are OK but have a non-trivial kernel for the map $\kappa$. Correspondingly, we will call a Hitchin system with $\dim(\text{ker}(\kappa)) > 0$ ``\textit{ugly}." An extreme case of an ugly theory is one consisting of free hypermultiplets. The corresponding Hitchin moduli space is a point and hence there are no global mass deformations of the Hitchin system in this case, whereas the SCFT has relevant deformations corresponding to turning on hypermultiplet masses. 
\item \emph{Good}: These are Hitchin systems which are OK and not ugly. 
\end{itemize}

For the purposes of this paper, we will mostly consider Hitchin systems that are \text{OK} on $C_{g,n}$.  For all such theories, the Deligne-Mumford moduli space $\overline{\mathcal{M}}_{g,n}$ can be identified with the space of marginal parameters of the corresponding SCFT and we will rely on this identification to study the weakly coupled gauge groups in \S\ref{classifying}. Much of our paper treats the good and ugly cases on an equal footing, so we do not dwell on the differences between the two cases. 


However, to understand certain aspects of the story, we will need to include some bad Hitchin systems on $C_{0,n}$ in the discussion. We discuss this briefly in \S\ref{stability} and Appendix \ref{OKappendix} and leave a more detailed analysis to a future work. These theories also happen to be the ones for which the relation between our terminology and its physics interpretation is subtle. For certain \textit{bad} Hitchin systems, it will turn out that the corresponding 4d theory is a theory of free hyper-multiplets or a perfectly \textit{good} 4d SCFT. There is, however, a point of view proposed in \cite{Gaiotto:2011xs} according to which one should think of the corresponding compactified 6d theory (with finite area for the Riemann surface C) as being a \textit{bad} theory in these cases. Our OK/Bad dichotomy for Hitchin systems is more directly related to this point of view.

\subsection{Outline of the paper}

The paper is structured in the following way. We review properties of the tame Hitchin system on a smooth underlying curve in \S\ref{2_tame_hitchin_systems_on_smooth_curves}. We then build a global model for the Hitchin system over $\overline{\mathcal{M}}_{0,4}$ in \S\ref{Hitchinglobal}. In particular, in \S\ref{32_the_bundle_of_hitchin_systems}  we note that the Hitchin bases fit together to form a nontrivial vector bundle over $\overline{\mathcal{M}}_{0,4}$ .  We compute that bundle explicitly in \S\ref{directimages}.
We then use this model to take a first look at the Hitchin system on nodal curves in \S\ref{4_standard_and_restricted_nodes}. In this section, we also define what it means for a node to be \textit{standard} (\S\ref{41_the_standard_node}) or \textit{restricted} (\S\ref{first_restricted_nodes}). The restricted nodes are labeled by a pair, $(O,H)$, where $O$ is a nilpotent orbit in $\mathfrak{sl}(N)$ and $H$ is an $SU(l)$ or $Sp(l)$ subgroup of $SU(N)$. We should emphasize that the pair $(O,H)$ is not a complete invariant of the singular spectral curve which covers the node. Examples \hyperlink{example_3}{3} and \hyperlink{example_4}{4} of \S\ref{examples} have the same $(O,H)=([4],SU(2))$, but the singularity structure of the spectral curve is different.

In \S\ref{Hitchinnodal}, we take the results of \S\ref{4_standard_and_restricted_nodes} as motivation to build a general framework for the Hitchin system on a nodal curve such that the family of Hitchin systems on a family of smooth curves is \emph{flat} in the limit as the smooth curve degenerates to the nodal one. Over the interior of the moduli space, the $\mathcal{L}_k$ fit together to form line bundles over the universal curve $\mathcal{C}\to \mathcal{M}_{g,n}$. In extending them to the boundary, we encounter an interesting phenomenon. For a \emph{restricted node}, in which ``$O$" is not the regular nilpotent, the $\mathcal{L}_k$ extend to line bundles $\mathcal{L}'_k$ which are the ``naive" $\mathcal{L}_k$ twisted by a (negative) power of a line bundle whose divisor is a component of the boundary in $\pi\colon\overline{\mathcal{C}}\to \overline{\mathcal{M}}_{g,n}$. For any \emph{given} $C$, the Hitchin base is $B=\bigoplus_k H^0(C,\mathcal{L}'_k)$. These vector spaces fit together (see Theorem  \ref{hitchinglobalthm}) to form a nontrivial vector bundle $\mathcal{B} =\bigoplus_k  \pi_* \mathcal{L}'_k$ over $\overline{\mathcal{M}}_{g,n}$.

The possible restricted nodes are strongly constrained by physics considerations arising from the role of the flavour symmetry, as we show in \S\ref{flavour}. In \S\ref{classifying}, we provide a classification of the allowed nodal degenerations using the methods of \S\ref{Hitchinnodal}. This classification is summarized in Theorem \ref{ohtheorem}. We also show that the results of \S\ref{classifying} are compatible with those in \S\ref{flavour}. 

In Appendix \ref{Proof}, we provide a proof of Theorem \ref{hitchinglobalthm}. In Appendix \ref{OKappendix}, we discuss the close relationship between our \textit{OK} condition and the semi-stability condition for Higgs bundles. We also state a conjecture relating the \textit{OK} condition to a corresponding Deligne-Simpson problem.

\subsection{Further directions} 

As motivation for future work, we mention here some further directions in which our work could be extended or applied. 

\begin{enumerate}

\item \emph{Other approaches to Higgs bundles on nodal curves} 

There has been considerable prior work on studying the moduli of bundles and parabolic bundles on nodal curves (see, for example, \cite{seshadri14023fibres,gieseker1984degeneration,bhosle1992generalised,pandharipande1996compactification}). There has been some recent progress on extending some of these results to Higgs bundles on nodal curves  \cite{bhosle2013generalized,balaji2015degeneration,logares2018higgs}. See also \cite{nekrasov1996holomorphic,talalaev2004hitchin} for some earlier work in this direction. For our purposes, it is important to understand how the \emph{family} of integrable systems behaves in the nodal limit. This appears to not have been addressed previously in the mathematical literature. So, we develop this from the basics. Relating our work to the framework of \cite{bhosle2013generalized,balaji2015degeneration,logares2018higgs} is an interesting direction for future work. 

\item \textit{Solutions to Hitchin's equations in the nodal limit}

Solutions to Hitchin's equations in the nodal limit of the base curve have also been studied recently in \cite{swoboda2017moduli} for the $\mathfrak{sl}(2)$ case with no punctures. Our classification of restricted nodes should also have interesting consequences for a higher-rank tame analog of \cite{swoboda2017moduli}.

\item \textit{Global topology of the character variety}

Another direction in which our framework could be used is in the study of the character variety. The character variety is defined to be the space of maps $\pi_1(C_{g,n})\rightarrow SL(N,\mathbb{C}$) and it is related to the moduli space of Higgs bundles through the non-abelian Hodge correspondence \cite{simpson1990harmonic,de2012topology}. Unlike the geometry of Higgs bundles, the geometry of the character variety is independent of the choice of a complex structure on $C$. In particular, this means we could choose to work with \textit{any} complex structure on $C$ and then use the non-abelian Hodge correspondence to obtain the character variety. 

A specific application in this direction would be to study the global topology of the character variety \cite{hausel2008mixed} from the point of view of Higgs bundles on a nodal curve. This is similar in spirit to the work in \cite{ramadas1996factorisation} where the Verlinde formula (in the $\mathfrak{sl}(2)$ case) was proven by studying the factorization properties of the generalized theta divisor in the nodal limit \cite{narasimhan1993factorisation}. 

\item \textit{Higher Fenchel-Nielsen coordinates}

Our results could also be of use in the study of natural Darboux coordinates on the moduli space of flat connections and/or the character variety and the behaviour of these coordinates under different choices of pants decompositions of the underlying Riemann surface. In the $\mathfrak{sl}(2)$ case, for every choice of a pants decomposition of the Riemann surface, there is a natural set of Darboux coordinates on the character variety called the Fenchel-Nielsen length and twist coordinates (see \cite{goldman2003complex} for a review). In the $\mathfrak{sl}(2)$ case, we get $(3g-3+n)$ pairs of length and twist coordinates - one pair each for every closed curve in $C_{g,n}$. In the higher rank cases, it is again possible to define analogous coordinates for every choice of a pants decomposition. It turns out that complex higher Fenchel-Nielsen coordinates arise from $\mathcal{N}=2$ theories as a natural system of Darboux coordinates on the Hitchin moduli space \cite{nekrasov2014darboux,hollands2016spectral}. In this context, they have recently been studied in specific higher rank examples employing different points of view \cite{Hollands:2016kgm,Hollands:2017ahy,Coman:2017qgv,Jeong:2018qpc,Brennan:2019hzm}. Closely related real Fenchel-Nielsen coordinates for higher Teichmuller spaces\footnote{These correspond to subspaces in the character variety where we only consider representations of the form $\pi_1(C_{g,n}) \rightarrow SL(n,\mathbb{R})$} go back to the work of \cite{goldman1990convex} for the $\mathfrak{sl}(3)$ case and have recently been studied in \cite{zhang2015degeneration} for the $\mathfrak{sl}(N)$ case.

Independent of the methods used, a new feature that one notices in the higher rank cases is that there are non-trivial coordinates associated to a thrice punctured sphere. These coordinates are sometimes denoted as \textit{internal} Fenchel-Nielsen coordinates \cite{zhang2015degeneration,wienhard2018deforming}. And as in the $\mathfrak{sl}(2)$ case, we continue to have coordinates attached to the nodes themselves. In the case of a standard node, the number of coordinates attached to the node is $2 rank(G)$. But, in the case of a restricted node, there is a reduction in this number to $2 rank(H)$. The existence of restricted nodes (see examples in \S\ref{examples}) also makes it clear that the number of internal and nodal (or center) parameters need not be separately invariant under changes of pants decompositions. It is an interesting problem to study the precise relationship between the coordinates arising from different choices of pants decompositions. For the classical Fenchel-Nielsen coordinates, this has been done in \cite{takayuki1993effects,alessandrini2012behaviour}. We believe our results on the allowed restricted nodes (in \S\ref{restricted_nodes} and \S\ref{classifying}) will be helpful in finding such relationships in the higher rank cases.

\item \textit{The Deligne-Simpson problem}

Finally, we would like to mention a conjectural application of our results to the existence problem for tame, irreducible $SL_N$ character varieties. When the underlying Riemann surface is a $n$-punctured sphere $C_{0,n}$, this problem has been studied by Deligne and Simpson \cite{simpson1992products}. To solve this problem, one needs to provide conditions under which tame, irreducible $SL_N$ character varieties are guaranteed to exist. For a particular class of examples, Simpson \cite{simpson1992products} obtained a pair of geometric conditions that achieve this goal. 
 
 Interestingly, we find that our \textit{OK} condition on tame Higgs bundles (with nilpotent Higgs fields) has a close connection to the existence problem for the corresponding character variety. 
 Specifically, we prove in Appendix \ref{OKappendix} that the \textit{OK} condition is necessary and sufficient for Simpson's conditions (from \cite{simpson1992products}) to hold for the corresponding character variety. This result is, however, limited to the case where at least one of the punctures has a regular residue for the Higgs field. In Appendix \ref{OKappendix}, we outline a conjecture for the more general cases.

\end{enumerate}

\section{Tame Hitchin Systems on Smooth Curves}\label{2_tame_hitchin_systems_on_smooth_curves}

In this and subsequent sections, we will be relying on many standard results about the moduli space of curves and linear systems on families of curves.  We refer the reader to \cite{harris2006moduli,arbarello2011geometry} for an exposition of these results. 

Recall that the total space of the Hitchin system is the moduli space $Higgs$ of $J_{\mathbb{C}}$-Higgs bundles which are defined to be moduli space of pairs $(V,\Phi)$ where $V$ is a principal $J_{\mathbb{C}}$ bundle and $\Phi \in H^0(C,ad(V)\otimes K)$ in the case without ramification. In this paper, we will mostly take $Higgs$ to also obey an appropriate stability condition with the exception being the discussion in \S\ref{stability}.  Hitchin observed \cite{hitchin1987stable} that there is a natural map which is now called the Hitchin map :
\begin{equation}
\mu : Higgs \rightarrow \bigoplus_{k} H^0(C, K^{\otimes k})
\end{equation}
where $k$ runs over the degrees of J-invariant polynomials on $\mathfrak{j}$ ($k=2,3,\dots,N$ for $\mathfrak{j}=A_{N-1}$). We denote the image $\bigoplus_{k} H^0(C, K^{\otimes k})$ as the base $B$ of the Hitchin system. The fibers $\mu^{-1}(b)$ over some generic point $b \in B$ are complex Lagrangian tori. In other words, $(Higgs, \mu)$ defines a complex integrable system. The fibers of $\mu$ admit succinct descriptions in terms of Jacobians/Prym varieties associated to the spectral/cameral curves built out of $\Phi$ \cite{beauville1989spectral,donagi1995spectral}. 

There is a further generalization where we replace $K$ by a more general line bundle. We are, in particular, interested in the case where $K$ is replaced by $K(D)$ where $D$ is a divisor of marked points on $C$. Such a replacement leads us to the meromorphic Hitchin system. In this setting, we have Higgs field $ \Phi \in H^0(C,ad(V)\otimes K(D))$. The resulting moduli space $Higgs_D$ of pairs $(V,\Phi)$ is a Poisson manifold. If we restrict the residues of $\Phi$ at the marked points to be fixed conjugacy classes of $\mathfrak{j}$, then we restrict to a particular symplectic leaf in the Poisson manifold. The Hitchin map $\mu$, when restricted to this symplectic leaf, again describes a complex integrable system\cite{markman1994spectral,bottacin1995symplectic}. 

When the Hitchin system is associated to a 4d $\mathcal{N}=2$ theory, one can deduce the geometry of the integrable system by formulating the 4d $\mathcal{N}=2$ theory on $\mathbb{R}^{1,2}\times S^1_R$ and studying how the resulting moduli space is fibered over the 4d Coulomb branch\cite{Seiberg:1996nz,Gaiotto:2008cd}. This argument is based on constraints from supersymmetry and the nature of the $R \rightarrow 0$ limit which corresponds to a dimensional reduction of the 4d $\mathcal{N}=2$ theory to a 3d $\mathcal{N}=4$ theory. These arguments also carry over to the case where the base curve $C$ develops a nodal singularity. In particular, we expect the Hitchin map $\mu$ to be Lagrangian. There is, however, one important new feature and this has to do with the fact that the Hitchin map $\mu$ could fail to be proper when $C$ is singular. This means that some of the fiber directions of $\mu$ may no longer be compact. Physically, this is to be expected since the spectral curve $\Sigma_b$ is singular in the limit where we take $Im(\tau_{UV}) \rightarrow \infty $ and the fibers of the Hitchin map are the generalized Jacobians associated to singular curve. The base directions which are symplectic dual to the noncompact fiber directions become additional Casimir parameters in the sense of \cite{markman1994spectral}. For reasons that will be explained in  \S\ref{4_standard_and_restricted_nodes} and \S\ref{Hitchinnodal}, we will denote these additional Casimir parameters as \textit{center parameters}. These center parameters will turn out to play an important role in our discussions.

\subsection{Nilpotent orbits and spectral curves: local story on a smooth curve.}\label{21_nilpotent_orbits_and_spectral_curves_local_story_on_a_smooth_curve}

We work with the Hitchin system for $J_{\mathbb{C}}=SL_N$ on a smooth curve $C$ with marked points in a reduced divisor $D = \sum_i  p_i$. 

At each point $p_i$, we insert a regular four dimensional defect of the 6d $(2,0)$ theory $\mathscr{X}[A_{N-1}]$. The effect of this defect is to produce a simple pole in the Higgs field
\begin{equation}
\Phi = \frac{a}{z} + \ldots
\end{equation}
where $\text{Res}(\Phi)=a$ is an element of the complex Lie algebra $\mathfrak{j}$ associated to $J$. Since we want to study tame Hitchin systems corresponding to \emph{conformal} theories, we additionally assume that $a$ is a nilpotent element in $\mathfrak{j}$. There is a natural $\mathfrak{j}$ action on the adjoint valued Higgs field $\Phi$. Inequivalent boundary conditions are labeled by the conjugacy class $O_a$ to which the residue $a$ belongs. In type $A$, nilpotent orbits can be classified by using the Jordan normal form and counting the sizes of the Jordan blocks. We label an orbit by a partition of $N$, which we can equally-well think of as the heights of the columns of a Young diagram. The partition (or Young diagram) is called the Hitchin label for the defect in the physics literature. There is a related, dual label called the Nahm label which is more directly associated to the Higgs branch. For type A Hitchin systems, the Nahm label is just given by the transpose partition. Since our study here will be confined to the Coulomb branch, we will privilege the Hitchin label over the Nahm label for most of the paper. However, in \S\ref{flavour}, we will discuss the flavour symmetry and the Higgs branch and for those discussions, the Nahm label is more natural. 

The spectral curve is now given by $w^N = \sum_{ k=2}^N   a_k  w^{N-k}.$ Here $a_k$ is a pluridifferential on $C$ i.e. a section of $(K_C)^{\otimes k}$ with allowed pole of order up to $\pi_k$. Alternatively, it is a section on $C$ of $(K_C(D))^{\otimes k}$ with a zero of order $\ge \chi_k$, where $\pi_k + \chi_k = k.$

For a given nilpotent orbit $O_i$ inserted at $p_i$, the order of  the zero $\chi_k^{(i)}$ of $a_k$ at $p_i$ is the column-number of the column containing $k^{\text{th}}$ box in the Young diagram corresponding to $O_i$ (where the boxes of the Young diagram are labeled consecutively, starting with the first box of the first column and proceeding vertically and then to the right). This determines $\pi_k^{(i)} =k - \chi_k^{(i)}$. For example, the regular nilpotent orbit (partition $[N]$) gives orders of vanishing $\chi_k^{(i)} =1$ for all $k$, the subregular orbit (partition $[N-1,1]$) gives $\chi_k^{(i)} =1$ for $k \lt N$ and $\chi_N^i =2$. At the opposite extreme, the minimal nilpotent orbit (partition $[2,1^{N-1}]$) gives  $\chi_k^{(i)} = k-1$.

Note that the orbit $O$ determines a generic form of the spectral cover $\Sigma_O$. The actual cover could be any specialization of the generic form, i.e. the orders of vanishing of the coefficients are allowed to go up but not down.

More precisely, the orbit determines not the type of singularity of the spectral curve but the local structure of the spectral sheaf on it. For example, a matrix is regular if and only if it has a one dimensional eigenspace per eigenvalue. In the Hitchin moduli space, this implies that an orbit at $p$ is regular  if and only if the spectral sheaf has rank 1 everywhere above $p$, i.e. it is a line bundle on the spectral curve $\Sigma$ near the inverse image of $p$ \cite{beauville1989spectral}. If $\Sigma$  is non-singular then all spectral sheaves on it are line bundles, so the Hitchin fiber (= the Jacobian $J(\Sigma)$ ) consists only of line bundles. On a singular spectral curve $\Sigma$ , most spectral sheaves are still line bundles, but some are not. For example, when $\Sigma$ is an irreducible nodal curve, the fiber is the compactified Jacobian $\overline{J(\Sigma)}$. This has the Jacobian $J(\Sigma)$ as a dense open subset, but the other (closed, lower dimensional) stratum consists of sheaves that are not line bundles - they have rank 2 at the node, arising instead as direct images of line bundles on the normalization of $\Sigma$. So the regular orbit can correspond to (line bundles on) either smooth or arbitrarily singular spectral curves, the subregular orbit corresponds to a spectral curve with at least a node (and a sheaf that has rank exactly 2 at one point above the singularity) and so on.

Conversely, a given spectral cover $\Sigma$ determines a smallest orbit $O_{\Sigma}$. The actual orbit obtained from some sheaf on $\Sigma$ may be any orbit containing $O_{\Sigma}$ in its closure. If $\nu: N \to \Sigma$ is the normalization of spectral curve, then the smallest orbit $O_\Sigma$  corresponds to the sheaf $\nu_*(\mathcal{O})$, while the largest (= regular) orbit  corresponds to the structure sheaf $\mathcal{O}_{\Sigma}$. If the spectral curve were to have a nodal singularity, then $O_{\Sigma}$ is the subregular orbit and so on. 

\subsection{Nilpotent orbits and spectral curves: global story on a smooth curve}\label{22_nilpotent_orbits_and_spectral_curves_global_story_on_a_smooth_curve}

Now consider the global situation, taking $C\coloneqq\mathbb{P}^1$. The coefficient $a_k$ is a section of a line bundle, $L_k$, of degree:

\begin{equation}
\deg(L_k)=
k (-2+\deg(D)) - \sum_{p_i \in D} \chi_k^{(i)}  =
-2k + \sum_{p_i \in D} \pi_k^i .
\label{degrees}
\end{equation}
The space of all such sections is a vector space of dimension:
\begin{equation}
 b_k\coloneqq max(1+\deg(L_k), 0).
\end{equation}
As an example, consider the case when $\chi_k^{(i)} = 1$ for all $i,k$. We then have $\deg(L_k) = -2k + \deg(D)(k-1),\, k=2,3,\ldots N$. If we have $\deg(D) \geq 3$ then $\deg(L_k) \geq -1$. So, the dimension $b_k$ of the Hitchin base $B$ in degree $k$ is just given by
\begin{equation}
b_k = k(\deg(D)-2) + 1 - \deg(D)
\end{equation}
Summing over degrees, we get 
\begin{equation}
\label{basedimension}
\begin{split}
\text{dim}(B) &= \sum_{k=2}^{N} k(\deg(D)-2) + 1 - \deg(D)\\
 &= (-1)(N^2-1) + \frac{\deg(D)(N^2-N)}{2}
\end{split}
\end{equation}

The dimension of the total space $Higgs_D$ in this case can be easily computed using Riemann-Roch, cf. \cite{markman1994spectral,bottacin1995symplectic}. Alternatively, this can be evaluated using the non-abelian Hodge correspondence and the realization of the Hitchin moduli space as the character variety $\pi_1(C_{0,k})\rightarrow SL_N$ with fixed regular holonomy around each puncture (\cite{hausel2011arithmetic}) :  
\begin{equation}
\label{totaldimension}
\begin{split}
\text{dim}(Higgs_D) &= (-2)(N^2-1) + \sum_{i=1}^{\deg(D)} \text{dim}(O_{\text{reg}})\\
&= (-2)(N^2-1) + \deg(D)(N^2-N) 
\end{split}
\end{equation}
where we have used the fact that the dimension of a regular orbit in $SL_N$ is $\text{dim}(O_{\text{reg}})=(N^2-N)$.

Comparing \eqref{basedimension} and \eqref{totaldimension}, we see that
\begin{equation}
\text{dim}(Higgs_D) = 2 \text{dim}(B).
\end{equation}
This is in keeping with our expectations since $(Higgs_D,\mu)$, where $\mu$ is the Hitchin map restricted to a symplectic leaf of $Higgs_D$, defines a complex symplectic integrable system.

\section{Global Interlude I : Hitchin System over $\overline{\mathcal{M}}_{0,4}$}\label{Hitchinglobal}

In this section, we will study global aspects of the Hitchin system on family of curves by specializing to the case of a four punctured sphere. We will then use this global model in \S\ref{41_the_standard_node} to take a first look at the Hitchin system on a nodal curve. 

\subsection{A global model for the universal curve $\overline{\mathcal{C}}_{0,4}$}\label{31_a_global_model_for_the_universal_curve_}

Let the locations of four punctures be $z_1,z_2,z_3,z_4$ and let $\lambda$ be their cross ratio

\begin{displaymath}
\lambda = \frac{(z_1-z_3)(z_2-z_4)}{(z_1-z_4)(z_2-z_3)}
\end{displaymath}
We use the following global model for the universal curve $\pi\colon \overline{\mathcal{C}}\to\overline{\mathcal{M}}_{0,4}$.

Consider $\mathbb{CP}^2$ blown up at four points : $E_1\to(1,0,0),\, E_2\to(0,1,0),\,E_3\to(0,0,1),\,E_4\to(1,1,1)$. Let us denote the blown up surface as $\widetilde{\mathbb{CP}}^2$. Let $\lambda_1,\lambda_2$ be homogeneous coordinates on $\overline{\mathcal{M}}_{0,4}=\mathbb{CP}^1$. The cross ratio $\lambda = \lambda_1/\lambda_2$.

We identify the universal curve, $\overline{\mathcal{C}}\simeq \widetilde{\mathbb{CP}}^2$ and the projection $\pi\colon \widetilde{\mathbb{CP}}^2\to \overline{\mathcal{M}}_{0,4}$ is defined as the solution to

\begin{equation}
\lambda_1 x(y-z) + \lambda_2y(z-x) = 0
\label{universalcurve}\end{equation}
which determines $\lambda_{1,2}$ up to a common scaling. Here $x,y,z$ are (the pullbacks to $\widetilde{\mathbb{CP}}^2$ of) the standard projective coordinates on $\mathbb{CP}^2$. As a function on $\mathbb{CP}^2$, the ratio $\lambda= \lambda_1/ \lambda_2$ is well defined except at the four points $E_i$. It extends to give a well defined morphism $\pi: \widetilde{\mathbb{CP}}^2 \to \mathbb{CP}^1$ on the blowup $\overline{\mathcal{C}}= \widetilde{\mathbb{CP}}^2$.

For generic, $\lambda=\lambda_1/\lambda_2$ the fiber, $C_\lambda= \pi^{-1}(\lambda)$, is smooth. But at the three boundary points of $\overline{\mathcal{M}}_{0,4}$, corresponding to $\lambda=0,1,\infty$, $C_\lambda$ degenerates into a pair of lines

\vbox{
\begin{displaymath}
C_0=\{y(z-x) = 0\},\quad\text{(with the node at}\,n_0=(1,0,1)\,)
\end{displaymath}
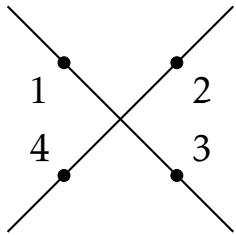
\begin{center}
\begin{tikzpicture}
\begin{scope}[scale=.75]
\draw[thick,-] (0,0) -- (4,-4);
\draw[thick,-] (4,0) -- (0,-4);
\filldraw (1,-1) circle (3pt) node[anchor=north east, scale=1.25] {$1$};
\filldraw (3,-3) circle (3pt) node[anchor=south west, scale=1.25] {$3$};
\filldraw (3,-1) circle (3pt) node[anchor=north west, scale=1.25] {$2$};
\filldraw (1,-3) circle (3pt) node[anchor=south east, scale=1.25] {$4$};
\end{scope}
\end{tikzpicture}
\end{center}
}

\vbox{
\begin{displaymath}
C_1=\{z(x-y) = 0\},\quad\text{(with the node at}\,n_1=(1,1,0)\,)
\end{displaymath}
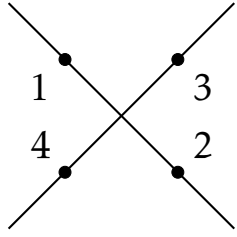
\begin{center}
\begin{tikzpicture}
\begin{scope}[scale=.75]
\draw[thick,-] (0,0) -- (4,-4);
\draw[thick,-] (4,0) -- (0,-4);
\filldraw (1,-1) circle (3pt) node[anchor=north east, scale=1.25] {$1$};
\filldraw (3,-3) circle (3pt) node[anchor=south west, scale=1.25] {$2$};
\filldraw (3,-1) circle (3pt) node[anchor=north west, scale=1.25] {$3$};
\filldraw (1,-3) circle (3pt) node[anchor=south east, scale=1.25] {$4$};
\end{scope}
\end{tikzpicture}
\end{center}
}

\vbox{
\begin{displaymath}
C_{\infty}=\{x(y-z) =0\},\quad\text{(with the node at}\,n_{\infty} = (0,1,1)\,)
\end{displaymath}
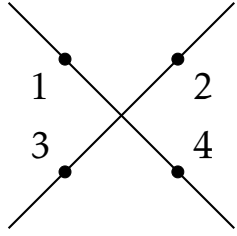
\begin{center}
\begin{tikzpicture}
\begin{scope}[scale=.75]
\draw[thick,-] (0,0) -- (4,-4);
\draw[thick,-] (4,0) -- (0,-4);
\filldraw (1,-1) circle (3pt) node[anchor=north east, scale=1.25] {$1$};
\filldraw (3,-3) circle (3pt) node[anchor=south west, scale=1.25] {$4$};
\filldraw (3,-1) circle (3pt) node[anchor=north west, scale=1.25] {$2$};
\filldraw (1,-3) circle (3pt) node[anchor=south east, scale=1.25] {$3$};
\end{scope}
\end{tikzpicture}
\end{center}
}

\subsection{The bundle of Hitchin bases}\label{32_the_bundle_of_hitchin_systems}

Pick a collection of 4 nilpotent orbits, $O_i$ in $\mathfrak{sl}(N)$. We will interchangeably consider two models for the spectral curve $\Sigma_\lambda \to C_\lambda$. In both cases, it is the vanishing locus of a homogeneous polynomial

\begin{displaymath}
0 = \det(w\mathds{1} -\Phi) = w^N-\sum_{k=2}^N\phi_k w^{N-k}
\end{displaymath}
in the total space of a line bundle $L\to C$. One model is to take $L=K_C$ and allow the $\phi_k$ to have poles of order $\pi^{(i)}_k=(k-\chi_k^{(i)})$ (dictated by the choice of $O_i$) at the punctures. In the second model, we take $L=K_C\bigl(\sum E_i\bigr)$ and demand that the $\phi_k$ have zeroes of order $\chi_k^{(i)}$ at the punctures. The latter model is more convenient for our global discussion, as it naturally produces $\Sigma$ as a compact curve.

Let

\begin{displaymath}
\mathcal{L}_k = \mathcal{O}(k)\bigl(-\sum_i\chi_k^{(i)} E_i\bigr)
\end{displaymath}
On each curve, $C_\lambda$, the $\phi_k$ are holomorphic sections of $\mathcal{L}_k\vert_{C_\lambda}$. These fit together to form

\begin{displaymath}
\phi_k\in H^0(\overline{\mathcal{C}},\mathcal{L}_k)= H^0(\overline{\mathcal{M}}_{0,4},\pi_* \mathcal{L}_k)
\end{displaymath}
We now proceed to compute the direct image sheaves, $\pi_*\mathcal{L}_k$ on $\overline{\mathcal{M}}_{0,4}$.

\subsection{Computing the direct image bundles}\label{directimages}

We are interested in various line bundles, $\mathcal{L}=\mathcal{O}(k)\left(-\sum_i n_i E_i\right)$ on $\widetilde{\mathbb{CP}}^2$ and their direct-image sheaves $\pi_* \mathcal{L}$ on $\overline{\mathcal{M}}_{0,4}$. A-priori, the direct image is torsion-free and hence (since we are in complex dimension-1) a vector bundle, $V$. The fiber of $V$ over  $\lambda\in \mathcal{M}_{0,4}$ is $H^0\left(C_\lambda, \mathcal{L}\right)$. Over the boundary points, the dimension of $H^0\left(C_{\lambda=0,1,\infty}, \mathcal{L}\right)$ can sometimes jump. If it does then the fiber is $H^0\left(C_{\lambda=0,1,\infty}, \mathcal{L}'\right)\subset H^0\left(C_{\lambda=0,1,\infty}, \mathcal{L}\right)$ for $\mathcal{L}'$ defined in \eqref{substituteL}.

Any vector bundle on $\mathbb{CP}^1$ splits as a direct sum of line bundles. So we have

\begin{equation}
\label{directimage}
\pi_* \mathcal{O}(k)\Bigl(-\sum_i n_i E_i\Bigr)= \sum_{i\in\mathbb{Z}} m_i \mathcal{O}_{\mathbb{P}^1}(i)
\end{equation}
for some collection of $m_i\geq 0$. We thus get one relation,

\begin{displaymath}
\sum_i m_i = h^0\left(C_\lambda, \mathcal{L}\right)
\end{displaymath}
among this infinite number of unknowns. To find more relations (and, ultimately, to solve for the $m_i$), the trick is to tensor \eqref{directimage} with $\mathcal{O}_{\mathbb{P}^1}(-l)$.

\begin{equation}
\pi_* \Bigl(\mathcal{O}(k)(-\sum_i n_i E_i)\otimes \pi^*(\mathcal{O}_{\mathbb{P}^1}(-l))\Bigr)= \sum_{i\in\mathbb{Z}} m_i \mathcal{O}_{\mathbb{P}^1}(i-l)
\label{directimagetwisted}\end{equation}
and use

\begin{equation}
\pi^*\left(\mathcal{O}_{\mathbb{P}^1}(1)\right) = \mathcal{O}(2)\bigl(-\sum_i E_i\bigr)
\label{pullback}\end{equation}
Putting \eqref{directimagetwisted} and \eqref{pullback} together, we have

\begin{equation}
\pi_* \Bigl(\mathcal{O}(k-2l)(-\sum_i (n_i-l) E_i)\Bigr)= \sum_{i\in\mathbb{Z}} m_i \mathcal{O}_{\mathbb{P}^1}(i-l)
\label{directimagebetter}\end{equation}
for each $l\in\mathbb{Z}$. Taking $H^0$ of both sides and using that, for any $f\colon X\to Y$ and $\mathcal{F}$ a sheaf on $X$, $H^0(Y,f_*\mathcal{F})=H^0(X,\mathcal{F})$, we get for each $l$ a relation on the $m_i$.

For $l\geq n_i$ we have, by Hartog's Theorem,

\begin{equation}
h^0\Bigl(\widetilde{\mathbb{CP}}^2, \mathcal{O}(k-2l)\bigl(\sum_i(l-n_i)E_i\bigr)\Bigr)=\begin{cases}\left(\begin{smallmatrix}k-2l+2\\2\end{smallmatrix}\right)&k\geq2l\\0&\text{otherwise}\end{cases}
\label{h0first}\end{equation}
When the $n_i=1$, the conditions imposed by demanding that the sections vanish to some order at the $E_i$ are independent. In that case, we can generalize \eqref{h0first} to

\begin{equation}
h^0\Bigl(\widetilde{\mathbb{CP}}^2, \mathcal{O}(k-2l)\bigl(\sum_i(l-1)E_i\bigr)\Bigr)= \max\left(\left(\begin{smallmatrix}k-2l+2\\2\end{smallmatrix}\right),0\right)-4\max\left(\tfrac{1}{2}(1-l)(2-l),0\right)
\label{h0second}\end{equation}
More generally, if the $n_i$ are ``small enough'' so that the vanishing constraints at the $E_i$ are independent, we have

\begin{equation}
h^0\Bigl(\widetilde{\mathbb{CP}}^2, \mathcal{O}(k-2l)\bigl(\sum_i(l-n_i)E_i\bigr)\Bigr)= \max\left(\left(\begin{smallmatrix}k-2l+2\\2\end{smallmatrix}\right),0\right)-\tfrac{1}{2}\sum_i(n_i-l+1)\max\left(n_i-l,0\right)
\label{h0third}\end{equation}
The requisite condition\footnote{The actual condition is $(n_i-l)+(n_j-l) \leq (k-2l)+1$, but the $l$`s cancel, yielding \eqref{smallenough}.}  is

\begin{equation}
\begin{aligned}
\sum_i n_i &\leq 2k+1\\
n_i+n_j &\leq k+1, \forall\, \text{pairs}\, i,j
\end{aligned}
\label{smallenough}\end{equation}
When this holds, \eqref{directimagebetter} yields

\begin{equation}
\sum_i m_i \max(0, i-l+1) = \max\left(\left(\begin{smallmatrix}k-2l+2\\2\end{smallmatrix}\right),0\right)-\tfrac{1}{2}\sum_i(n_i-l+1)\max\left(n_i-l,0\right)
\label{summi}\end{equation}
for all $l\in\mathbb{Z}$.

When there's a pair $n_i,n_j$ which violates \eqref{smallenough} --- say $n_i+n_j = k+1+p$, for some $p\gt 0$ --- then we replace

\begin{equation}
\mathcal{L}\to \mathcal{L}'= \mathcal{L} \otimes \bigl({\mathcal{O}(-1)(E_i+E_j)}\bigr)^{\otimes p}
\label{substituteL}\end{equation}
This preserves $h^0(\widetilde{\mathbb{CP}}^2,\mathcal{L}')=h^0(\widetilde{\mathbb{CP}}^2,\mathcal{L})$ (and it preserves the inequality for the other pairs) while making $n_i+n_j$ ``small enough.'' Computing $\pi_*\mathcal{L}'$ produces a vector bundle on $\overline{\mathcal{M}}_{0,4}$ of the same rank and the same first Chern class as $\pi_*\mathcal{L}$. Moreover, since $\mathcal{L}'$ is a subsheaf of $\mathcal{L}$, $\pi_*\mathcal{L}'$ is a subsheaf of $\pi_*\mathcal{L}$. Since they are vector bundles of the same rank and first Chern class on $\overline{\mathcal{M}}_{0,4}=\mathbb{CP}^1$, and the former is a subbundle of the latter, they are isomorphic.

To see how this works, let's specialize to setting all the $n_i=1$. Then \eqref{summi} becomes

\begin{equation}
m_l+2m_{l+1}+3m_{l+1}+\dots =\begin{cases}
\tfrac{1}{2}\left(k^2+(3-4l)k+6(l-1)\right)&l\leq0\\
\tfrac{1}{2}(k-2l+2)(k-2l+1)&0\lt l\leq \tfrac{k}{2}\\
0&l\gt \tfrac{k}{2}
\end{cases}
\label{summispecial}\end{equation}
To solve this system of equations, there are two cases
\begin{itemize}%
\item $k=2p=\text{even}$
\begin{itemize}%
\item We have $m_l=0$ for $l\geq p+1$. Then clearly, $m_p=1$ and hence $4=m_{p-1}=m_{p-2}=\dots =m_1$. Then $m_0=0$ and $m_l=0$, for $l\lt 0$.
\end{itemize}
\item $k=2p+1=\text{odd}$
\begin{itemize}%
\item We have $m_l=0$ for $l\geq p+1$. Then clearly, $m_p=3$ and hence $4=m_{p-1}=m_{p-2}=\dots =m_1$. Then $m_0=0$ and $m_l=0$, for $l\lt 0$.
\end{itemize}
\end{itemize}
To summarize:
\begin{equation}\label{pushsummary}
\pi_* \mathcal{O}(k)\bigl(-\sum_i E_i\bigr) =
\begin{cases}[1.5]
4\mathcal{O}(1)\oplus 4\mathcal{O}(2)\oplus\dots\oplus4\mathcal{O}(p-1)\oplus\mathcal{O}(p)&\text{for}\; k=2p\\
4\mathcal{O}(1)\oplus 4\mathcal{O}(2)\oplus\dots\oplus4\mathcal{O}(p-1)\oplus3\mathcal{O}(p)&\text{for}\; k=2p+1\\
\end{cases}
\end{equation}
Extending this to $1\leq n_i\leq k-1$ (subject to $\sum_i n_i\leq 2k+1$) will be useful in the following, so let us tabulate the results. We mark in $\color{red}\text{red}$ the cases where we had to apply \eqref{substituteL}.

For $k=2$, there's only one case, $(n_1,n_2,n_3,n_4)=(1,1,1,1)\Rightarrow m_1=1$. For $k=3,4,5$, the results are summarized in Table \ref{directimagetable}.
\begin{table}[!h]
\subcaptionbox*{$k=3$}{
\begin{tabular}{|c|c|}
\hline
$(n_1,n_2,n_3,n_4)$&$m_1$\\
\hline
$(1, 1, 1, 1)$&3\\
$(2, 1, 1, 1)$&2\\
$(2, 2, 1, 1)$&1\\
$(2, 2, 2, 1)$&0\\
\hline
\end{tabular}}
\subcaptionbox*{$k=4$}{
\begin{tabular}{|c|c|c|}
\hline
$(n_1,n_2,n_3,n_4)$&$m_2$&$m_1$\\
\hline
$(1, 1, 1, 1)$&1&4\\
$(2, 1, 1, 1)$&1&3\\
$(2, 2, 1, 1)$&1&2\\
$(2, 2, 2, 1)$&1&1\\
$(2, 2, 2, 2)$&1&0\\
$(3, 1, 1, 1)$&0&3\\
$(3, 2, 1, 1)$&0&2\\
$(3, 2, 2, 1)$&0&1\\
$(3, 2, 2, 2)$&0&0\\
$\color{red}(3, 3, 1, 1)$&0&1\\
$\color{red}(3, 3, 2, 1)$&0&0\\
\hline
\end{tabular}}
\subcaptionbox*{$k=5$}{
\begin{tabular}{|c|c|c|}
\hline
$(n_1,n_2,n_3,n_4)$&$m_2$&$m_1$\\
\hline
$(1, 1, 1, 1)$&3&4\\
$(2, 1, 1, 1)$&3&3\\
$(2, 2, 1, 1)$&3&2\\
$(2, 2, 2, 1)$&3&1\\
$(2, 2, 2, 2)$&3&0\\
$(3, 1, 1, 1)$&2&3\\
$(3, 2, 1, 1)$&2&2\\
$(3, 2, 2, 1)$&2&1\\
$(3, 2, 2, 2)$&2&0\\
$(3, 3, 1, 1)$&1&2\\
$(3, 3, 2, 1)$&1&1\\
$(3, 3, 2, 2)$&1&0\\
$(3, 3, 3, 1)$&0&1\\
$(3, 3, 3, 2)$&0&0\\
$(4, 1, 1, 1)$&0&4\\
$(4, 2, 1, 1)$&0&3\\
$(4, 2, 2, 1)$&0&2\\
$(4, 2, 2, 2)$&0&1\\
$\color{red}(4, 3, 1, 1)$&0&2\\
$\color{red}(4, 3, 2, 1)$&0&1\\
$\color{red}(4, 3, 2, 2)$&0&0\\
$\color{red}(4, 3, 3, 1)$&0&0\\
$\color{red}(4, 4, 1, 1)$&0&1\\
$\color{red}(4, 4, 2, 1)$&0&0\\
\hline
\end{tabular}}
\caption{The values of $m_i$ for $k=3,\, k=4,\, k=5$}\label{directimagetable}
\end{table}

%

\section{Global Interlude II : Standard and Restricted Nodes}\label{4_standard_and_restricted_nodes}

In this section, we will use the global model developed in \S\ref{Hitchinglobal} to study the kinds of nodes that can arise for a tame Hitchin system on a four punctured sphere. To illustrate the main points, we will pick the Hitchin system for $\mathfrak{j}=\mathfrak{sl}(4)$.

\subsection{The standard node}\label{41_the_standard_node}

We are interested in the behaviour of the spectral curve when the base curve, $C$ develops a node. There is a ``generic'' behaviour that we will call ``the standard node.'' This is when the number of Casimirs (or ``center parameters'', in the nomenclature to be introduced below) on the Hitchin base is equal to the rank of $\mathfrak{j}$ ($N-1$ for $\mathfrak{sl}(N)$).

When the Hitchin orbits at the punctures are sufficiently big, when the number of punctures is sufficiently large or if the genus of each component of the nodal curve is $\geq 1$, then every node is a standard node. Similarly, in the $A_1$ theory, all nodes are standard.

On the 4-punctured sphere, we can ensure that we get a standard node by taking the residue of the Higgs field to lie in the regular nilpotent orbit, $\text{Res}(\Phi)=a \in [N]$ at each of the punctures. This corresponds to requiring that each $\phi_k$ has a simple zero at each of the four punctures $E_i,i=1,2,3,4$.

We will be particularly interested in the behaviour of the family of spectral curves (or equivalently, of Hitchin bases $\mathcal{B}$) as we approach a boundary of $\mathcal{M}_{g,n}$. In $\overline{\mathcal{M}}_{0,4}$, the three boundary points look similar, so let us focus on one of them: the $\lambda = 1$ boundary. For simplicity, we will specialize to $\mathfrak{sl}(4)$. The generalization to arbitrary $\mathfrak{sl}(N)$ is straightforward.

We list the contributions to $dim(\mathcal{B})$ from each degree. Each of the $\phi_k$ have the form of certain homogeneous polynomials of degree $k$ in $x,y,z$. That is, $\phi_k$ is a holomorphic section of the line bundle $\mathcal{L}_k=\mathcal{O}(k)(-\sum_i E_i)$ on $\widetilde{\mathbb{CP}}^2$. We computed the direct images $\pi_*\mathcal{L}_k$ in \S\ref{directimages}. For each $k$, the results are summarized in the first line of the corresponding sub-table of table \ref{directimagetable} or equivalently in \eqref{pushsummary}.

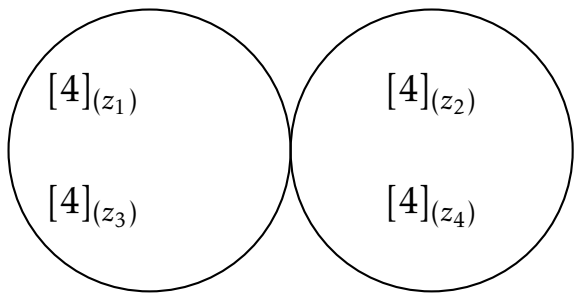
\begin{center}
\begin{tikzpicture}
\begin{scope}[scale=.75]
\node[scale=1.2] at (-1,1){$[4]_{(z_1)}$};
\node[scale=1.2] at (-1,-1){$[4]_{(z_3)}$};
\node[scale=1.2] at (5,1){$[4]_{(z_2)}$};
\node[scale=1.2] at (5,-1){$[4]_{(z_4)}$};
\draw[thick] (0,0.0) circle [radius=2.5];
\draw[thick] (5,0.0) circle [radius=2.5];
\end{scope}
\end{tikzpicture}
\end{center}

\medskip
\noindent\underline{$\pi_*\mathcal{L}_2 = \mathcal{O}_{\mathbb{P}^1}(1)$.} So the space of $\phi_2$s is 2-dimensional. We can think of it as being spanned by

\begin{displaymath}
C_0 = y(z-x),\, C_1 = z(x-y),\, C_\infty= x(y-z)
\end{displaymath}
subject to the relation

\begin{equation}
C_0+C_1+C_\infty=0
\label{crel}\end{equation}
Restricted to any given $C_\lambda$, there's an additional relation \eqref{universalcurve},

\begin{equation}
C_\infty \lambda_1 + C_0 \lambda_2 =0
\label{ccrel}\end{equation}
which means that, restricted to $C_\lambda$, the space of $\phi_2$s is 1-dimensional. But notice that $\mathcal{O}_{\mathbb{P}^1}(1)$ is nontrivial. Any global holomorphic section has a zero for some $\lambda\in \mathbb{P}^1$.

We will choose a trivialization which is good everywhere except at $\lambda=\infty$ and write

\begin{displaymath}
\phi_2 = u_{2;C} x(y-z)
\end{displaymath}
Near $\lambda=1$, this $\phi_2$ does not vanish either on the right component (the line $x-y=0$) nor on the left component (the line $z=0$) and thus belongs to the center (equivalently the node itself). Hence the ``C'' subscript.

\medskip
\noindent\underline{$\pi_*\mathcal{L}_3 = 3\mathcal{O}_{\mathbb{P}^1}(1)$.} So the space of $\phi_3$s is 6-dimensional. We can view this as being spanned by

\begin{displaymath}
C_0 x,\, C_0 y,\, C_0 z,
  C_1 x,\, C_1 y,\, C_1 z,
  C_\infty x,\, C_\infty y,\, C_\infty z,
\end{displaymath}
subject to the relation \eqref{crel}. Restricting to any given $C_\lambda$, we get the additional relation \eqref{ccrel}, which cuts the dimension of the space of $\phi_3$s down to 3. Again, we'll choose a trivialization of $\pi_*\mathcal{L}_3$ which is good everywhere but at $\lambda=\infty$ and write

\begin{displaymath}
\phi_3 = x(y-z)(u_{3;1}x + u_{3;2}y + u_{3;3}z)
\end{displaymath}
At $\lambda=1$, $u_{3;L}=u_{3;1}$ is supported only on the left component of the curve (the line $z=0$), $u_{3;R}=u_{3;3}$ is supported only on the right component of the curve (the line $x-y=0$), and $u_{3;C}=u_{3;1}+u_{3;2}$ is supported on both.

\medskip
\noindent\underline{$\pi_*\mathcal{L}_4 = 4\mathcal{O}_{\mathbb{P}^1}(1)\oplus \mathcal{O}_{\mathbb{P}^1}(2)$.} So the space of $\phi_4$s is $4\times 2+3=11$-dimensional. Restricting to any given $C_\lambda$ reduces the dimension to 5. For any $\lambda\neq\infty$, we can take this 5-dimensional space to be spanned by

\begin{displaymath}
\phi_4 = x(y-z)[u_{4;L,L}(x-y)(z-x)+u_{4;L,R}(x-y)y +u_{4;R,L}z(z-x) + u_{4;R,R}z y + u_{4;C}x(y-z)]
\end{displaymath}
up to terms which vanish by \eqref{universalcurve}. Here, the first $L(R)$ subscript pertains to sections supported on the left(right) component of the nodal curve at $\lambda=1$ and the second $L(R)$ subscript pertains to sections supported on the left(right) component of the nodal curve at $\lambda=0$. Moreover, in this parametrization, $u_{4;C}$ is the parameter which transforms as a section of $\mathcal{O}_{\mathbb{P}^1}(2)$ (with a double pole at $\lambda=\infty$), whereas the other four parameters transform as sections of $\mathcal{O}_{\mathbb{P}^1}(1)$.

Taken together, we have the following family of spectral curves

\begin{equation}\label{master_curve}
\begin{split}
Det(\Phi - w\mathds{1}) &= w^4 - x(y-z)\bigl[u_{2;C} w^2+ (u_{3;1}x + u_{3;2}y + u_{3;3}z)w \\
&\qquad+ (u_{4;L,L}(x-y)(z-x)+u_{4;L,R}(x-y)y +u_{4;R,L}z(z-x) + u_{4;R,R}z y\\
&\qquad+ u_{4;C}x(y-z))\bigr]\\ &= 0 
\end{split}
\end{equation}
And we have following graded base dimensions $b_k^{L,C,R}, k = 2,3,4$ :
\begin{equation}
\begin{split}
b_k^L &= \{0,1,2 \} \\
b_k^C &= \{1,1,1 \} \\
b_k^R &= \{0,1,2 \}
\end{split}
\end{equation}
We have trivialized the bundle of Hitchin bases on the complement of $\lambda=\infty$. If we want a description that extends to $\lambda=\infty$, we could choose (say) a trivialization which was good everywhere except at $\lambda=0$ and set $u'_{2;C} = C_0$ (and similarly for the rest of the $u$'s). Clearly, these are related by

\begin{displaymath}
u'_{2;C} = - \tfrac{1}{\lambda} u_{2;C}
\end{displaymath}
That is, the bundle of Hitchin bases is nontrivial over $\overline{\mathcal{M}}_{0,4}$. It splits as a direct sum of line bundles, as we computed in \S\ref{directimages}. Reading off the results from table \ref{directimagetable} or \eqref{pushsummary},
\begin{equation}
\label{globalallregular}
\mathcal{B}=\overset{k=2}{\overbrace{\mathcal{O}(1)}}\oplus \overset{k=3}{\overbrace{3\mathcal{O}(1)}}\oplus\overset{k=4}{\overbrace{4\mathcal{O}(1)+\mathcal{O}(2)}}=8\mathcal{O}(1)\oplus\mathcal{O}(2)\quad.
\end{equation}

\begin{remark}
More generally, for $\mathfrak{j}=\mathfrak{sl}(N)$, and 4 regular nilpotents on $C_{0,4}$, 
\begin{equation}\label{finerstratification}
\mathcal{B}_k =\begin{cases}[3]
\mathcal{O}(l)\oplus{\displaystyle \bigoplus_{m=1}^{l-1}} \mathcal{O}(m)^{\oplus4}&\text{for}\; k=2l\\
\mathcal{O}(l)^{\oplus3}\oplus{\displaystyle \bigoplus_{m=1}^{l-1}} \mathcal{O}(m)^{\oplus4}&\text{for}\; k=2l+1
\end{cases}
\end{equation}
This gives a new stratification of the Hitchin base, finer than the decomposition into $\mathcal{B}_k$, even for smooth C. We do not understand its mathematical significance. From a physical perspective, this stratification gives the transition functions needed to relate the Coulomb branch parameters in different S-duality frames. The precise physical significance of the transition functions implied by \eqref{finerstratification} remains to be explored.
\end{remark}

When $C$ approaches the nodal limit, the Lagrangian fibers of the Hitchin map $\mu$ acquire certain non-compact directions. The non-compact directions in the fiber are symplectic dual \footnote{Let $u_i,\theta_i$ be a system of coordinates on the base and fibers of the Hitchin integrable system such that $\Omega_I = \sum_i du_i \wedge d\theta_i$. In the nodal limit, some of the $u_i$ correspond to the center parameters $u_{i;C}$. Their (symplectically) dual directions parameterized by $\theta_{i;C}$ are the non-compact directions.} to the center parameters. If we quotient out the non-compact directions, we are then left with a Poisson integrable system in which the center parameters act as Casimir parameters. So, the label ``C'' in $u_{k;C}$ could equally-well stand for Casimir. 

For generic values of the center parameters, $\Phi$ has a simple pole with semisimple residue at the node (for $\lambda=1$, this is the point $x=y=1,\, z=0$). The fiber over the node on $C$, consists of $N$ nodes of the spectral curve, which is otherwise smooth.

If we focus our attention  on just the right component of $C$ (the line $x-y=0$), we get the following spectral curve

\begin{displaymath}
w^4 - x(x-z) [u_{2;C}w^2 + (u_{3;C}x + u_{3;R}z)w + (u_{4;C}x(x-z) + u_{4;R,R}x z + u_{4;R,L} z(z-x) ) ] = 0
\end{displaymath}
where, as above, $u_{3;C}=u_{3;1}+u_{3;2}$, $u_{3;R}=u_{3;3}$ and $u_{3;L}=u_{3;1}$. On the left component (the line $z=0$)
\begin{displaymath}
w^4 - x y\bigl[u_{2;C} w^2+ (u_{3;L}(x-y) + u_{3;C}y)w
+ ( u_{4;C}x y-u_{4;L,L}x(x-y)+u_{4;L,R}(x-y)y 
)\bigr]= 0 
\end{displaymath}

Setting the center parameters to zero (and freezing the corresponding non-compact fiber directions), we obtain a symplectic integrable subsystem. Since the normalization of the nodal curve is disconnected, the integrable system is the product of an integrable system associated to the 3-punctured sphere on the left with an integrable system associated to the 3-punctured sphere on the right. On $C_L$ the spectral curve for the symplectic integrable subsystem is
\begin{equation}
w^4 - x y(x-y)\bigl[u_{3;L}w
+ (-u_{4;L,L}x+u_{4;L,R}y 
)\bigr]= 0 
\end{equation}
and on $C_R$ we obtain
\begin{equation}
w^4 - x z (x-z) [u_{3;R}w + (u_{4;R,R}x - u_{4;R,L} (x-z) ) ] = 0\quad.
\end{equation}
Each of these is the spectral curve for the Hitchin integrable system associated to $C_{0,3}$ with 3 regular Hitchin punctures.

The physics of this degeneration is well-understood: the theory contains an $SU(4)$ $\mathcal{N}=2$ vector multiplet which becomes weakly coupled as we approach the nodal limit. The center parameters are the VEVs of (gauge-invariant polynomials in) the scalar fields in the vector multiplet. These are in 1-1 correspondence with the independent Casimirs of $SU(4)$. This vector multiplet gauges a diagonal $SU(4)$ subgroup of the product of the two SCFTs (associated to the 3-punctured spheres) which are called $T_N$ (for $N=4$) in \cite{Gaiotto:2009we}.

\subsection{Restricted nodes}\label{first_restricted_nodes}

\begin{table}
\begin{center}
\begin{tabular}{|c|c|}
\hline 
Hitchin orbit & Zero orders $\vec{\chi}$ \\ 
\hline 
$[4]$ & $(1,1,1)$ \\ 
\hline 
$[3,1]$ & $(1,1,2)$ \\ 
\hline 
$[2,2]$ & $(1,2,2)$ \\ 
\hline 
$[2,1^2]$ & $(1,2,3)$ \\ \hline 
\end{tabular} 
\end{center}
\caption{The zero orders $\vec{\chi}$ for the non-zero nilpotent orbits in $\mathfrak{sl}(4)$}\label{zeroorderssl4}
\end{table}

So far, we have assumed four punctures with residues in the regular Hitchin nilpotent orbit. If we were to choose the residues to be in some smaller nilpotent orbit, then the zero orders would go up.  For example, if we choose the residue at the puncture $E_1$ to be in the Hitchin orbit $[2^2]$, this forces $\phi_3$ and $\phi_4$ to have double zeroes (instead of simple zeroes) at $E_1$. That is, it changes the vector $\vec{\chi}^1$ from $(1,1,1)$ to $(1,2,2)$ where the entries of $\vec{\chi}$ correspond to $k=2,3,4$. This, in turn, imposes linear relations among the coefficients. These relations can be deduced by looking at the \eqref{master_curve}. We see that the only terms in $\phi_3$ and $\phi_4$ that don't have a double zero at $E_1$ (which is the locus $y=0,z=0$) are the terms with coefficients $u_{3;1}$ and $u_{4;L,L}$. If we set these coefficients to zero, we force the residue at $E_1$ to live in the nilpotent orbit $[2^2]$. One can deduce similar constraints for all other nilpotent orbits and the locations $E_i$. For the reader's convenience, we tabulate the zero orders $\vec{\chi}$ for the various nilpotent orbits in $\mathfrak{sl}(4)$ in table \ref{zeroorderssl4}. When we have three regular nilpotents and one non-regular nilpotent inserted at $E_i$, the constraints obtained in this way are summarized in table \ref{constraint_table}.

\begin{table}
\begin{center}
\begin{tabular}{|c|c|c|c|c|}
\hline
$O_H$&$E_1$&$E_2$&$E_3$&$E_4$\\
\hline
$[3,1]$&$u_{4;L,L}=0$&$u_{4;L,R}=0$&$u_{4;R,L}=0$&$u_{4;R,R}=0$\\ \hline
$[2^2]$&$\begin{gathered}u_{3;1}=0\\u_{4;L,L}=0\end{gathered}$&$\begin{gathered}u_{3;2}=0\\u_{4;L,R}=0\end{gathered}$&$\begin{gathered}u_{3;3}=0\\u_{4;R,L}=0\end{gathered}$&$\begin{gathered}u_{3;1}+u_{3;2}+u_{3;3}=0\\u_{4;R,R}=0\end{gathered}$\\ \hline 
$[2,1^2]$&$\begin{gathered}u_{3;1}=0\\u_{4;L,L}=0\\\begin{aligned}(&u_{4;R,L}+u_{4;C})\\+& \lambda(u_{4;L,R}-u_{4;R,L})\\&=0\end{aligned}\end{gathered}$&$\begin{gathered}u_{3;2}=0\\u_{4;L,R}=0\\\begin{aligned}(&u_{4;R,R}+u_{4;C})\\+&\lambda(u_{4;L,L}-u_{4;R,R})\\&=0\end{aligned}\end{gathered}$&$\begin{gathered}u_{3;3}=0\\u_{4;R,L}=0\\\begin{aligned}(&u_{4;L,L}-u_{4;C})\\+&\lambda(u_{4;R,R}-u_{4;L,L})\\&=0\end{aligned}\end{gathered}$&$\begin{gathered}u_{3;1}+u_{3;2}+u_{3;3}=0\\u_{4;R,R}=0\\\begin{aligned}(&u_{4;L,R}-u_{4;C})\\+&\lambda(u_{4;R,L}-u_{4;L,R})\\&=0\end{aligned}\end{gathered}$\\
\hline
\end{tabular}
\end{center}
\caption{Conditions imposed on Coulomb branch parameters in the $\mathfrak{sl}(4)$ Hitchin system on $C_{0,4}$ with three regular nilpotent and one non-regular nilpotent residue}\label{constraint_table}
\end{table}

Replacing one of the residues with a non-regular nilpotent changes the bundle of Hitchin bases \eqref{globalallregular}. For example, if one of the residues is in the orbit $[2,1^2]$ and the other three remain regular, the space of $\phi_4$s is the kernel of a map
\begin{displaymath}
4\mathcal{O}_{\mathbb{P}^1}(1)\oplus \mathcal{O}_{\mathbb{P}^1}(2) \to \mathcal{O}_{\mathbb{P}^1}(1)\oplus \mathcal{O}_{\mathbb{P}^1}(2)
\end{displaymath}
where the image of $u_{4;C}$ is nonzero. Hence the kernel is isomorphic to $3\mathcal{O}_{\mathbb{P}^1}(1)$, as follows from the analysis of \S\ref{directimages} (see the entry for $(n_1,n_2,n_3,n_4)=(3,1,1,1)$ in the $k=4$ subtable of table \ref{directimagetable}). Similarly, the space of $\phi_3$s is $2\mathcal{O}_{\mathbb{P}^1}(1)$. Assembling all the pieces together, we get 
\begin{equation}
\label{globalalloneminimal}
\mathcal{B}=\overset{k=2}{\overbrace{\mathcal{O}(1)}}\oplus \overset{k=3}{\overbrace{2\mathcal{O}(1)}}\oplus\overset{k=4}{\overbrace{3\mathcal{O}(1)}}=6\mathcal{O}(1)\quad.
\end{equation}
In replacing one of the regular nilpotents by $[2,1^2]$,  we had to impose three linear constraints from  table \ref{constraint_table}. This reduced the dimension of the Hitchin base, $\mathcal{B}$, from 9 to 6. And indeed \eqref{globalalloneminimal} is a rank-6 sub-bundle of \eqref{globalallregular}.

Proceeding as we did in the case with four regular nilpotents, we would now like to study the behaviour of the spectral curves when the base curve $C$ develops a node while allowing for some of the residues to be non-regular. We will see that the specialization of the spectral curve, implied by imposing the constraints of table \ref{constraint_table}, \emph{changes} its behaviour when the base curve, $C$, degenerates.

The first type of change is a reduction in the number of center parameters (that is, the residue of $\Phi$ at the node is no longer a generic semisimple) - rather than forming the Casimirs of $SU(N)$, it will turn out that they form the Casimirs of some simple subgroup $H\subset SU(N)$. We will give a mathematical proof of this claim in \S\ref{classifying}. This is also guaranteed by certain physics considerations which we recall in \S\ref{flavour}.

A second type of change that could occur is that when the center parameters are set to zero, the $Res(\Phi)$ at the node need not be in the regular nilpotent orbit. When we set the center parameters to zero in the standard node \eqref{master_curve}, the coefficients of $w^{N-k}$ for $k=3,\dots,N$ vanish linearly at the node ($x-y=z=0$). So, we conclude that the orbit $O$ at the node is $[N]$, the regular nilpotent. For $O \neq [N]$, some of these coefficients vanish to higher order. When the orbit at the node is non-regular, the center parameters then live in the closure of the sheet\footnote{We refer the reader to \cite{Balasubramanian:2018pbp} for an introduction to sheets in complex Lie algebras and further background references.} that contains the orbit $O$ at its boundary.

In \S\ref{classifying}, we prove that the vanishing orders uniquely determine such a $O$ and also characterize the nilpotent orbits that could occur in this way. To capture these two phenomena, we will label a restricted node by the pair $(O,H)$. Even though our proofs appear in \S\ref{classifying}, to simplify the presentation, we have adopted the notation $(O,H)$ to label restricted nodes through out the paper. In this notation, the standard node would be $([N],SU(N))$.

\subsection{Examples}\label{examples}

Let us illustrate the above considerations with some examples. For brevity, we'll focus on the behaviour near the $\lambda=1$ degeneration of $C$.

\hypertarget{example_1}{}\subsubsection*{Example 1}
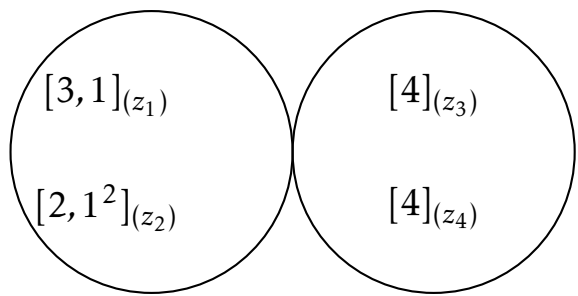
\label{example_1}

\begin{center}
\begin{tikzpicture}
\begin{scope}[scale=.75]
\node[scale=1.2] at (-.8,1){$[3,1]_{(z_1)}$};
\node[scale=1.2] at (-.8,-1){$[2,1^2]_{(z_2)}$};
\node[scale=1.2] at (5,1){$[4]_{(z_3)}$};
\node[scale=1.2] at (5,-1){$[4]_{(z_4)}$};
\draw[thick] (0,0.0) circle [radius=2.5];
\draw[thick] (5,0.0) circle [radius=2.5];
\end{scope}
\end{tikzpicture}
\end{center}

The orders of the  zeroes of the $\phi_k$ for this example are:

\begin{center}
\begin{tabular}{l|l|l|l|l}
$E_i$&$O_H$&$\phi_2$&$\phi_3$&$\phi_4$\\
\hline 
$E_1$&$[3,1]$&1&1&2\\
$E_2$&$[2,1^2]$&1&2&3\\
$E_3$&$[4]$&1&1&1\\
$E_4$&$[4]$&1&1&1\\
\end{tabular}
\end{center}
Using the results from table \ref{constraint_table}, for placing $[3,1]$ at $E_1$ and $[2,1^2]$ at $E_2$, we get the constraints
$$
u_{3;2}=0,\quad u_{4;L,L}=0,\quad u_{4;L,R}=0,\quad u_{4;C}= (\lambda-1) u_{4;R,R}
$$
Plugging these into \eqref{master_curve}, we get the following family of spectral curves in this example (dropping a term proportional to $u_{4;R,R}C_\lambda$):

\begin{equation}
w^4 - x(y-z)[u_{2;C} w^2 + (u_{3;C}x + u_{3;R}z)w + (u_{4;R,L}z(z-x)+ u_{4;R,R}x z)] = 0
\end{equation}
From this, we deduce the graded base dimensions $b_k^{L,C,R}$ for $k = 2,3,4$ :
\begin{equation}
\begin{split}
 b_k^L &= \{0,0,0 \} \\
 b_k^C &= \{1,1,0 \}  \\
 b_k^R &= \{0,1,2 \}
 \end{split}
 \end{equation}
With the center parameters turned on, the spectral curve still has 4 nodes covering the node on $C$. But, rather than being free parameters (controlled by the $u_{k;C}$), the location of one of the nodes is fixed to $w=0$. Setting the center parameters to zero, the Hitchin integrable system on the 3-punctured sphere on the right is unchanged; it governs the Coulomb branch of the $T_4$ theory, as above. But, on the left, the symplectic integrable system is just a point. The physical theory on $C_L$ is ``ugly'': consisting of 6 free hypermultiplets, transforming as 2 copies of the defining representation of $SU(3)$.

The gauge group has been reduced from $G=SU(4)$ to $H=SU(3)$ and the center parameters are the Casimirs of $H$. The restricted node is thus $([4],SU(3))$.

\hypertarget{example_2}{}\subsubsection*{Example 2}
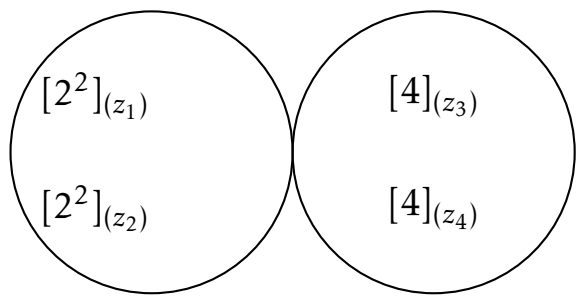
\label{example_2}

\begin{center}
\begin{tikzpicture}
\begin{scope}[scale=.75]
\node[scale=1.2] at (-1,1){$[2^2]_{(z_1)}$};
\node[scale=1.2] at (-1,-1){$[2^2]_{(z_2)}$};
\node[scale=1.2] at (5,1){$[4]_{(z_3)}$};
\node[scale=1.2] at (5,-1){$[4]_{(z_4)}$};
\draw[thick] (0,0.0) circle [radius=2.5];
\draw[thick] (5,0.0) circle [radius=2.5];
\end{scope}
\end{tikzpicture}
\end{center}

The orders of the zeroes of the $\phi_k$ are:

\begin{center}
\begin{tabular}{l|l|l|l|l}
$E_i$&$O_H$&$\phi_2$&$\phi_3$&$\phi_4$\\
\hline 
$E_1$&$[2^2]$&1&2&2\\
$E_2$&$[2^2]$&1&2&2\\
$E_3$&$[4]$&1&1&1\\
$E_4$&$[4]$&1&1&1\\
\end{tabular}
\end{center}
Using the results from table \ref{constraint_table}, for placing $[2^2]$ at $E_1$ and at $E_2$, we get the constraints
$$
u_{3;1}=u_{3;2}=0,\quad u_{4;L,L}=u_{4;L,R}=0
$$
Plugging these into \eqref{master_curve}, we get the spectral curve:

\begin{equation}
w^4 - x(y-z) [u_{2;C}w^2 + u_{3;R}z w + (u_{4;C}x(y-z) + u_{4;R,L}z(z-x) + u_{4;R,R}z y)] = 0
\end{equation}
So the graded base dimensions $b_k^{L,C,R}$ for $k = 2,3,4$ are:
\begin{equation}
\begin{split}
b_k^L &= \{0,0,0 \} \\
b_k^C &= \{1,0,1 \} \\
b_k^R &= \{0,1,2 \}
\end{split}
\end{equation}
Again with the center parameters turned on, there are 4 nodes on the spectral curve covering the node on $C$. This time, they are symmetrically-distributed about $w=0$.

The Hitchin system on the right remains that of the $T_4$ theory, while on the left, it is a point. The theory on the left is ``ugly'': 8 free hypermultiplets, transforming as 2 copies of the 4-dimensional defining representation of $Sp(2)$.

The center parameters are the Casimirs of the gauge group, $H=Sp(2)$ and we label the restricted node as $([4],Sp(2))$.

\hypertarget{example_3}{}\subsubsection*{Example 3}
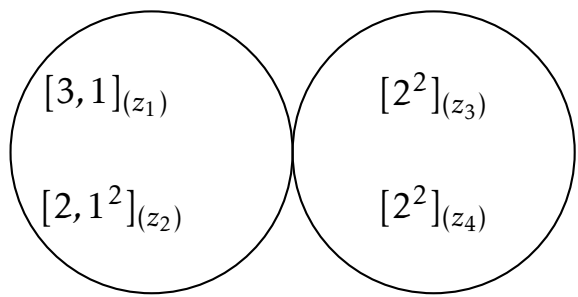
\label{example_3}

\begin{center}
\begin{tikzpicture}
\begin{scope}[scale=.75]
\node[scale=1.2] at (-.8,1){$[3,1]_{(z_1)}$};
\node[scale=1.2] at (-.7,-1){$[2,1^2]_{(z_2)}$};
\node[scale=1.2] at (5,1){$[2^2]_{(z_3)}$};
\node[scale=1.2] at (5,-1){$[2^2]_{(z_4)}$};
\draw[thick] (0,0.0) circle [radius=2.5];
\draw[thick] (5,0.0) circle [radius=2.5];
\end{scope}
\end{tikzpicture}
\end{center}
This example is, in a sense, a combination of Examples \hyperlink{example_1}{1} and \hyperlink{example_2}{2}.

The orders of the zeroes of the $\phi_k$ are:

\begin{center}
\begin{tabular}{l|l|l|l|l}
$E_i$&$O_H$&$\phi_2$&$\phi_3$&$\phi_4$\\
\hline 
$E_1$&$[3,1]$&1&1&2\\
$E_2$&$[2,1^2]$&1&2&3\\
$E_3$&$[2^2]$&1&2&2\\
$E_4$&$[2^2]$&1&2&2\\
\end{tabular}
\end{center}

The punctures on the left component of $C$ impose the constraints
$$
u_{3;2}=0,\quad u_{4;L,L}=0,\quad u_{4;L,R}=0,\quad u_{4;C}= (\lambda-1) u_{4;R,R}
$$
while the punctures on the right component of $C$ impose the constraints
$$
u_{3;3}=u_{3;1}+u_{3;2}=0,\quad u_{4;R,L}=u_{4;R,R}=0
$$
Putting these together, we get the spectral curve

\begin{equation}\label{su2nf4red}
w^2[ w^2- u_{2;C} x(y-z)]  = 0
\end{equation}
There's just one center parameter, corresponding to $H= SU(2) = SU(3)\cap Sp(2)$. Turning it off, the Hitchin integrable systems on both the left and the right are trivial. The restricted node is thus $([4],SU(2))$.

Note that this theory is \emph{globally} an ugly one: the 14 hypermultiplets (6 from the left and 8 from the right) transform as 4 copies of the defining representation of $SU(2)$ and 6 copies of the trivial representation. That is, 6 hypermultiplets remain free, \emph{everywhere} on $\overline{\mathcal{M}}_{0,4}$. In addition to the free hypermultiplets, we have the Hitchin integrable system for $SU(2)$ with $N_f=4$, whose spectral curve is the component of \eqref{su2nf4red} in square brackets.

\hypertarget{example_4}{}\subsubsection*{Example 4}
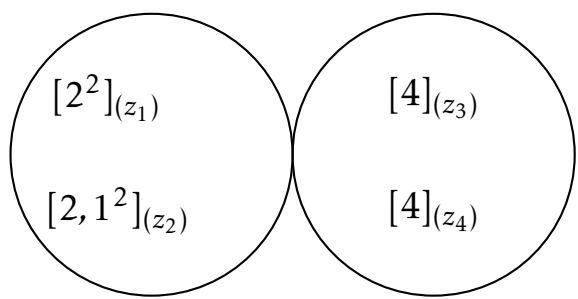
\label{example_4}

\begin{center}
\begin{tikzpicture}
\begin{scope}[scale=.75]
\node[scale=1.2] at (-.8,1){$[2^2]_{(z_1)}$};
\node[scale=1.2] at (-.6,-1){$[2,1^2]_{(z_2)}$};
\node[scale=1.2] at (5,1){$[4]_{(z_3)}$};
\node[scale=1.2] at (5,-1){$[4]_{(z_4)}$};
\draw[thick] (0,0.0) circle [radius=2.5];
\draw[thick] (5,0.0) circle [radius=2.5];
\end{scope}
\end{tikzpicture}
\end{center}

The orders of the zeroes of the $\phi_k$ are:

\begin{center}
\begin{tabular}{l|l|l|l|l}
$E_i$&$O_H$&$\phi_2$&$\phi_3$&$\phi_4$\\
\hline 
$E_1$&$[2^2]$&1&2&2\\
$E_2$&$[2,1^2]$&1&2&3\\
$E_3$&$[4]$&1&1&1\\
$E_4$&$[4]$&1&1&1\\
\end{tabular}
\end{center}

Using the constraints from table \ref{constraint_table}, the spectral curve (again dropping a term proportional to $u_{4;R,R}C_\lambda$) is

\begin{displaymath}
0=w^4 -x(y-z)[u_{2;C}w^2 + u_{3;R}z w+(u_{4;R,L}z(z-x)+u_{4;R,R} z x)]
\end{displaymath}
and graded base dimensions $b_k^{L,C,R}$ for $k = 2,3,4$ :
\begin{equation}
\begin{split}
b_k^L &= \{0,0,0 \} \\
b_k^C &= \{1,0,0 \} \\
b_k^R &= \{0,1,2 \}
\end{split}
\end{equation}
The singularity of the spectral curve, covering the node on $C$ is different from that  in \hyperlink{example_3}{Example 3}, but the restricted node is again $([4],SU(2))$. The constraint on $H$ is coming entirely from the left.

Turning off the center parameter, the integrable system on the right component is the same 3-dimensional Hitchin integrable system as in Examples \hyperlink{example_1}{1} and \hyperlink{example_2}{2}. On the left, we have the trivial theory, with no degrees of freedom.

\hypertarget{example_5}{}\subsubsection*{Example 5}
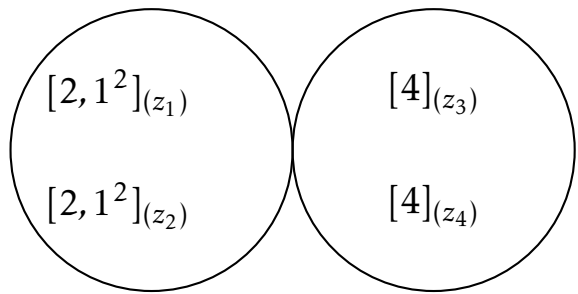
\label{example_5}

\begin{center}
\begin{tikzpicture}

\begin{scope}[scale=.75]
\node[scale=1.2] at (-.6,1){$[2,1^2]_{(z_1)}$};
\node[scale=1.2] at (-.6,-1){$[2,1^2]_{(z_2)}$};
\node[scale=1.2] at (5,1){$[4]_{(z_3)}$};
\node[scale=1.2] at (5,-1){$[4]_{(z_4)}$};
\draw[thick] (0,0.0) circle [radius=2.5];
\draw[thick] (5,0.0) circle [radius=2.5];
\end{scope}
\end{tikzpicture}
\end{center}

The orders of the zeroes of the $\phi_k$ are:

\begin{center}
\begin{tabular}{l|l|l|l|l}
$E_i$&$O_H$&$\phi_2$&$\phi_3$&$\phi_4$\\
\hline 
$E_1$&$[2,1^2]$&1&2&3\\
$E_2$&$[2,1^2]$&1&2&3\\
$E_3$&$[4]$&1&1&1\\
$E_4$&$[4]$&1&1&1\\
\end{tabular}
\end{center}

Using the constraints from table \ref{constraint_table}, the spectral curve (dropping a term proportional to $u_{4;R,R}C_\lambda$, and defining $u_{4;R,L}= u_{4;R,R}\equiv  u_{4;R}$) is

\begin{displaymath}
0=w^4 - x(y-z) [u_{2;C} w^2 + u_{3;R} z w + u_{4;R}z^2]
\end{displaymath}
The graded base dimensions $b_k^{L,C,R}$ for $k = 2,3,4$ :
\begin{equation}
\begin{split}
b_k^L &= \{0,0,0 \} \\
b_k^C &= \{1,0,0 \} \\
b_k^R &= \{0,1,1 \}
\end{split}
\end{equation}

Now the constraints imposed by the punctures on the left have forced a \emph{change} in the sub-integrable system on the right. It is no longer the 3-dimensional Hitchin system associated to the sphere with 3 regular Hitchin punctures. 

$\text{Res}(\Phi)_{z_3,z_4} \in [4]$ as before \emph{but} now $\text{Res}(\Phi)_{z'} \in [3,1]$, where $z'$ is the third puncture on $C_{0,3}^R$ (the right component in the normalization of the nodal curve). We can see this directly by looking at the spectral curve on the right component. With the center parameter turned off,
$$
0=w^4 - x(x-z) [u_{3;R} z w + u_{4;R}z^2]
$$
At two of the punctures ($x=0$ and $x=z$), the coefficients of $u_{3;R}$ and $u_{4;R}$ vanish to linear order, as expected for the Hitchin nilpotent, $[4]$. At the node ($z=0$), the coefficient of $u_{3;R}$ vanishes to linear order, but  the coefficient of $u_{4;R}$ vanishes to  quadratic order. This is the behaviour at the Hitchin nilpotent $[3,1]$.  On the left, we have the ugly theory, consisting of two free hypermultiplets. The Hitchin sub-integrable system on the right is the one associated to the 3-punctured sphere with Hitchin nilpotents $[3,1]$, $[4]$ and $[4]$. It governs the Coulomb branch geometry of the SCFT named $R_{0,4}$ in \cite{Chacaltana:2010ks}.

The restricted node is thus $([3,1],SU(2))$.

\hypertarget{example_6}{}\subsubsection*{Example 6}
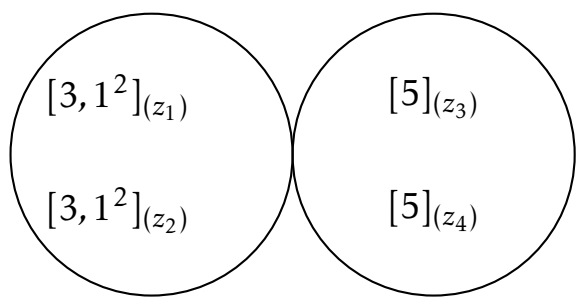
\label{example_6}

\begin{center}
\begin{tikzpicture}
\begin{scope}[scale=.750]
\node[scale=1.2] at (-.6,1){$[3,1^2]_{(z_1)}$};
\node[scale=1.2] at (-.6,-1){$[3,1^2]_{(z_2)}$};
\node[scale=1.2] at (5,1){$[5]_{(z_3)}$};
\node[scale=1.2] at (5,-1){$[5]_{(z_4)}$};
\draw[thick] (0,0.0) circle [radius=2.5];
\draw[thick] (5,0.0) circle [radius=2.5];
\end{scope}
\end{tikzpicture}
\end{center}

So far, all of our examples of restricted nodes have been accompanied by a \emph{trivial} sub-integrable system on one (or both) of the components of the nodal curve. This need not be the case, but the first nontrivial example occur in $\mathfrak{sl}(5)$.

Here is the table of zero orders of $\phi_k$ for this example

\begin{center}
\begin{tabular}{l|l|l|l|l|l}
$E_i$&$O_H$&$\phi_2$&$\phi_3$&$\phi_4$&$\phi_5$\\
\hline 
$E_1$&$[3,1^2]$&1&1&2&3\\
$E_2$&$[3,1^2]$&1&1&2&3\\
$E_3$&$[5]$&1&1&1&1\\
$E_4$&$[5]$&1&1&1&1\\
\end{tabular}
\end{center}

The spectral curve is

\begin{equation}
\begin{split}
0 = w^5 - &x(y-z) [u_{2;C}w^3 + (u_{3;1}x+u_{3;2}y+u_{3;3}z )w^2 \\
&+ (u_{4;C}x(y-z) + u_{4;R,L}z(z-x) + u_{4;R,R}z y)w\\
& + u_{5;R,R}y z^2+ u_{5;R,L}(z-x) z^2 + u_{5;R,C}x(y-z)z ]
\end{split}
\end{equation}

The graded dimensions $b_k^{L,C,R}$ for $k=2,3,4,5$ for the base  are
\begin{equation}
\begin{split}
b_k^L &= \{0,1,0,0 \} \\
b_k^C &= \{1,1,1,0 \} \\
b_k^R &= \{0,1,2,3\}
\end{split}
\end{equation}

Turning off the center parameters, the sub-integrable system on the right is the 6-dimensional Hitchin system corresponding to the three-punctured sphere with three regular punctures, $[5]$. On the left, the spectral curve is
$$
0=w\cdot\left[w^4-xy\left(u_{2;C}w^2+(u_{3;L}(x-y)+u_{3;C}y)w +u_{4;C}xy\right)\right]
$$
which, upon turning off the center parameters, becomes
\begin{equation}\label{MLsl5}
0=w^2\cdot\left[w^3-u_{3;L}xy(x-y)\right]
\end{equation}

The restricted node is $([5],SU(4))$. Physically, the neighbourhood of $\lambda=1$ is described as a weakly-coupled $\mathcal{N}=2$ $SU(4)$ gauge theory, gauging a diagonal $SU(4)$ subgroup of the $E_6$ global symmetry of the Minahan-Nemeschansky theory and one of the $SU(5)$s in the $SU(5)^3$ global symmetry of the $T_5$ theory.

The factor in square brackets of (\ref{MLsl5}) is the spectral curve of the one-dimensional integrable system governing the Coulomb branch of the $E_6$ Minahan-Nemeschansky theory \cite{Minahan:1996cj}. This integrable system \emph{does} have (more than one) realization as a Hitchin integrable system. For instance, it can be realized as the 3-punctured sphere with 3 regular nilpotents, $[3]$ of $\mathfrak{sl}(3)$. In the present case, we are seeing it appear as a (limit of the) $SL_5$ Hitchin system. There is, however, a crucial difference. Unlike in its $SL_3$ realization, there is no semistable $SL_5$ Higgs bundle moduli space on the 3-punctured sphere with \eqref{MLsl5} as its spectral curve. Any $SL_5$ Higgs bundle on $C_{0,3}$ with residues in $([3,1^2],[3,1^2],[5])$ is necessarily unstable (see \S\ref{stability} and Appendix \ref{OKappendix}). 


\hypertarget{example_7}{}\subsubsection*{Example 7}\label{example_7}

Finally, let us close this section with an example that combines the features of examples \hyperlink{example_5}{5} and \hyperlink{example_6}{6}: the sub-integrable systems on the left and right are both nontrivial \emph{and} the restricted node is not the regular nilpotent.

Consider the 4-punctured sphere
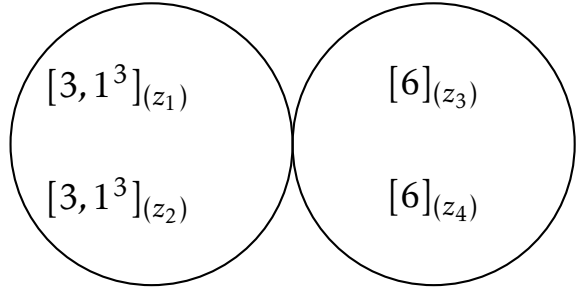
\begin{center}
\begin{tikzpicture}
\begin{scope}[scale=.75]
\node[scale=1.2] at (-.6,1){$[3,1^3]_{(z_1)}$};
\node[scale=1.2] at (-.6,-1){$[3,1^3]_{(z_2)}$};
\node[scale=1.2] at (5,1){$[6]_{(z_3)}$};
\node[scale=1.2] at (5,-1){$[6]_{(z_4)}$};
\draw[thick] (0,0.0) circle [radius=2.5];
\draw[thick] (5,0.0) circle [radius=2.5];
\end{scope}
\end{tikzpicture}
\end{center}
for $\mathfrak{sl}(6)$.

The zero orders of $\phi_k$ are

\begin{center}
\begin{tabular}{l|l|l|l|l|l|l}
$E_i$&$O_H$&$\phi_2$&$\phi_3$&$\phi_4$&$\phi_5$&$\phi_6$\\
\hline 
$E_1$&$[3,1^3]$&1&1&2&3&4\\
$E_2$&$[3,1^3]$&1&1&2&3&4\\
$E_3$&$[6]$&1&1&1&1&1\\
$E_4$&$[6]$&1&1&1&1&1\\
\end{tabular}
\end{center}
The spectral curve is
\begin{equation}
\begin{split}
0 = w^6 - &x(y-z) [u_{2;C}w^4 + (u_{3;1}x+u_{3;2}y+u_{3;3}z )w^3 \\
&+ (u_{4;C}x(y-z) + u_{4;R,L}z(z-x) + u_{4;R,R}z y)w^2\\
& + (u_{5;R,R}y z+ u_{5;R,L}(z-x) z + u_{5;R,C}x(y-z)) z w \\
& + (u_{6;R,R}y z+ u_{6;R,L}(z-x) z + u_{6;R,C}x(y-z)) z^2
 ]
\end{split}
\end{equation}
which yields the graded dimensions $b_k^{L,C,R}$ for $k=2,3,4,5,6$ for the Hitchin base :
\begin{equation}
\begin{split}
b_k^L &= \{0,1,0,0,0 \} \\
b_k^C &= \{1,1,1,0,0 \} \\
b_k^R &= \{0,1,2,3,3\}
\end{split}
\end{equation}
On the left ($z=0$), we get
$$
0= w^2\cdot \bigl[w^4-xy\bigl(u_{2;C}w^2+ (u_{3;C}y + u_{3;L}(x-y))w + u_{4;C}xy\bigr)\bigr]
$$
which, upon turning off the center parameters becomes
$$
0=w^3\cdot\left[w^3-u_{3;L}xy(x-y)\right]
$$
We recognize, again, the irreducible component in square brackets as the spectral curve of the $E_6$ Minahan-Nemeschansky SCFT. Again, this SCFT is not obtained as the Hitchin system on a 3-punctured sphere, with punctures $([3,1^3], [3,1^3], X)$ for any choice of nilpotent $X$.

On the right ($x-y=0$), the spectral curve is
\begin{equation*}
\begin{split}
0 = w^6 - &x(x-z) [u_{2;C}w^4 + (u_{3;C}y+u_{3;R}z )w^3 \\
&+ (u_{4;C}x(x-z) + u_{4;R,L}z(z-x) + u_{4;R,R}z y) w^2\\
& + (u_{5;R,R}x z+ u_{5;R,L}(z-x) z + u_{5;R,C}x(x-z)) z w \\
& + (u_{6;R,R}x z+ u_{6;R,L}(z-x) z + u_{6;R,C}x(x-z)) z^2
 ]
\end{split}
\end{equation*}
Setting the center parameters to zero, this becomes
\begin{equation*}
\begin{split}
0 = w^6 - &x(x-z) [ u_{3;R}z w^3 \\
&+ (u_{4;R,L}z(z-x)+ u_{4;R,R}xz)w^2\\
& + (u_{5;R,R}x z+ u_{5;R,L}(z-x) z + u_{5;R,C}x(x-z)) z w \\
& + (u_{6;R,R}x z+ u_{6;R,L}(z-x) z + u_{6;R,C}x(x-z)) z^2
 ]
\end{split}
\end{equation*}
Here, we see that $\phi_6$ has a double zero at the node, rather than a simple zero, implying that $O=[5,1]$. The center parameters are the Casimirs of $H=SU(4)$. So the restricted node is $([5,1], SU(4))$.

We saw that the tame Hitchin system on $C_{0,4}$ may have a standard node or a restricted node depending on the residues of the Higgs field at each of those punctures. So, it is natural ask what are the general conditions under which restricted nodes could occur and how does one characterize or classify restricted nodes. We now take up these questions in a systematic way in \S\ref{Hitchinnodal}, \S\ref{flavour} and \S\ref{classifying}.

\section{Tame Hitchin Systems on Nodal Curves}
\label{Hitchinnodal}

In \S\ref{4_standard_and_restricted_nodes}, we found a family of Hitchin integrable systems, with base $\mathcal{B}\to \mathcal{M}_{0,4}$, which extended as a flat family to the boundary of the moduli space where $C$ develops a node. Over the boundary, we found symplectic sub-integrable systems, with bases $B_L\oplus B_R\hookrightarrow B$. 

We would like to extend this story to $\overline{\mathcal{M}}_{g,n}$. Let $\tilde{C}$ be the normalization of the nodal curve $C$. The complex structure moduli space of $\tilde{C}$ is a component of the boundary of $\mathcal{M}_{g,n}$. More specifically, there are two qualitatively different cases,
\begin{equation}
\mathcal{M}_{\tilde{C}}=
\begin{cases}[1.5]
\mathcal{M}_{g_L,n_L+1}\times \mathcal{M}_{g_R,n_R+1}&\text{where}\; g_L+g_R=g\;\text{and}\; n_L+n_R=n\\
\mathcal{M}_{g-1,n+2}
\end{cases}
\end{equation}
In each case, this a codimension-1 divisor in $\overline{\mathcal{M}}_{g,n}$. The former is called a ``separating node''; the latter is called a ``non-separating node.''\footnote{In \S\ref{further_degen}, we will be interested in higher-codimension components of the boundary of $\mathcal{M}_{g,n}$, where $C$ has multiple nodes. It will still make sense to ask whether normalizing a given node splits the curve into disconnected pieces. This will \emph{always} be the case when $g=0$. That is, \emph{all} of the irreducible components of the boundary divisor of $\overline{\mathcal{M}}_{0,n}$ are of the form $\overline{\mathcal{M}}_{0,n_1+1}\times \overline{\mathcal{M}}_{0,n_2+1}$, with $n_1+n_2=n$.
}

Our aim, in this section is to sketch the construction of a family of Hitchin integrable systems with base $\mathcal{B}\to \overline{\mathcal{M}}_{g,n}$ and a symplectic sub-integrable system with base $\tilde{\mathcal{B}}\hookrightarrow \mathcal{B}\vert_{\mathcal{M}_{\tilde{C}}}$. We further want to exhibit, in the case of a separating node,  a decomposition $\tilde{\mathcal{B}}= \mathcal{B}_L\oplus \mathcal{B}_R$ where $\mathcal{B}_L\to \overline{\mathcal{M}}_{g_L,n_L+1}$ and $\mathcal{B}_R\to \overline{\mathcal{M}}_{g_R,n_R+1}$

In most cases, the resulting sub-integrable system will, again, be a semistable $J$-Hitchin system on $\tilde{C}$. The exception will echo what we saw in \S\ref{4_standard_and_restricted_nodes} in the case of $\overline{\mathcal{M}}_{0,4}$: when the node is restricted --- something we will see happens only in the case of a separating node when one or both of the components has genus-zero --- that  genus-zero component will yield a complex integrable system which is \emph{not} a semistable $J$-Hitchin system.

\subsection{Hitchin system on a nodal curve}\label{23_hitchin_system_on_a_gorenstein_curve}

Replace the smooth base curve $C$ by a Gorenstein curve: roughly, it can be singular, as long as there is still a good canonical line bundle (also called the dualizing sheaf) $K_C$. Any curve that is a divisor in a smooth surface will do. The canonical line bundle is given by the adjunction formula. This includes any curve whose only singularities are nodes. On a nodal curve, the sections of the canonical bundle are 1-forms on the normalization with first order poles allowed at the (inverse images of) the nodes, with opposite residues at the two inverse images of each node.

As in the smooth case, the Hitchin system for $C$ and a reductive group $G$ is the space $Higgs$ of (isomorphism classes of) $K_C$-valued $G$-Higgs bundles on $C$. A $G$-Higgs bundle is a pair $(V,\Phi)$ where $V$ is a principal $G$-bundle on $C$ and $\Phi \in H^0(C, ad(V) \otimes K_C)$. For now we will focus on the case $G=GL(N)$, so $V$ is a vector bundle and $\Phi: V \to V \otimes K_C$; or $G=SL(N)$, where $det(V)$ is required to be $O_C$ and the trace of $\Phi$ is required to vanish . As in the smooth case, one can consider a GIT version where the Higgs bundles are subject to a stability condition; or one can allow all Higgs bundles and work with the resulting stack.

Also as in the smooth case, the spectral curve of $(V,\Phi)$ is the curve in the total space of $K_C$ defined by the vanishing of the characteristic polynomial of the endomorphism $\Phi$. The Hitchin base $B$ is defined to be the space of all spectral curves. This can be identified with the vector space:

\begin{equation}
B \coloneqq \bigoplus_{k} H^0(C, K_C^{\otimes k})
\end{equation}
where $k$ runs over the degrees of the $G$-invariants. For $G=SL(N)$, these degrees are $k=2,\dots,N$. The Hitchin map $h: Higgs \to B$ sends $(V,\Phi)$ to (the coefficients of) the characteristic polynomial of $\Phi$.

When the Higgs field has poles along a divisor $D$ consisting of distinct smooth points of $C$, one can again define a Hitchin map whose image is now given by

\begin{equation}
B \coloneqq \bigoplus_{k} H^0(C, (K_C(D))^{\otimes k})
\end{equation}

In the conformal case, the residue of $\Phi$ at each of the marked points, $p_i$, must lie in some specified nilpotent orbit, $O_i$. Correspondingly, 
\begin{equation}\label{naive_Hitchin_base}
B \coloneqq \bigoplus_{k} H^0(C, L_k)
\end{equation}
where
\begin{equation}\label{Lkdef}
L_k = \bigl(K_C(D)\bigr)^{\otimes k} \otimes \mathcal{O}\bigl(-\sum_{p_i\in D} \chi_k^{(i)} p_i\bigr)
\end{equation}

The line bundles $L_k$ over each fiber fit together to form a line bundle $\mathcal{L}_k$ over the universal curve $\pi\colon \overline{\mathcal{C}}\to \overline{\mathcal{M}}_{g,n}$. The Hitchin bases \eqref{naive_Hitchin_base} fit together to form a torsion-free sheaf
\begin{equation}\label{HitchinDirectNaive}
\mathcal{B} =\bigoplus_{k=2}^N \pi_* \mathcal{L}_k
\end{equation}
For $\overline{\mathcal{M}}_{0,4}$, we computed these direct images rather explicitly in  \S\ref{directimages}. 

By Riemann-Roch\footnote{For $g>1$ and no marked points, $\mathcal{L}_k=K_C^{\otimes k}$, which has vanishing $H^1$. Adding marked points increases  $\deg(\mathcal{L}_k)$ and again $H^1=0$. For $g=1$, stability requires at least one marked point. This ensures $\deg(\mathcal{L}_k)>0$ and hence $H^1=0$. It is only for $g=0$ that a nonzero $H^1$ is possible and we just impose by hand that $\deg(\mathcal{L}_k) \geq -1$. }, the graded dimensions of $B$, when $C$ is smooth, are
\begin{equation}\label{gradedBdims}
b_k= (g-1)(2k-1)+\sum_{p_i\in D}(k- \chi_k^{(i)})
\end{equation}
These are necessarily non-negative for $g>0$. For $g=0$, we assume that the $\{O_i\}$ are such that the $b_k$ are non-negative for each $k$ (i.e.~that the SCFT is ``OK''). For a \emph{stable} nodal curve, the same holds true \emph{except} when $C$ is reducible and one (or more) of the components is genus-0. In that case, $H^0(C,\mathcal{L}_k)$ and $H^1(C,\mathcal{L}_k)$ can jump in dimension (though the difference remains constant). In that case, as we shall see in \S\ref{hitchinreducible}, the definition of the Hitchin base \eqref{naive_Hitchin_base} will need to be modified to \eqref{HitchinBaseTrueDef}, so that its graded dimensions are still given by \eqref{gradedBdims}. Globally, this will mean modifying \eqref{HitchinDirectNaive} to
\begin{equation}\label{HitchinDirect}
\mathcal{B} =\bigoplus_{k=2}^N \pi_* \mathcal{L}'_k
\end{equation}
so that $B$ is actually a vector bundle over $ \overline{\mathcal{M}}_{g,n}$. The global definition of the $ \mathcal{L}'_k$, as line bundles over the universal curve will be given in \S\ref{globalstory}. We will first work out what they have to be, fiber-by-fiber, starting with curves with a single node, and progressing to more singular nodal curves in \S\ref{further_degen}.

We now proceed to study the behaviour at the codimension-1 boundaries of $\mathcal{M}_{g,n}$, i.e.~where $C$ is smooth apart from a single node. We will see that the behaviour of the Hitchin system in the nodal limit is that dictated by the ``standard node'' described in \S\ref{4_standard_and_restricted_nodes}, \emph{except} when $C$ has two components, one (or both) of which is genus-$0$. When this is the case, and when a certain condition on the collection of marked points on that component is satisfied, we obtain a restricted node, $(O,H)$. As will be clear from the analysis, the same conclusion applies if we further degenerate surface. We obtain a restricted node only in the case where one (or both) side(s) of the node consists of a tree of  $\mathbb{P}^1$s (with the same condition on the marked points on that side).

\subsection{Hitchin system on a reducible nodal curve}
\label{hitchinreducible}
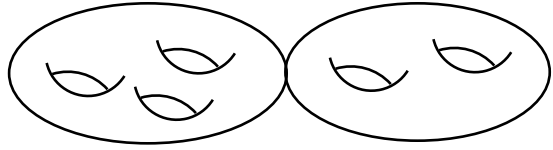
\begin{figure}
\begin{center}
\ifx\du\undefined
  \newlength{\du}
\fi
\setlength{\du}{10\unitlength}
\begin{tikzpicture}[even odd rule]
\pgftransformxscale{1.000000}
\pgftransformyscale{-1.000000}
\definecolor{dialinecolor}{rgb}{0.000000, 0.000000, 0.000000}
\pgfsetstrokecolor{dialinecolor}
\pgfsetstrokeopacity{1.000000}
\definecolor{diafillcolor}{rgb}{1.000000, 1.000000, 1.000000}
\pgfsetfillcolor{diafillcolor}
\pgfsetfillopacity{1.000000}
\pgfsetlinewidth{0.100000\du}
\pgfsetdash{}{0pt}
\definecolor{diafillcolor}{rgb}{1.000000, 1.000000, 1.000000}
\pgfsetfillcolor{diafillcolor}
\pgfsetfillopacity{1.000000}
\pgfpathellipse{\pgfpoint{30.600000\du}{15.750000\du}}{\pgfpoint{5.250000\du}{0\du}}{\pgfpoint{0\du}{2.650000\du}}
\pgfusepath{fill}
\definecolor{dialinecolor}{rgb}{0.000000, 0.000000, 0.000000}
\pgfsetstrokecolor{dialinecolor}
\pgfsetstrokeopacity{1.000000}
\pgfpathellipse{\pgfpoint{30.600000\du}{15.750000\du}}{\pgfpoint{5.250000\du}{0\du}}{\pgfpoint{0\du}{2.650000\du}}
\pgfusepath{stroke}
\pgfsetlinewidth{0.100000\du}
\pgfsetdash{}{0pt}
\definecolor{diafillcolor}{rgb}{1.000000, 1.000000, 1.000000}
\pgfsetfillcolor{diafillcolor}
\pgfsetfillopacity{1.000000}
\pgfpathellipse{\pgfpoint{42.900000\du}{15.700000\du}}{\pgfpoint{5.000000\du}{0\du}}{\pgfpoint{0\du}{2.600000\du}}
\pgfusepath{fill}
\definecolor{dialinecolor}{rgb}{0.000000, 0.000000, 0.000000}
\pgfsetstrokecolor{dialinecolor}
\pgfsetstrokeopacity{1.000000}
\pgfpathellipse{\pgfpoint{40.800000\du}{15.700000\du}}{\pgfpoint{5.000000\du}{0\du}}{\pgfpoint{0\du}{2.600000\du}}
\pgfusepath{stroke}
\pgfsetlinewidth{0.100000\du}
\pgfsetdash{}{0pt}
\pgfsetbuttcap
\definecolor{dialinecolor}{rgb}{0.000000, 0.000000, 0.000000}
\pgfsetstrokecolor{dialinecolor}
\pgfsetstrokeopacity{1.000000}
\pgfpathmoveto{\pgfpoint{30.949978\du}{14.499911\du}}
\pgfpatharc{167}{32}{1.613125\du and 1.613125\du}
\pgfusepath{stroke}
\pgfsetlinewidth{0.100000\du}
\pgfsetdash{}{0pt}
\pgfsetbuttcap
\definecolor{dialinecolor}{rgb}{0.000000, 0.000000, 0.000000}
\pgfsetstrokecolor{dialinecolor}
\pgfsetstrokeopacity{1.000000}
\pgfpathmoveto{\pgfpoint{33.250017\du}{15.500019\du}}
\pgfpatharc{318}{253}{2.018356\du and 2.018356\du}
\pgfusepath{stroke}
\pgfsetlinewidth{0.100000\du}
\pgfsetdash{}{0pt}
\pgfsetbuttcap
\definecolor{dialinecolor}{rgb}{0.000000, 0.000000, 0.000000}
\pgfsetstrokecolor{dialinecolor}
\pgfsetstrokeopacity{1.000000}
\pgfpathmoveto{\pgfpoint{26.771947\du}{15.356879\du}}
\pgfpatharc{167}{32}{1.613125\du and 1.613125\du}
\pgfusepath{stroke}
\pgfsetlinewidth{0.100000\du}
\pgfsetdash{}{0pt}
\pgfsetbuttcap
\definecolor{dialinecolor}{rgb}{0.000000, 0.000000, 0.000000}
\pgfsetstrokecolor{dialinecolor}
\pgfsetstrokeopacity{1.000000}
\pgfpathmoveto{\pgfpoint{29.071985\du}{16.356987\du}}
\pgfpatharc{318}{253}{2.018356\du and 2.018356\du}
\pgfusepath{stroke}
\pgfsetlinewidth{0.100000\du}
\pgfsetdash{}{0pt}
\pgfsetbuttcap
\definecolor{dialinecolor}{rgb}{0.000000, 0.000000, 0.000000}
\pgfsetstrokecolor{dialinecolor}
\pgfsetstrokeopacity{1.000000}
\pgfpathmoveto{\pgfpoint{37.486947\du}{15.156879\du}}
\pgfpatharc{167}{32}{1.613125\du and 1.613125\du}
\pgfusepath{stroke}
\pgfsetlinewidth{0.100000\du}
\pgfsetdash{}{0pt}
\pgfsetbuttcap
\definecolor{dialinecolor}{rgb}{0.000000, 0.000000, 0.000000}
\pgfsetstrokecolor{dialinecolor}
\pgfsetstrokeopacity{1.000000}
\pgfpathmoveto{\pgfpoint{39.786985\du}{16.156987\du}}
\pgfpatharc{318}{253}{2.018356\du and 2.018356\du}
\pgfusepath{stroke}
\pgfsetlinewidth{0.100000\du}
\pgfsetdash{}{0pt}
\pgfsetbuttcap
\definecolor{dialinecolor}{rgb}{0.000000, 0.000000, 0.000000}
\pgfsetstrokecolor{dialinecolor}
\pgfsetstrokeopacity{1.000000}
\pgfpathmoveto{\pgfpoint{41.401947\du}{14.456879\du}}
\pgfpatharc{167}{32}{1.613125\du and 1.613125\du}
\pgfusepath{stroke}
\pgfsetlinewidth{0.100000\du}
\pgfsetdash{}{0pt}
\pgfsetbuttcap
\definecolor{dialinecolor}{rgb}{0.000000, 0.000000, 0.000000}
\pgfsetstrokecolor{dialinecolor}
\pgfsetstrokeopacity{1.000000}
\pgfpathmoveto{\pgfpoint{43.701985\du}{15.456987\du}}
\pgfpatharc{318}{253}{2.018356\du and 2.018356\du}
\pgfusepath{stroke}
\pgfsetlinewidth{0.100000\du}
\pgfsetdash{}{0pt}
\pgfsetbuttcap
\definecolor{dialinecolor}{rgb}{0.000000, 0.000000, 0.000000}
\pgfsetstrokecolor{dialinecolor}
\pgfsetstrokeopacity{1.000000}
\pgfpathmoveto{\pgfpoint{30.116947\du}{16.256879\du}}
\pgfpatharc{167}{32}{1.613125\du and 1.613125\du}
\pgfusepath{stroke}
\pgfsetlinewidth{0.100000\du}
\pgfsetdash{}{0pt}
\pgfsetbuttcap
\definecolor{dialinecolor}{rgb}{0.000000, 0.000000, 0.000000}
\pgfsetstrokecolor{dialinecolor}
\pgfsetstrokeopacity{1.000000}
\pgfpathmoveto{\pgfpoint{32.416985\du}{17.256987\du}}
\pgfpatharc{318}{253}{2.018356\du and 2.018356\du}
\pgfusepath{stroke}
\end{tikzpicture}
\end{center}
\caption{A Riemann surface $C$ of genus $g$ develops a separating node. The nodal curve has two components whose  genera $g_L,g_R$ obey $g=g_L+g_R$.}
\end{figure}

Let us first consider the case of a \emph{separating node}. The base curve $C_{g,n}$ of the Hitchin system has a single node at the point $p$ and is \emph{reducible}. The normalization $v\colon C_{L} \amalg C_{R} \rightarrow C_{g,n}$, where $C_L$ and $C_R$ have genus $g_{L,R}$ respectively, satisfying
$$
   g_{L}+g_{R}=g
$$
In this subsection, we will only consider degenerations where $C_L$ and $C_R$ are themselves smooth stable curves.

Let $\mathcal{C}_L$ and $\mathcal{C}_R$ be the corresponding divisors in $\overline{\mathcal{C}}$. Then $\mathcal{O}(-\mathcal{C}_L)$ and $\mathcal{O}(-\mathcal{C}_R)$ are line bundles on $\overline{\mathcal{C}}$. Restricted to each component of the fiber,
\begin{subequations}\label{globaltwist}
\begin{align}
\deg(\mathcal{O}_{C_L}(-\mathcal{C}_L))&=1\label{globaltwistlA}\\
\deg(\mathcal{O}_{C_R}(-\mathcal{C}_R))&=1\label{globaltwistlB}\\
\mathcal{O}_{C_L}(-\mathcal{C}_R))&= \mathcal{O}_{C_L}(-p)\label{globaltwistlC}\\
\mathcal{O}_{C_R}(-\mathcal{C}_L))&= \mathcal{O}_{C_R}(-p)\label{globaltwistlD}
\end{align}
\end{subequations}

Using the restriction maps to $C_L$ and $C_R$, we define the sheaves on $C$
\begin{equation}
\begin{split}
\mathcal{L}_{k,L}& = \ker(r_R: \mathcal{L}_k\to\mathcal{L}_k |_{C_R}) = \mathcal{L}_k\otimes \mathcal{O}_{C_L}(-p)\\
\mathcal{L}_{k,R}&= \ker(r_L:   \mathcal{L}_k\to\mathcal{L}_k |_{C_L}) = \mathcal{L}_k\otimes \mathcal{O}_{C_R}(-p)
\end{split}
\end{equation}
These fit into a short exact sequence
\begin{equation}\label{mySES}
0\to \mathcal{L}_{k,L}\oplus \mathcal{L}_{k,R}\to \mathcal{L}_k\to S_p\to 0
\end{equation}
where $S_p$ is a skyscraper sheaf supported at $p$.

First, let us assume that $H^1(C,\mathcal{L}_k)=0$. Then, since $H^0(C,S_p)=\mathbb{C}$, the long exact sequence associated to \eqref{mySES}
\begin{equation}
0\to H^0(C,\mathcal{L}_{k,L})\oplus H^0(C,\mathcal{L}_{k,R})\to H^0(C,\mathcal{L}_k)\xrightarrow{\alpha}H^0(S_p)\to H^1(C,\mathcal{L}_{k,L})\oplus H^1(C,\mathcal{L}_{k,R})\to 0
\end{equation}
splits, either as
\begin{subequations}
\begin{equation}\label{heresalpha}
\begin{gathered}
0\to H^0(C,\mathcal{L}_{k,L})\oplus H^0(C,\mathcal{L}_{k,R})\to H^0(C,\mathcal{L}_k)\xrightarrow{\alpha} \mathbb{C}\to 0\\
H^1(C,\mathcal{L}_{k,L})\oplus H^1(C,\mathcal{L}_{k,R})=0
\end{gathered}
\end{equation}
or as
\begin{equation}
\begin{gathered}
H^0(C,\mathcal{L}_k) = H^0(C,\mathcal{L}_{k,L})\oplus H^0(C,\mathcal{L}_{k,R})\\
H^1(C,\mathcal{L}_{k,L})\oplus H^1(C,\mathcal{L}_{k,R})=\mathbb{C}
\end{gathered}
\end{equation}
\end{subequations}
depending on whether the residue map $\alpha$ is nonzero.

If $H^1(C,\mathcal{L}_k\otimes\mathcal{O}_{C_L})\neq 0$, then it follows from the long exact sequence associated to
\begin{equation}
0\to \mathcal{L}_{k,R} \to \mathcal{L}_k\to \mathcal{L}_k\otimes\mathcal{O}_{C_L}\to 0
\end{equation}
that $H^1(C,\mathcal{L}_k)\neq 0$. Since this vanished on the smooth curve, we are in the situation where the cohomology groups of $\mathcal{L}_k$ jump in the nodal limit.

To fix this, we tensor with an appropriate power of the line bundles defined in \eqref{globaltwist}. Let $n_k^L$ be the smallest non-negative integer such that $H^1\bigl(C_L,\mathcal{L}_k\otimes \mathcal{O}_{C_L}(n_k^Lp)\bigr)=0$ and let $n_k^R$ be the smallest non-negative integer such that $H^1\bigl(C_R,\mathcal{L}_k\otimes \mathcal{O}_{C_R}(n_k^R p)\bigr)=0$. 

As we shall see in \S\ref{nilpnode} below,
\begin{subequations}
\begin{align}
n_k^L&=\max(0, -d_k^L-1)\label{nlkdef}\\
n_k^R&=\max(0, -d_k^R-1)\label{nrkdef}
\end{align}
 where $d_k^L$ and $d_k^R$ are defined in \eqref{virtdimLR}. Our "OK" assumption is that $H^1(C,\mathcal{L}_k)=0$ on the smooth curve. If that is the case, then at most one of $n_k^L$ and $n_k^R$ can be nonzero\footnote{The argument is used repeatedly in this section, so let us spell it out here. $n_k$ can be nonzero only if the component has genus-$0$. If both components have genus-$0$, then the OK condition is that the total degree of $\mathcal{L}_k$ is $\geq -1$. If the degree of $\mathcal{L}_k$ is $\leq -2$ on one component (the condition for $n_k$ to be nonzero on that component), then it must be positive on the other component.}. Set
\begin{equation}\label{isthislprimedef}
\mathcal{L}'_k \coloneqq \mathcal{L}_k\otimes \mathcal{O}(-n_k^L \mathcal{C}_L-n_k^R \mathcal{C}_R)
\end{equation}
Then, as before, we define
\label{twistdef}
\begin{equation}\label{Lprimedef}
\begin{split}
\mathcal{L}'_{k,L}& \coloneqq \mathcal{L}'_k\otimes \mathcal{O}_{C_L}(-p)=  \mathcal{L}_{k,L} \otimes \mathcal{O}((n_k^L-n_k^R) p)\\
\mathcal{L}'_{k,R} &\coloneqq \mathcal{L}'_k\otimes \mathcal{O}_{C_R}(-p)
=\mathcal{L}_{k,R} \otimes \mathcal{O}(-(n_k^L-n_k^R) p)
\end{split}
\end{equation}
\end{subequations}

It is important to note that, for $n^L_k >0$, the twisting \eqref{Lprimedef} does not introduce\footnote{Similarly, for $n^R_k>0$, we do not induce a nonzero $H^1(C,\mathcal{L}'_{k,L})$.} a nonzero $H^1(C,\mathcal{L}'_{k,R})$. To see this, note that, from the definitions of $\mathcal{L}_{k,R}$ and $n^L_k$, we can write $\mathcal{L}'_{k,R}$ in terms of the canonical bundle of the normalization
\begin{equation*}
\begin{split}
\mathcal{L}'_{k,R} &= K^{\otimes k}_{C_R} \otimes \mathcal{O}\Bigl((k-n^L_k-1)p+\sum_{p_i\in D_R} (k-\chi_k^{(i)})p_i\Bigr)\\
&= K^{\otimes k}_{C_R} \otimes \mathcal{O}\Bigl(\min\bigl(k-1,-1+g_L(2k-1)+\sum_{p_j\in D_L}(k-\chi_k^{(j)})\bigr)p+\sum_{p_i\in D_R} (k-\chi_k^{(i)})p_i\Bigr)
\end{split}
\end{equation*}
To show that $H^1(C_R,\mathcal{L}'_{k,R})=0$, we need to consider a few different cases.
\begin{itemize}
\item For $g_R>1$, $K^{\otimes k}_{C_R}$ has vanishing $H^1$ and the line bundle we twist it by has positive degree.
\item For $g_R=1$, $K^{\otimes k}_{C_R}$ is trivial but, again, we are twisting by a positive line bundle. In both cases, this leads to vanishing $H^1$.
\item For $g_R=0$, $D_R$ must consist of at least 2 marked points and the nilpotents located there must be such that $H^1(C,\mathcal{L}_k)=0$. This, again, is sufficient to ensure that $\deg(\mathcal{L}'_{k,R}) \geq -1$ and hence $H^1=0$.
\end{itemize}
As we shall see shortly, the twist \eqref{nlkdef} is only nontrivial when $g_L=0$, in which case, $\deg\bigl(\mathcal{L}_k\otimes \mathcal{O}_{C_L}(n_k^L p)\bigr)=-1$ and hence $H^0\bigl(C,\mathcal{L}_k\otimes \mathcal{O}_{C_L}(n_k^L p)\bigr)=0$ as well.

Finally, we can give our definition for the Hitchin base on the reducible nodal curve. The Hitchin base is
\begin{equation}\label{HitchinBaseTrueDef}
 B:=\bigoplus_{k=2}^N H^0(C,\mathcal{L}'_k)
\end{equation}
Or, globally, the family of Hitchin bases is 
\begin{equation}\label{GlobalHitchinBaseTrueDef}
\mathcal{B}:=\bigoplus_{k=2}^N \pi_*\mathcal{L}'_k
\end{equation}

On each component of the nodal curve, we define the Hitchin bases
\begin{equation}
\label{newbase}
\begin{split}
B_L& \coloneqq \bigoplus_{k=2}^N H^0(C, \mathcal{L}'_{k,L})\\
B_R& \coloneqq \bigoplus_{k=2}^N H^0(C, \mathcal{L}'_{k,R})
\end{split}
\end{equation}
By construction, we have the inclusion
\begin{equation}\label{inclusion}
H^0(C, \mathcal{L}'_{k,L})\oplus H^0(C, \mathcal{L}'_{k,R})\hookrightarrow H^0(C, \mathcal{L}'_{k})
\end{equation}
which is characterized by the following properties
\begin{itemize}
\item $H^0(C, \mathcal{L}'_{k,L})$ and $H^0(C, \mathcal{L}'_{k,R})$ are \emph{disjoint} subspaces of $H^0(C, \mathcal{L}'_{k})$ and hence the map in \eqref{inclusion} is injective.
\item For each $k$, the quotient in \eqref{inclusion} is either $0$ or $1$ dimensional. 
\begin{itemize}
\item When the twist, $n_k^L$ is nonzero then we have 
$$H^0(C, \mathcal{L}'_{k,L})=0\;\;\text{and}\;\; H^0(C, \mathcal{L}'_{k})=H^0(C, \mathcal{L}'_{k,R})\quad .$$
Conversely, when $n_k^R$ is nonzero, we have
$$H^0(C, \mathcal{L}'_{k,R})=0\;\;\text{and}\;\; H^0(C, \mathcal{L}'_{k})=H^0(C, \mathcal{L}'_{k,L})\quad .$$
In either case, the quotient vanishes, as it also does when the map $\alpha$ in \eqref{heresalpha} is zero.
\end{itemize}
\end{itemize}
 The space of \emph{center parameters},
\begin{equation}
B_C = B/(B_L\oplus B_R)
\end{equation}
is the direct sum of these $1$ dimensional spaces. The Hitchin fibers which are symplectic-dual to $B_C$ are the ones which become noncompact in the nodal limit.

Denoting the graded dimensions of these spaces as $b_k^L$, $b_k^R$ $b_k^C$ and $b_k$, we obviously have
\begin{equation}\label{btotal}
b_k = b_k^L + b_k^R + b_k^C
\end{equation}
The case where $b_k^C= 1$ for all $k$ corresponds to what we called a \textit{standard node} in \S\ref{4_standard_and_restricted_nodes}. Conversely, when some of the $b_k^C=0$, the node is \emph{restricted}.

The dimensions in \eqref{btotal} are given by
\begin{equation}\label{truedimLR}
b_k^L=\max\left(d_k^L-\max(-d_k^R-1,0),0\right),\qquad b_k^R=\max\left(d_k^R-\max(-d_k^L-1,0),0\right)
\end{equation}
where the index of $\mathcal{L}_{k,L}$ and $\mathcal{L}_{k,R}$ are
\begin{equation}\label{virtdimLR}
\begin{split}
d_k^L&= (g_L-1)(2k-1) +(k-1) + \sum_{p_i\in D_L}(k- \chi_k^{(i)})\\
d_k^R&= (g_R-1)(2k-1)+(k-1) + \sum_{p_i\in D_R}(k- \chi_k^{(i)})\\
\end{split}
\end{equation}
If $g_L$ and $g_R$ are both positive, then $d_k^L$ and $d_k^R$ are both $\geq k-1$ and hence, combining \eqref{virtdimLR} with \eqref{gradedBdims},  $b_k^C=1$ for each $k$. So we get the standard node.

Only if one or both of $C_{L,R}$ are genus-$0$, will a separating node be restricted. Without loss of generality, let $C_L$ have genus-$0$. In this case \eqref{virtdimLR} simplifies to
\begin{equation*}
d_k^L= -k + \sum_{p_i\in D_L} (k-\chi_k^{(i)})
\end{equation*}
where, for stability, we must have $\deg(D_L)\geq 2$. So if
\begin{equation}\label{bck0cond}
\sum_{p_i\in D_L} (k-\chi_k^{(i)}) < k  
\end{equation}
for some $k$, then we have the corresponding $b_k^C=0$, and hence a restricted node, as discussed in \S\ref{first_restricted_nodes}.

Our definition of restricted node amounted to asserting that $b_k^C = 0$ for some $k$. It will prove useful in \S\ref{classifying} to have some alternative formulations of this condition. To this end, we prove the following Proposition. 

\begin{prop}The following conditions are equivalent.
\label{centervanishinglemma}
\begin{enumerate}
\item[i.] $b_k^C = 0$ .
\item[ii.]Either\hfil\break
${\vphantom{\biggl(}}\quad H^1(C, \mathcal{L}'_{k,L})=\mathbb{C},\quad
H^0(C, \mathcal{L}'_{k,R})=0$\hfil\break
or\hfil\break
${\vphantom{\biggl(}}\quad H^0(C, \mathcal{L}'_{k,L})=0,\quad H^1(C, \mathcal{L}'_{k,R})=\mathbb{C}
$\hfil\break
(but not both).
\item[iii.] $H^0(C, \mathcal{L}'_k\otimes \mathcal{O}_{C_L}) = 0$ or $H^0(C, \mathcal{L}'_k\otimes \mathcal{O}_{C_R}) = 0$ (or both).
\item[iv.] $H^0(C, \mathcal{L}_k\otimes \mathcal{O}_{C_L}) = 0$ or $H^0(C, \mathcal{L}_k\otimes \mathcal{O}_{C_R}) = 0$ (or both).
\end{enumerate}
\end{prop}

\begin{proof}
To prove the equivalence of (i) and (ii), consider the short exact sequence (equation \eqref{mySES}, but for $\mathcal{L}'_k$ instead of $\mathcal{L}_k$):
\begin{equation}\label{myprimed SES}
0\to \mathcal{L}'_{k,L}\oplus \mathcal{L}'_{k,R}\to \mathcal{L}'_k\to S'_p\to 0
\end{equation}
and the corresponding long exact sequence
\begin{equation}
0\to H^0(C,\mathcal{L}'_{k,L})\oplus H^0(C,\mathcal{L}'_{k,R})\to H^0(C,\mathcal{L}'_k)\xrightarrow{\alpha'}H^0(S'_p)\to H^1(C,\mathcal{L}'_{k,L})\oplus H^1(C,\mathcal{L}'_{k,R})\to 0
\end{equation}
The latter splits, either as
\begin{subequations}
\begin{equation}
\begin{gathered}
0\to H^0(C,\mathcal{L}'_{k,L})\oplus H^0(C,\mathcal{L}'_{k,R})\to H^0(C,\mathcal{L}'_k)\xrightarrow{\alpha'} \mathbb{C}\to 0\\
H^1(C,\mathcal{L}'_{k,L})\oplus H^1(C,\mathcal{L}'_{k,R})=0
\end{gathered}
\end{equation}
or as
\begin{equation}
\begin{gathered}
H^0(C,\mathcal{L}'_k) = H^0(C,\mathcal{L}'_{k,L})\oplus H^0(C,\mathcal{L}'_{k,R})\\
H^1(C,\mathcal{L}'_{k,L})\oplus H^1(C,\mathcal{L}'_{k,R})=\mathbb{C}
\end{gathered}
\end{equation}
\end{subequations}
depending on whether the residue map $\alpha'$ is nonzero. The former corresponds to $b^C_k=1$; the latter to $b^C_k=0$. But the latter holds if and only if (ii).

To prove the equivalence of (ii) and (iii), we note that $\deg(\mathcal{L}'_k\otimes \mathcal{O}_{C_{L}})\geq k(2g_{L}-1)$. So if $g_L\geq 1$ then  $H^1(\mathcal{L}'_{k,L})= H^1(\mathcal{L}'_k\otimes \mathcal{O}_{C_{L}}(-p))=0$ and $H^0(\mathcal{L}'_k\otimes \mathcal{O}_{C_{L}})\neq0$. To get $H^1(\mathcal{L}'_{k,L})\neq 0$ or $H^0(\mathcal{L}'_k\otimes \mathcal{O}_{C_{L}})=0$, we must have $g_L=0$. But if $g_L=0$ then  $H^1(\mathcal{L}'_k\otimes \mathcal{O}_{C_{L}}(-p))\neq 0\, \Leftrightarrow \, H^0(\mathcal{L}'_k\otimes \mathcal{O}_{C_{L}})=0$. The same applies to $\mathcal{L}'_{k,R}$.

Finally, let us consider (iv). If $n^L_k=n^R_k=0$, then $\mathcal{L}'_k=\mathcal{L}_k$ and (iii) is equivalent to (iv). So, without loss of generality, consider the case $g_L=0$ and  $n^L_k >0$. Then
$$ \deg(\mathcal{L}_k\otimes \mathcal{O}_{C_L}) = \deg(\mathcal{L}'_k\otimes \mathcal{O}_{C_L}) - n^L_k=-1- n^L_k\leq -2,$$
in which case $H^0(\mathcal{L}_k\otimes \mathcal{O}_{C_L})=0=H^0(\mathcal{L}'_k\otimes \mathcal{O}_{C_L})$.
\end{proof}

\subsection{The nilpotent at the node}\label{nilpnode}

Proposition \ref{centervanishinglemma} gives us the conditions under which we have a reduction in the number of center parameters. One observes that these condition obtain only if at least one of the components is a genus zero curve. For this section, let us imagine we are in one of these situations and without loss of generality, let us take $C_L$ to be the genus zero component.  We prove in \S\ref{classifying} that the non-zero center parameters can \textit{always} be identified with the invariant polynomials for some simple Lie subgroup $H \subset J$. This determines the $H$ in the pair $(O,H)$ with which we label restricted nodes. 

We now turn to the other entry in the pair. To understand its role, we will need to look in more detail at the behaviour of sections of $\mathcal{L}'_k$ at a node. 



If $d_k^L$ and $d_k^R$ in \eqref{virtdimLR} are both $\geq -1$, for each $k$, then $O=[N]$. If, for some values of $k$, $d_k^L\leq -2$, then the corresponding  $\mathcal{L}_k\otimes \mathcal{O}_{C_L}$ has higher cohomology.
\begin{equation}
h^1(C,\mathcal{L}_k\otimes \mathcal{O}_{C_L}) = \max(-1-d_k^L,0)= n_k^L
\label{h1equation}
\end{equation}
That component of the Hitchin base, $B_R$, is given not by $H^0(C,\mathcal{L}_{k,R})$, but rather by
\begin{equation}\label{twist}
H^0(C,\mathcal{L}'_{k,R}) = H^0\left(C, \mathcal{L}_{k,R}\otimes \mathcal{O}(-n_k^L p)\right)
\end{equation}
That is, $\phi_k$ has a zero of order $n_k^L+1$, rather than a simple zero, at the node. This is the Hitchin base associated to $C_R$, where the nilpotent, $O$, at the node has vanishing orders of the $\phi_k$ given by
 
\begin{equation}
\chi_k = 1 + n_k^L
\label{modifiedvanishing}
\end{equation} 

We show in \S\ref{classifying} that $\chi_k$ as defined above is the set of vanishing orders for some nilpotent orbit $O$. Since the vanishing orders uniquely identify nilpotent orbits in $\mathfrak{sl}(N)$, this gives us the other entry in the pair $(O,H)$.  We also give a different prescription for determining $O$  in \S\ref{restricted_nodes} based on properties of the Higgs branch. 

Let us, however, pause to follow the intuition we developed using global methods in examples \hyperlink{example_5}{5}, \hyperlink{example_7}{7} of \S\ref{examples} to see if it matches the prescription in \eqref{modifiedvanishing}. In each of those examples, when we turn off the center parameters, some of the $\phi_k$ on $C_R$ have higher-order zeroes (instead of simple zeroes) at the node. We claim that these higher order zeros are precisely accounted for by the twist in \eqref{twist}.

Let's check these assertions. Specializing to $g_L=0$ and $\deg(D_L)=2$, \eqref{virtdimLR} simplifies to
$$
d_k^L= k-\sum_{p_i\in D_L} \chi_k^{(i)}
$$
For both examples, $d_k^R=k-2$.
\begin{itemize}
\item For \hyperlink{example_5}{example 5}, $\chi_k^{[2,1^2]} = k-1$, so $d_k^L=2-k$. Hence $b_k^L=(0,0,0)$, $b_k^R=(0,1,1)$, $b_k^C=(1,0,0)$ and $h^1(C,\mathcal{L}_k\otimes\mathcal{O}_{C_L})=(0,0,1)$. This agrees with our previous analysis, where we found that (after turning off the center parameters) $\phi_4$ had a double zero at the node.
\item For \hyperlink{example_7}{example 7}, $\chi_k^{[3,1^3]}=\max(k-2,1)$, so $d_k^L=(0,1,0,-1,-2)$. Hence $b_k^L=(0,1,0,0,0)$, $b_k^R=(0,1,2,3,3)$,  $b_k^C=(1,1,1,0,0)$ and  $h^1(C,\mathcal{L}_k\otimes\mathcal{O}_{C_L})=(0,0,0,0,1)$, which agrees with our previous analysis, where we found that (after turning off the center parameters) $\phi_6$ had a double zero at the node.
\end{itemize}

\subsection{Hitchin system on an irreducible nodal curve}\label{24_hitchin_system_on_an_irreducible_nodal_curve}

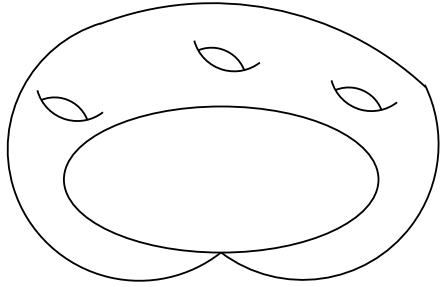
\begin{figure}
\begin{center}
\ifx\du\undefined
  \newlength{\du}
\fi
\setlength{\du}{7\unitlength}
\begin{tikzpicture}[even odd rule]
\pgftransformxscale{1.000000}
\pgftransformyscale{-1.000000}
\definecolor{dialinecolor}{rgb}{0.000000, 0.000000, 0.000000}
\pgfsetstrokecolor{dialinecolor}
\pgfsetstrokeopacity{1.000000}
\definecolor{diafillcolor}{rgb}{1.000000, 1.000000, 1.000000}
\pgfsetfillcolor{diafillcolor}
\pgfsetfillopacity{1.000000}
\pgfsetlinewidth{0.100000\du}
\pgfsetdash{}{0pt}
\definecolor{diafillcolor}{rgb}{1.000000, 1.000000, 1.000000}
\pgfsetfillcolor{diafillcolor}
\pgfsetfillopacity{1.000000}
\pgfusepath{fill}
\definecolor{dialinecolor}{rgb}{0.000000, 0.000000, 0.000000}
\pgfsetstrokecolor{dialinecolor}
\pgfsetstrokeopacity{1.000000}
\pgfpathellipse{\pgfpoint{41.450000\du}{14.050000\du}}{\pgfpoint{8.500000\du}{0\du}}{\pgfpoint{0\du}{3.950000\du}}
\pgfusepath{stroke}
\pgfsetlinewidth{0.100000\du}
\pgfsetdash{}{0pt}
\pgfsetbuttcap
{
\definecolor{diafillcolor}{rgb}{0.000000, 0.000000, 0.000000}
\pgfsetfillcolor{diafillcolor}
\pgfsetfillopacity{1.000000}
\definecolor{dialinecolor}{rgb}{0.000000, 0.000000, 0.000000}
\pgfsetstrokecolor{dialinecolor}
\pgfsetstrokeopacity{1.000000}
\pgfpathmoveto{\pgfpoint{34.950213\du}{5.599934\du}}
\pgfpatharc{253}{52}{7.118235\du and 7.118235\du}
\pgfusepath{stroke}
}
\pgfsetlinewidth{0.100000\du}
\pgfsetdash{}{0pt}
\pgfsetbuttcap
\definecolor{dialinecolor}{rgb}{0.000000, 0.000000, 0.000000}
\pgfsetstrokecolor{dialinecolor}
\pgfsetstrokeopacity{1.000000}
\pgfpathmoveto{\pgfpoint{52.500178\du}{9.05561\du}}
\pgfpatharc{313}{249}{16.910069\du and 16.910069\du}
\pgfusepath{stroke}
\pgfsetlinewidth{0.100000\du}
\pgfsetdash{}{0pt}
\pgfsetbuttcap
{
\definecolor{diafillcolor}{rgb}{0.000000, 0.000000, 0.000000}
\pgfsetfillcolor{diafillcolor}
\pgfsetfillopacity{1.000000}
\definecolor{dialinecolor}{rgb}{0.000000, 0.000000, 0.000000}
\pgfsetstrokecolor{dialinecolor}
\pgfsetstrokeopacity{1.000000}
\pgfpathmoveto{\pgfpoint{41.450217\du}{18.000165\du}}
\pgfpatharc{487}{334}{7.332817\du and 7.332817\du}
\pgfusepath{stroke}
}
\pgfsetlinewidth{0.100000\du}
\pgfsetdash{}{0pt}
\pgfsetbuttcap
\definecolor{dialinecolor}{rgb}{0.000000, 0.000000, 0.000000}
\pgfsetstrokecolor{dialinecolor}
\pgfsetstrokeopacity{1.000000}
\pgfpathmoveto{\pgfpoint{31.516724\du}{9.225577\du}}
\pgfpatharc{165}{52}{2.240625\du and 2.240625\du}
\pgfusepath{stroke}
\pgfsetlinewidth{0.100000\du}
\pgfsetdash{}{0pt}
\pgfsetbuttcap
\definecolor{dialinecolor}{rgb}{0.000000, 0.000000, 0.000000}
\pgfsetstrokecolor{dialinecolor}
\pgfsetstrokeopacity{1.000000}
\pgfpathmoveto{\pgfpoint{34.216763\du}{10.825712\du}}
\pgfpatharc{339}{248}{1.897349\du and 1.897349\du}
\pgfusepath{stroke}
\pgfsetlinewidth{0.100000\du}
\pgfsetdash{}{0pt}
\pgfsetbuttcap
\definecolor{dialinecolor}{rgb}{0.000000, 0.000000, 0.000000}
\pgfsetstrokecolor{dialinecolor}
\pgfsetstrokeopacity{1.000000}
\pgfpathmoveto{\pgfpoint{39.994699\du}{6.538553\du}}
\pgfpatharc{165}{52}{2.240625\du and 2.240625\du}
\pgfusepath{stroke}
\pgfsetlinewidth{0.100000\du}
\pgfsetdash{}{0pt}
\pgfsetbuttcap
\definecolor{dialinecolor}{rgb}{0.000000, 0.000000, 0.000000}
\pgfsetstrokecolor{dialinecolor}
\pgfsetstrokeopacity{1.000000}
\pgfpathmoveto{\pgfpoint{42.694738\du}{8.138687\du}}
\pgfpatharc{339}{248}{1.897349\du and 1.897349\du}
\pgfusepath{stroke}
\pgfsetlinewidth{0.100000\du}
\pgfsetdash{}{0pt}
\pgfsetbuttcap
\definecolor{dialinecolor}{rgb}{0.000000, 0.000000, 0.000000}
\pgfsetstrokecolor{dialinecolor}
\pgfsetstrokeopacity{1.000000}
\pgfpathmoveto{\pgfpoint{47.409699\du}{8.688553\du}}
\pgfpatharc{165}{52}{2.240625\du and 2.240625\du}
\pgfusepath{stroke}
\pgfsetlinewidth{0.100000\du}
\pgfsetdash{}{0pt}
\pgfsetbuttcap
\definecolor{dialinecolor}{rgb}{0.000000, 0.000000, 0.000000}
\pgfsetstrokecolor{dialinecolor}
\pgfsetstrokeopacity{1.000000}
\pgfpathmoveto{\pgfpoint{50.109738\du}{10.288687\du}}
\pgfpatharc{339}{248}{1.897349\du and 1.897349\du}
\pgfusepath{stroke}
\end{tikzpicture}
\end{center}
\caption{A Riemann surface $C$ of genus $g$ develops a non-separating node. The normalization $\tilde{C}$ is a Riemann surface of genus $g-1$ }
\end{figure}

Finally, let us consider the case where $C$ has a non-separating node. Here, the normalization of the nodal curve, $v: \tilde{C}\to C$, is a curve of genus $g_{\tilde{C}}=g-1$ with two points, $(q_1, q_2)$ covering the node $p$.

As before, define
\begin{equation}
\tilde{\mathcal{L}}_k= \mathcal{L}_k\otimes \mathcal{O}(-q_1-q_2)
\end{equation}
The degree
\begin{equation*}
\deg{\tilde{\mathcal{L}}_k}= 2k g_{\tilde{C}}+\sum_{p_i\in D}(k- \chi_k^{(i)})
\end{equation*}
is in the stable range so that $H^1(\tilde{C},\tilde{\mathcal{L}}_k)=0$. (This also means that $H^1(C,\mathcal{L}_k)=0$.)
The Hitchin base for the normalization is
\begin{equation}
\tilde{B} = \bigoplus_k H^0(\tilde{C},\tilde{\mathcal{L}}_k)
\end{equation}
which has graded dimensions
\begin{equation}\label{gradednonseparatingdims}
\tilde{b}_k = (g-2)(2k-1)  +2(k-1)+\sum_{p_i\in D}(k- \chi_k^{(i)})
\end{equation}
We have the  natural inclusion, $\tilde{B}\hookrightarrow B$, and (with a slight abuse of notation), we will call the quotient
$$
B_C= B/\tilde{B}
$$
the space of center parameters, as before.

Combining \eqref{gradedBdims} and \eqref{gradednonseparatingdims}, we have
\begin{equation}
b_k^C = b_k - \tilde{b}_k = 1
\end{equation}
for each $k$. So a non-separating node is always the standard node and $\chi_k=1$ for each $k$.

\subsection{Further degeneration}\label{further_degen}
\label{furtherdegeneration}

So far, we have assumed that the curve $C$ is smooth, except for a single node. We have seen that \emph{only when} the nodal curve has two components and when at least one of those components (say, $C_L$)  is genus-zero is it possible for the node to be restricted: $(O,H)\neq ([N],SU(N))$.

Without loss of generality we can take $C_L$ to be the genus-zero component, and the twists $n^R_k=0$, so that the vanishing order, $\chi_k$, at the node is determined by the twists $n^L_k$. As we shall see, in \S\ref{classifying}, the collection of $\chi_k$ determine the nilpotent $O$ at the node and moreover that if, for any $k$, $n^L_k\neq 0$ the \emph{all} of the $n^R_k=0$ (and vice versa). So $O$ is (in fact) entirely determined by the data on $C_L$. With one exception noted in \S\ref{restricted_nodes}, $H$ is also entirely determined by the data on $C_L$. 
Here, we will not assume this. We will study what happens for a fixed value of $k$. 
Hence we will use the pair $(\chi_k,b_k^C)$ as our stand-in for $(O,H)$.

In this subsection, we would like to inquire what happens upon further degeneration of $C_L$ and $C_R$. Let us first consider degenerating $C_L$. Since it has genus-zero, the degeneration has the form of a \emph{tree} of $\mathbb{P}^1$s, with the root node being the one we started with. For the present discussion, we only need to focus on the vanishing orders ${\color{red}\chi_k}$ at the root of the tree and graded dimension $b_k^C$ of the space of center parameters at the root node. 

First, we wish to show that all such trees give the same ${\color{red}\chi_k}$ at the root node. To see this, it suffices to note that we can pass from one tree to any other tree, via an elementary move on 4-punctured spheres. This is depicted in figure \ref{treemove}.

\begin{figure}
\begin{center}
\begin{tikzpicture}
\begin{scope}[shift={(0,.25)}]
\draw [thick] (5.25,-.75) arc [radius=3, start angle=45, end angle= 135];
\draw[thick]  (-.4,2.25) -- (3.6,-.25);
\draw[thick]  (2.25,2.25) -- (.75,.75);
\draw[fill] (0,2) circle [radius=0.07];
\draw[fill] (2,2) circle [radius=0.07];
\draw[fill] (1,1) circle [radius=0.07];
\draw[fill=blue,color=blue] (1.225,1.225) circle [radius=0.07];
\draw[fill=red,color=red] (3,.125) circle [radius=0.07];
\node[scale=0.9] at (.125,2.25){$3$};
\node[scale=0.9] at (1.825,2.25){$1$};
\node[scale=0.9] at (.75,1.125){$2$};
\node[scale=0.7] at (2.625,-.175){$(O,H)$};
\end{scope}
\node at (6,1){$\sim$};
\begin{scope}[shift={(7.75,0)}]
\draw [thick] (3,-.5) arc [radius=3, start angle=45, end angle= 135];
\draw[thick]  (-.5,2) -- (2.5,2);
\draw[thick]  (1,2.5) -- (1,-.25);
\draw[fill] (0,2) circle [radius=0.07];
\draw[fill] (2,2) circle [radius=0.07];
\draw[fill] (1,1) circle [radius=0.07];
\draw[fill=blue,color=blue] (1,2) circle [radius=0.07];
\draw[fill=red,color=red] (1,.375) circle [radius=0.07];
\node[scale=0.9] at (0,2.25){$3$};
\node[scale=0.9] at (2,2.25){$1$};
\node[scale=0.9] at (.75,1){$2$};
\node[scale=0.7] at (.5,.125){$(O,H)$};
\end{scope}
\node at (11.25,1){$\sim$};
\begin{scope}[shift={(14.5,.25)}]
\draw [thick] (1.25,-.75) arc [radius=3, start angle=45, end angle= 135];
\draw[thick]  (-.25,2.25) -- (1.25,.75);
\draw[thick]  (2.4,2.25) -- (-1.6,-.25);
\draw[fill] (0,2) circle [radius=0.07];
\draw[fill] (2,2) circle [radius=0.07];
\draw[fill] (1,1) circle [radius=0.07];
\draw[fill=blue,color=blue] (.775,1.225) circle [radius=0.07];
\draw[fill=red,color=red] (-1,.125) circle [radius=0.07];
\node[scale=0.9] at (.125,2.25){$3$};
\node[scale=0.9] at (1.825,2.25){$1$};
\node[scale=0.9] at (1.25,1.125){$2$};
\node[scale=0.7] at (-.625,-.125){$(O,H)$};
\end{scope}
\end{tikzpicture}
\caption{The basic ``tree move''. Fixing the locations of the three points, $1,2,3$, on $\mathbb{P}^1$ and letting the attachment point $\color{red}\bullet$ roam over the $\mathbb{P}^1$, we interpolate between the three trees.}\label{treemove}
\end{center}
\end{figure}
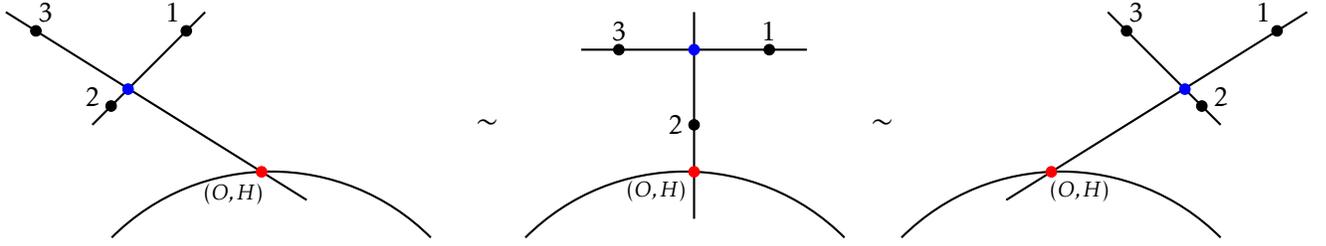

Let's compute the vanishing order $\color{red}\chi_k$ at the root of the tree on the left. Here, ``1'' and ``2'' fuse first, yielding the vanishing order
$$
{\color{blue}\tilde{\chi}_k}= 1+\max(0,\chi^{(1)}_k+\chi^{(2)}_k-k-1)
$$
at the node labeled in blue. This then combines with ``3'' to yield the vanishing order
\begin{equation}
\begin{split}
{\color{red}\chi_k}&=1+\max(0,{\color{blue}\tilde{\chi}_k}+\chi_k^{(3)}-k-1)\\
&=1+\max\bigl(0,\max(0,\chi_k^{(1)}+\chi_k^{(2)}-k-1)+\chi_k^{(3)}-k\bigr)
\end{split}
\end{equation}
at the root of the tree.
There are two cases to consider
\begin{itemize}
\item[a)] If $\chi_k^{(1)}+\chi_k^{(2)}\leq k$, then we have \[
\begin{split}
{\color{red}\chi_k}&=1+\max(0,\chi_k^{(3)}-k)\\
&=1
\end{split}
\]
\item[b)] If $\chi_k^{(1)}+\chi_k^{(2)}\geq k+1$, then we have \[
{\color{red}\chi_k}=1+\max(0,\chi_k^{(1)}+\chi_k^{(2)}+\chi_k^{(3)}-2k-1)
\]
\end{itemize}
In the latter case, if $\chi_k^{(1)}+\chi_k^{(3)}\leq k$ \emph{or} $\chi_k^{(2)}+\chi_k^{(3)}\leq k$ then we again have ${\color{red}\chi_k}=1$. Of course, it also gives ${\color{red}\chi_k}=1$ for $\chi_k^{(1)}+\chi_k^{(2)}\leq k$. Hence this formula subsumes case (a) and we can write
\begin{equation}\label{chitree}
{\color{red}\chi_k}=
1+\max(0,\chi_k^{(1)}+\chi_k^{(2)}+\chi_k^{(3)}-2k-1)
\end{equation}
This result is clearly invariant under permutations of 1,2,3 and hence applies to all three trees. Thus, we get the same ${\color{red}\chi_k}$ at the root of the tree, regardless of which tree we choose. Performing this move \emph{locally}, we can transform an \emph{arbitrary} tree of $\mathbb{P}^1$s into any other tree. Thus ${\color{red}\chi_k}$ at the root of the tree depends only on the $\chi_k^{(i)}$ on $C_L$, and not on how $C_L$ degenerates.

We can perform the same analysis for ${\color{red}b^C_k}$. With the exception  noted in  \S\ref{restricted_nodes}, the number of center parameters ${\color{red}b^C_k}$ is determined by the data on $C_L$. We have
\begin{equation}
{\color{red}b^C_k}=\min\bigl(1,\max(0,d^L_k+1)\bigr)
\end{equation}
where, in the case at hand,
\begin{equation*}
d_k^L=-k+\sum_{p_i\in D_L} (k-\chi_k^{(i)})
\end{equation*}
Applying this to the tree on the left,
\begin{equation*}
\begin{split}
1+d^L_k &= 1+k-{\color{blue}\tilde{\chi}_k}-\chi_k^{(3)}\\
&=k-\chi_k^{(3)}-\max(0,\chi_k^{(1)}+\chi_k^{(2)}-k-1)
\end{split}
\end{equation*}
Again, we have two cases
\begin{itemize}
\item[a)] If $\chi_k^{(1)}+\chi_k^{(2)}\leq k$, then $1+d^L_k= k-\chi_k^{(3)}$ and hence ${\color{red}b^C_k}=1$.
\item[b)] If $\chi_k^{(1)}+\chi_k^{(2)}\geq k+1$, then $1+d^L_k=2k+1-\chi_k^{(1)}-\chi_k^{(2)}-\chi_k^{(3)}$ and hence \[{\color{red}b^C_k}=\min\bigl(1,\max(0,2k+1-\chi_k^{(1)}-\chi_k^{(2)}-\chi_k^{(3)})\bigr).\]
\end{itemize}
In the latter case, if \emph{either} $\chi_k^{(1)}+\chi_k^{(3)}\leq k$ or $\chi_k^{(2)}+\chi_k^{(3)}\leq k$, then $2k+1-\chi_k^{(1)}-\chi_k^{(2)}-\chi_k^{(3)} \geq 2$ and hence ${\color{red}b^C_k}=1$. As before, this formula subsumes case (a) and we can write
\begin{equation}\label{bctree}
{\color{red}b^C_k}=
\min\bigl(1,\max(0,2k+1-\chi_k^{(1)}-\chi_k^{(2)}-\chi_k^{(3)})\bigr)
\end{equation}
As with \eqref{chitree}, this is manifestly symmetric under permutations of 1,2,3 and hence invariant under the ``tree move''. Thus the pair $({\color{red}\chi_k},{\color{red}b^C_k})$ at the root of the tree depends only on the $O_i$ on $C_L$, and not on how $C_L$ degenerates.

Equations \eqref{chitree},\eqref{bctree} ensure that $({\color{red}\chi_k},{\color{red}b^C_k})$  at the node is independent of how $C_L$ further degenerates.  The upshot is that $({\color{red}\chi_k},{\color{red}b^C_k})$  depends only on the puncture data and is completely independent of the complex structure of $C_L$. In fact, we can read off the general answer for an arbitrary number of marked points on a genus-zero $C_L$ from \eqref{virtdimLR}, \eqref{modifiedvanishing}, \eqref{twistdef}, \eqref{btotal} and \eqref{truedimLR}:
\begin{subequations}\label{final}
\begin{equation}\label{finalchi}
\chi_k=1+\max\bigl(0,k-1-\sum_{p_i\in D_L} (k-\chi_k^{(i)})\bigr)
\end{equation}
\begin{equation}\label{finalbc}
b^C_k=\min\Bigl(1,\max\bigl(0,1-k+\sum_{p_i\in D_L} (k-\chi_k^{(i)})\bigr)\Bigr)
\end{equation}
independent of the complex structure of $C_L$ and, in particular, of whether $C_L$ is smooth or degenerate. While we arrived at \eqref{final} through manipulations which preserved the number of components of $C_L$, the final result is independent of the complex structure of $C_L$ and holds for both $C_L$ smooth or nodal, with an arbitrary number of irreducible components.

As alluded to above, there is one case where a contribution to $b^C_k$ comes from both $C_L$ and $C_R$. This occurs (see \S\ref{restricted_nodes}) only when both $C_L$ and $C_R$ are genus-zero. In that case, we can replace \eqref{finalbc} by
\begin{equation}\label{finalbcmode}
b^C_k=\min\Bigl(1,\max\bigl(0,1-k+\sum_{p_i\in D_L} (k-\chi_k^{(i)})\bigr),
\max\bigl(0,1-k+\sum_{p_i\in D_R} (k-\chi_k^{(i)})\bigr)\Bigr)
\end{equation}
\end{subequations}

Similar considerations apply to degenerations of $C_R$. If $g_R\geq1$, then there are no constraints coming from $C_R$, whether smooth or degenerate. If $g_R=0$, then  our assumption that $H^1(C,\mathcal{L}_k)=0$  implies that the only possible constraint coming from $C_R$ is embodied in \eqref{finalbcmode}, and (repeating the arguments of this subsection) this persists under further degenerations of $C_R$.

\subsection{The global story}\label{globalstory}

In \S\ref{hitchinreducible}, we sketched the construction of a family of Hitchin bases that extended to the boundary of the moduli space, where the curve $C$ develops a node. Here we will sketch the general story, leaving most of the details to a followup \cite{Balasubramanian:global}.

For simplicity, let us first consider the genus-0 case, where all nodes are separating nodes. The components of the boundary of $\overline{\mathcal{M}}_{0,n}$ are labeled by subsets $S\subset \{p_1,p_2,\dots,p_n\}$ such that both $S$ and its complement, $S^\vee$ contain at least two points\footnote{Since there is no invariant distinction between left and right, exchanging $S\leftrightarrow S^\vee$ yields the same component of the boundary.}. The corresponding nodal curve $C$ has two irreducible components $C_{S}$ and $C_{S^\vee}$ (previously, we called these $C_L$ and $C_R$), such that the marked points in $S$ lie on $C_{S}$ and the marked points in $S^\vee$ lie on $C_{S^\vee}$.

For each $S$ and $k=2,\dots,N$ we assign a non-negative integer, $n_k^S$ as follows. Let $\mathcal{L}_k\to \overline{\mathcal{C}}$ be the Hitchin line bundle over the universal curve, obtained by fitting together the line bundles \eqref{Lkdef} over each fiber. As in \eqref{h1equation},
\begin{equation}
\label{nksdefgenuszero}
n_k^S\coloneqq\max\Bigl(0,k-1-\sum_{p_i\in S}(k-\chi_k^{(i)})\Bigr)\quad.
\end{equation}
When $C_S$ is a smooth genus-0 curve, $h^1(C_S, \mathcal{L}_k)= n_k^S$.

Let $\mathcal{C}_S$ be the Cartier divisor in $\overline{\mathcal{C}}$ corresponding to $C_S$.
We define the line bundle $\mathcal{L}'_k\to \overline{\mathcal{C}}$ to be
\begin{equation}\label{globallprimedef}
\mathcal{L}'_k\coloneqq \mathcal{L}_k\otimes \mathcal{O}\bigl(-\sum_S n_k^S \mathcal{C}_S\bigr)
\end{equation}
The family of Hitchin bases
\begin{equation}\label{hitchinbaseglobaldef}
\mathcal{B} =\bigoplus_k \pi_* \mathcal{L}'_k
\end{equation}
is a vector bundle\footnote{The astute reader will note that this is a generalization of the procedure developed in \S\ref{directimages} for the case of $\overline{\mathcal{M}}_{0,4}$. The twist \eqref{globallprimedef} is a direct generalization of \eqref{substituteL}. } over $\overline{\mathcal{M}}_{0,n}$.

The generalization to higher genus is straightforward. The Deligne-Mumford compactification, $\overline{\mathcal{M}}_{g,n}$, now contains boundary components corresponding to both separating and non-separating nodes. At the former, $C$ is a reducible curve $C=C_S \cup C_{S^\vee}$, where $C_S$ has genus $g_{C_S}$ and contains a subset $S\subset\{p_1,p_2,\dots,p_n\}$ of the marked points\footnote{For stability, $S$ must contain at least two points  when $g_{C_S}=0$. When $g_{C_S}>0$, there's no condition on the number of points in $S$. Similarly, for $g-g_{C_S}>0$, there's no condition on the number of points in $S^\vee$.}
and $C_{S^\vee}$ has genus $g-g_{C_S}$ and contains the complementary set of marked points.
We define 
\begin{equation}\label{nksdef}
n_k^S\coloneqq\max\Bigl(0,k-1- g_{C_S}(2k-1)-\sum_{p_i\in S}(k-\chi_k^{(i)})\Bigr)
\end{equation}
such that, when $C_S$ is smooth, $h^1(C_S, \mathcal{L}_k)= n_k^S$. For $g_{C_S} >0$ and $C_S$ smooth, we are in the stable range, where $h^1(C_S, \mathcal{L}_k)= n_k^S=0$. Similarly as we saw in \S\ref{24_hitchin_system_on_an_irreducible_nodal_curve}, when $C$ is smooth except for a non-separating node,  we also have $h^1( \mathcal{L}_k)=0$. So we might as well restrict ourselves to the case where $C_S$ has genus-0.
For each such $S$, let $\mathcal{C}_S$ be the Cartier divisor in $\overline{\mathcal{C}}$ corresponding to $C_S$. As in the genus-0 case, $\mathcal{L}'_k$ is defined by \eqref{nksdef}, \eqref{globallprimedef}. In Appendix \ref{Proof}, we prove
(see also  \cite{Balasubramanian:global}):

\begin{theorem}\label{hitchinglobalthm}
The family of Hitchin bases
\begin{equation}
\mathcal{B} \coloneqq \bigoplus_{k=2}^N \pi_* \mathcal{L}'_k
\end{equation}
is a vector bundle over $\overline{\mathcal{M}}_{g,n}$.
\end{theorem}

\section{Flavour Considerations and the Higgs Branch}
\label{flavour}

We have argued for the labeling of restricted nodes by a pair of the form $(O,H)$, where $O$ is a (Hitchin) nilpotent orbit in $\mathfrak{j}=\mathfrak{sl}(N)$ and the center parameters associated to the node are the Casimirs of the compact simple Lie group $H\subset SU(N)$. From a purely Hitchin system/Higgs bundles standpoint, it is not obvious why the center parameters should be the Casimirs of a subgroup (in particular a \textit{simple} subgroup) $H\subset J$. It is equally unclear which simple subgroups can arise in this way.

The purpose of this section is to use Higgs branch considerations to answer the question ``What are the possible pairs $(O,H)$ that could conceivably arise in the nodal limit?'' In  \S\ref{classifying}, we will use the results on Higgs bundles on nodal curves that we obtained in \S\ref{Hitchinnodal} to study the same question from a purely Coulomb branch point of view.

In the physics, the meaning of $H$ is clear. In the nodal limit, the SCFT becomes a weakly coupled gauge theory, with gauge group $H$, where the symmetry that is gauged is an $H$ subgroup of the flavour symmetry group of the SCFT associated to the normalization of the nodal curve $C$. This flavour symmetry group is the group of hyperK\"ahler isometries of the Higgs branch of that SCFT. 

\subsection{Flavour symmetry}\label{flavorconsiderations}

To see the flavour symmetry, it is more natural to consider the ``Nahm'' nilpotent orbit $O_N$, rather than the ``Hitchin'' nilpotent orbit $O_H$ which is the residue of $\Phi(z)$.

In type-A, the Nahm partition is just the transpose of the Hitchin partition. For types D and E, the map between $O_N$ and $O_H$ is more nontrivial \cite{Chacaltana:2012zy,Balasubramanian:2014jca}. In type-A, all nilpotent orbits are special (the map between Nahm and Hitchin orbits, given by the transpose, is an involution), beyond type-A there are non-special orbits.  When $O_N$ is non-special, the image on the Hitchin side is a pair $(O_H,\Gamma)$ where $O_H$ is special and the finite group $\Gamma$ is a subgroup, $\Gamma\subset\overline{A}(O_H)$, of Lusztig's canonical quotient \cite{MR742472} of the group $A(O_H)$ defined in \eqref{AOdef}. That is, the Hitchin system data is enriched by a finite group associated to each non-special Nahm orbit at a puncture.

For this reason, the physicists prefer to label punctures by their Nahm nilpotent orbit. Since, in this paper, we have restricted ourselves to type-A, where all nilpotent orbits are special, we have labeled punctures by their Hitchin nilpotent orbit.

By Jacobson-Morozov, each choice of nilpotent, X, corresponds to a distinguished triple --- an embedding $\rho\colon \mathfrak{sl}(2)\hookrightarrow \mathfrak{j}$. A nilpotent orbit $O\ni X$ thus corresponds to an $\mathfrak{sl}(2)$ embedding up to conjugacy. For $A_{N-1}$, such an embedding up-to-conjugacy defines a 
partition of $N$ which we denote by $[q_1^{n_1},q_2^{n_2},\dots]$ where $q_1\geq q_2\geq\dots$ and $\sum_i n_i q_i =N$. Associated to this partition is an $N$-box Young diagram with $n_1$ columns of height $q_1$, $n_2$ columns of height $q_2$, etc.

Let $\mathfrak{f}\subset \mathfrak{j}$ be the subalgebra that centralizes the embedding $\rho$ (i.e.~fixes every element of $im(\rho)\subset\mathfrak{j}$) corresponding to a given \emph{Nahm} nilpotent, $X\in O_N$. We denote the flavour symmetry $F$ to be the corresponding compact Lie group (which depends only on the orbit, $O_N$). For $\mathfrak{j}=A_{N-1}$,
\begin{displaymath}
F=S\left(\prod_i U(n_i)\right)
\end{displaymath}

For each simple subgroup, $F_i=SU(n_i)\subset F$, we assign a level $k_i\in \mathbb{N}$, as follows. Decompose $\mathfrak{j}$ under $\mathfrak{sl}(2)\times \mathfrak{f}_i$ as $\mathfrak{j}= \oplus_n V_n\otimes R_{i,n}$ where $V_n$ is the $n$-dimensional irrep of $\mathfrak{sl}(2)$ and $R_{i,n}$ is a (possibly reducible) representation of $\mathfrak{f}_i$. Let $l_{i,n}$ be the index\footnote{For a highest-weight representation $R$ with highest weight $\lambda$, $l(R) = \frac{\dim(R)}{\dim(adj)} (\lambda,\lambda+2\delta)$, where $\delta$ is the Weyl vector. The normalization is such that the defining representation of $SU(N)$ has $l=1$.} of $R_{i,n}$. Then

\begin{displaymath}
k_i= \sum_n l_{i,n}
\end{displaymath}
For $A_{N-1}$, the level $k_i$ of $SU(n_i)$ is just twice the total number of boxes in the first $q_i$ rows of the Young diagram. For example consider $O_N=[3^2,2^2]$. This has $F=S(U(2)^2)\sim SU(2)^2\times U(1)$ (where we ignore a discrete quotient), and the two $SU(2)$s have levels $k=20$ and $k=16$.

For later reference, we define the complementary level, for any simple $H\subset F_i$

\begin{equation}
k'_i(H) = 4h^\vee(H) - k_i
\end{equation}
where $h^\vee$ is the dual Coxeter number.

\subsection{Restricted nodes}\label{restricted_nodes}

A restricted node is a pair, $(O,H)$, consisting of a Hitchin nilpotent orbit and a simple ($SU(l)$ or $Sp(l)$) subgroup, $H\subset F_i\subset F_O$, of its flavour symmetry group. We could denote the standard node as $([N],SU(N))$ but will refrain, so as not to unduly clutter the notation.

With the preliminaries of \S\ref{flavorconsiderations},  we can state the algorithm (which was first used in \cite{Chacaltana:2010ks}, though the differences of notation would make that hard to discern) for determining the restricted nodes that can appear.

For any given $O$, the allowed $H$s are those simple subgroups $H\subset F_i$ for which the complementary level is non-negative\footnote{ $k(H)$ and $k'(H)$ are the levels of the current algebras for the $H$-flavour symmetry of the SCFTs associated to $C_R$ and $C_L$, respectively. The vanishing of the $\beta$-function for $H$ requires $k(H)+ k'(H) - 4h^\vee(H) = 0$, where $k(H)$ and $k'(H)$ are the contributions to the $\beta$-function from the ``matter'' sectors while $-4h^\vee$ is the contribution from the vector multiplet for gauge group $H$. Unitarity of the SCFTs requires $k(H),k'(H)\geq 0$.}, $k'(H) \geq 0$. 

For example, given $O=[N]$, the allowed $H$s are

\begin{displaymath}
H=
\begin{cases}[1.5]
SU(l),&2N\geq2l\geq N\\
Sp(l),&N\geq2l\geq N-2
\end{cases}
\end{displaymath}

More generally, $H$ can arise either as
\begin{itemize}
\item an $SU(l)$ or $Sp(l)$ subgroup of the $SU(n)$ associated to a Nahm partition of the form $[...,1^n]$, with $l$ large enough so that $k'(H)\geq 0$ or
\item the $SU(n)$ associated to the Nahm partition $[2^n]$.
\end{itemize}
By Theorem \ref{ohtheorem}(C), the latter case does not occur in the untwisted $A_{2n-1}$ theory. It does, however arise in the collision of punctures from the twisted sector \cite{Chacaltana:2012ch}\footnote{\label{SOodd}The astute reader might object that there appear to be two more possibilities for $(O,H_{k'})$ that are not on this list. You might think that $([2n],SO(2n-1)_0)$ (for $J=SU(2n)$)  or $([2n-1],SO(2n-1)_2)$ (for $J=SU(2n-1)$) are allowed by the $k'(H)\geq 0$ condition. As we shall prove in Prop \ref{alternatingproposition} of \S\ref{classifying}, there's a unique case where the Casimirs $\vec{b}^C=(1,0,1,0,\dots,1)$ arise at the node: namely, when $C_L$ contains two marked points with Hitchin partitions $[2^n],[2^n]$. But that uniformly leads to the theory on $C_L$ being two hypermultiplets in the defining representation of $Sp(n)$  --- yielding $(O,H_{k'})=([2n],Sp(n)_4)$. There's no collision of punctures that yields $\vec{b}^C=(1,0,1,0,\dots,1)$ and an \emph{empty} theory on $C_L$. For $([2n-1],SO(2n-1)_2)$, the story is even simpler. For $n>3$, there's a lower bound on the level of an $SO(2n-1)$ current algebra in a unitary $\mathcal{N}=2$ SCFT. This lower bound is $k'\geq 4$ and is saturated by a free hypermultiplet in the vector representation. So there's no candidate for the theory on $C_L$.}.

In the former case, let the Hitchin partition be $O=[p_1,p_2,\dots]$, with $p_1-p_2=m$. Then
\begin{equation}\label{row1}
H=
\begin{cases}
SU(l),&2m\geq2l\geq p_1\\
Sp(l),&m\geq2l\geq p_1-2
\end{cases}
\end{equation}
In particular, for \eqref{row1} to have solutions, we must have $p_1\geq 2\max(p_2,1)$.

The next concept we need to introduced is the partial-ordering on the set of nilpotent orbits, induced by orbit-closure. This ordering is typically captured by the Hasse diagram. Here is the Hasse diagram for $A_5$:

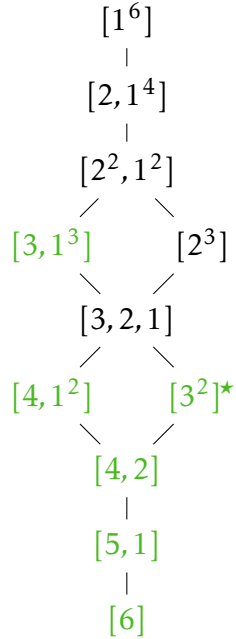
\begin{center}
\begin{tikzpicture}
\node[style={armygreen}] at (0,0) (a) {$[6]$};
\node[style={armygreen}] at (0,1) (b) {$[5,1]$};
\node[style={armygreen}] at (0,2) (c) {$[4,2]$};
\node[style={armygreen}] at (-1,3) (d1) {$[4,1^2]$};
\node[style={armygreen}] at (1,3) (d2) {$[3^2]^\star$};
\node at (0,4) (e) {$[3,2,1]$};
\node[style={armygreen}] at (-1,5) (f1) {$[3,1^3]$};
\node at (1,5) (f2) {$[2^3]$};
\node at (0,6) (g) {$[2^2,1^2]$};
\node at (0,7) (h) {$[2,1^4]$};
\node at (0,8) (i) {$[1^6]$};
\draw[-] (a) to (b);
\draw[-] (b) to (c);
\draw[-] (c) to (d1);
\draw[-] (c) to (d2);
\draw[-] (d1) to (e);
\draw[-] (d2) to (e);
\draw[-] (e) to (f1);
\draw[-] (e) to (f2);
\draw[-] (f1) to (g);
\draw[-] (f2) to (g);
\draw[-] (g) to (h);
\draw[-] (h) to (i);
\end{tikzpicture}
\end{center}
where we have denoted in green the possible $O$s, whose flavour symmetry group admits a subgroup $H$ satisfying $k'(H) \geq 0$. 

This leads to 13 possible\footnote{In the table, $(\pi_O)_k= k-(\chi_O)_k$ as usual. 
$\pi'_k = 2k-1-\pi_k -b^C_k$ and the twisting $n_k=\max(\pi'_k-k,0)$. $(n_h,n_v)$ are the contributions from the branch of the node to the effective number of hypermultiplets and the effective number of vector multiplets for the SCFT associated to $C_L$. Note, for instance, that $n_h$ depends only on $O$ and not on $H$.} restricted nodes for $A_5$:

\begin{center}
\begin{tabular}{c|c|c|l}
$(O,H_{k'})$&$\vec{\pi}_O$&$\vec{\pi}'_{O}$&$(n_h,n_v)$\\
\hline
$([6],{SU(5)}_8)$&(1, 2, 3, 4, 5)&(1, 2, 3, 4, 6)&(140, 136)\\
$([6],{Sp(3)}_4)$&(1, 2, 3, 4, 5)&(1, 3, 3, 5, 5)&(140, 139)\\
$([6],{SU(4)}_4)$&(1, 2, 3, 4, 5)&(1, 2, 3, 5, 6)&(140, 145)\\
$([6],{Sp(2)}_0)$&(1, 2, 3, 4, 5)&(1, 3, 3, 5, 6)&(140, 150)\\
$([6],{SU(3)}_0)$&(1, 2, 3, 4, 5)&(1, 2, 4, 5, 6)&(140, 152)\\
$([5,1],{SU(4)}_6)$&(1, 2, 3, 4, 4)&(1, 2, 3, 5, 7)&(156, 156)\\
$([5,1],{Sp(2)}_2)$&(1, 2, 3, 4, 4)&(1, 3, 3, 5, 7)&(156, 161)\\
$([5,1],{SU(3)}_2)$&(1, 2, 3, 4, 4)&(1, 2, 4, 5, 7)&(156, 163)\\
$([4,2],{SU(2)}_0)$&(1, 2, 3, 3, 4)&(1, 3, 4, 6, 7)&(168, 177)\\
$^\star([3^2],{SU(3)}_0)$&(1, 2, 2, 3, 4)&(1, 2, 5, 6, 7)&(176, 179)\\
$([4,1^2],{SU(3)}_4)$&(1, 2, 3, 3, 3)&(1, 2, 4, 6, 8)&(180, 183)\\
$([4,1^2],{SU(2)}_0)$&(1, 2, 3, 3, 3)&(1, 3, 4, 6, 8)&(180, 188)\\
$([3,1^3],{SU(2)}_2)$&(1, 2, 2, 2, 2)&(1, 3, 5, 7, 9)&(210, 215)\\
\end{tabular}
\end{center}

Which restricted node occurs at a given degeneration of $C$ can now be summarized by the following algorithm.

\hypertarget{algorithm}{}\begin{enumerate}
\item Let $\vec{\pi}_O$ be the vector of pole orders, corresponding to the Hitchin nilpotent orbit $O$ (recall that these are related to the $\vec{\chi}_O$ by $(\pi_O)_k=k-(\chi_O)_k$). For the regular Hitchin nilpotent, $\vec{\pi}_{[N]}=(1,2,3,\dots,N-1)$ for $k=2,3,\dots,N$.

\item Consider a separating node, where punctures $O_1,\dots O_n$ appear on the (genus-$0$) curve on the left. Form the vector $\vec{\pi} = \sum_{i=1}^n \vec{\pi}_{O_i}$.

\item Among the allowed $(O,H)$, find the \emph{largest} $O$ (the one lowest on the Hasse diagram) such that

\begin{enumerate}%
\item $\vec{\pi}-\vec{\pi}_O$ has only non-negative entries.
\item $H\subset F_{O}$ is the highest-rank simple subgroup of the flavour symmetry of $O$ whose independent Casimirs correspond to a subset of the positive entries of $\vec{\pi}-\vec{\pi}_O$. For a generic curve (and collection of punctures) on the right, these are the center parameters.
\item If $\vec{\pi}-\vec{\pi}_{[N]}$ has all positive entries, then the node is the standard node.
\end{enumerate}
\end{enumerate}
With one exception, this pair $(O,H)$ is the restricted node. The exception occurs when both the left and right components of the nodal curve impose such a restriction. In type-A, this occurs when $C$ has genus-0 in the $A_{2n-1}$ theory. If the Hitchin nilpotents at the punctures consist of two copies of $[2^n]$ and some number of ``hook partitions" of the form $[l_i+1, 1^{2n-l_i-1}]$ with
\begin{equation}
l_i\geq 1,\qquad \sum_i l_i = n-1
\end{equation}
then, when $C$ degenerates in such a way that the two $[2^n]$ punctures end up on one component (w.l.o.g, $C_R$) and the hook partitions end up on the other component ($C_L$), then $C_L$  picks out the restricted node $([2n],SU(2n-1))$ and $C_R$ picks out the restricted node $([2n],Sp(n))$. The actual restricted node is $([2n],Sp(n-1))$. I.e., $H=H_L\cap H_R$. We saw an example of this in \hyperlink{example_3}{example 3} of \S\ref{examples}.

For the purposes of this paper, we will refer to the above algorithm to find $(O,H)$ as the \textit{Higgs-Coulomb} algorithm due to the fact that the algorithm actually involves properties of \textit{both} the Higgs and Coulomb branches. We choose this terminology primarily to distinguish it from the discussion in \S\ref{classifying} where the Higgs branch does not play any role. There is an alternative proposal to find $H$ purely from the Higgs branch geometry due to \cite{Gaiotto:2011xs}. We comment on the relationship between our work and this proposal in \S\ref{compatibility}.

Note that, since the Hasse ordering is only a partial-ordering, one might worry that the procedure for selecting $O$ is ambiguous. For instance, in the $A_5$ theory, might we be unable to choose between $[4,1^2]$ and $[3^2]$? Fortunately, this ambiguity never arises. In the case at hand, if $\vec{\pi}-\vec{\pi}_{[4,1^2]}$ \emph{and} $\vec{\pi}-\vec{\pi}_{[3^2]}$ are both non-negative, then so is $\vec{\pi}-\vec{\pi}_{[4,2]}$. More generally, if both $\vec{\pi}-\vec{\pi}_{O_a}$ and $\vec{\pi}-\vec{\pi}_{O_b}$ are non-negative, then either one orbit lies in the closure of the other, or both orbits lie in the closure of $O_c$, which also satisfies $\vec{\pi}-\vec{\pi}_{O_c}$ non-negative.

Carrying out this procedure for $\mathfrak{j}=A_5$, we find that every pair $(O,H)$ except $([3^2],SU(3)_0)$ is indeed realized at the restricted node for some set of defects $O_i$ on the left. We will see in \S\ref{classifying} that one can actually give an \textit{a priori} explanation for why this pair does not occur from a Coulomb branch perspective.

Fairly obviously, this procedure yields the same result for $O$ as the twisting procedure described in \S\ref{nilpnode} (replacing $\mathcal{L}_{k,R}$ by $\mathcal{L}'_{k,R}= \mathcal{L}_{k,R}\otimes \mathcal{O}(-n_k p)$, where $n_k= \max(k-1-\pi_k, 0)$). The procedure for arriving at $H$ seems rather divorced from the cohomological computation of the center parameters in \S\ref{hitchinreducible}. We shall see in \S\ref{compatibility} that these two rather different looking approaches also yield the same answer for $H$. 

\section{Classifying Restricted Nodes}\label{classifying}

In this section, we would like to classify the possible restricted nodes from a purely Hitchin system perspective. Our starting point will be the conditions derived in \S\ref{Hitchinnodal} (equivalently, any of the conditions in Proposition \ref{centervanishinglemma}) for a reduction in center parameters \eqref{bck0cond} and the condition \eqref{h1equation} for the occurrence of non-zero $h^1(\mathcal{L}_k\otimes \mathcal{O}_{C_{L}})$.  

We now state our main results as a theorem and then prove them below.

\begin{theorem}
\label{ohtheorem}
In any tame $SL_N$ Hitchin system that is \textit{OK} in the sense of \S\ref{13_hitchin_systems_corresponding_to_good_theories}, the following statements are true for every restricted node : 

\begin{enumerate}[label=(\Alph*)]
\item The integers $\chi_k = 1 + n_k$ are the vanishing orders corresponding to a nilpotent orbit $O$ in $\mathfrak{sl}(N)$. We call this the Hitchin orbit at the node. 
\item The graded dimension of the center $B_C$ is of the form $b_k^C = (1,1,1, \ldots,1,0,0...,0) $ (or) $b_k^C = (1,0,1,0,1,0\ldots,0,1,0,0,\ldots,0)$, where $k=2,3,4,\ldots N$.
\item Recall that $b^C = \sum_k b_k^C$ and $[p_i]$ are the parts of the Hitchin orbit $O$. The allowed orbits obey $p_1 > 2p_2$ and furthermore, we always have $  p_1  - 2 \leq 2 b^C \leq 2(p_1 - p_2 - 1) $.
\end{enumerate}
\end{theorem}

In the process of proving Th \ref{ohtheorem}, we will show that the following useful proposition also holds. 

\begin{prop}
\label{alternatingproposition}
There is a unique choice of marked points on $C_L$ for which the graded dimension of the center $B_C$ is of the form $b_k^C = (1,0,1,0,\ldots,0,1), k = 2,3,\ldots N$ and $N$ is even. The corresponding $C_L$ is a $\mathbb{P}^1$ with $\text{deg}(D_L)=2$ and the residues of the Higgs field at the two marked points live in the nilpotent conjugacy class $[2^{N/2}]$.
\end{prop}

\subsection{Proof Strategy}
\label{proofstrategy}

First, recall from \S\ref{24_hitchin_system_on_an_irreducible_nodal_curve} that every non-separating node is standard. So, we only need to consider separating nodes to prove Th \ref{ohtheorem}. Let us denote a separating node to be \textit{one sided} if either $h^0(\mathcal{L}_k\otimes \mathcal{O}_{C_{L}}) = 0$ for some values of $k$ or $h^0(\mathcal{L}_{k'}\otimes \mathcal{O}_{C_{R}})=0$ for some values of $k'$ but not both. We will denote a separating node to be \textit{two sided} if both $h^0(\mathcal{L}_k\otimes \mathcal{O}_{C_{L}}) = 0$ and $h^0(\mathcal{L}_{k'}\otimes \mathcal{O}_{C_{R}})=0$ for some values of $(k,k')$ with $k \neq k'$.\footnote{Such a possibility for some $k=k'$ is ruled out by the fact that we are only considering nodal degenerations of OK Hitchin systems.}  It will turn out that a vast majority of restricted nodes arise from one sided separating nodes. For one sided separating nodes, the problem of classifying the allowed nodal degenerations is symmetric between the left and right. So, without loss of generality, we will assume that a one sided separating node has $h^0(\mathcal{L}_k\otimes \mathcal{O}_{C_{L}}) = 0$ and $h^0(\mathcal{L}_k\otimes \mathcal{O}_{C_{R}}) > 0$.

We can now outline our proof strategy. 
\begin{enumerate}
\item First, we prove Th \ref{ohtheorem} for one sided separating nodes with $\deg(D_L)=2$.
\item As a second step, we extend the proof to the $\deg(D_L) > 2$ cases by appealing to \S\ref{furtherdegeneration} and reducing the the problem to the $\deg(D_L) = 2$ case. 
\item In the third and final step, we treat the two sided separating nodes. 
\end{enumerate}

\subsection{One sided nodes with $\deg(D_L)=2$}

\subsubsection{Proof of \ref{ohtheorem}(A)}
\label{proof1aDL2}

Let us begin by recalling (from \S\ref{21_nilpotent_orbits_and_spectral_curves_local_story_on_a_smooth_curve}) how to obtain the vanishing orders $\chi_k$ associated to a Hitchin orbit whose partition label is $[p_i]$. We represent the Hitchin partition as a Young diagram by using its parts as column sizes. We fill the first column with `1's, the second column with `2's and so on. We then write down the numbers in the diagram column by column, dropping the leading `1'. The string of numbers so obtained are the vanishing orders. For the purposes of this section, we will include the leading `1' (corresponding to $k=1$) and form a vector $\vec{\chi}$ whose entries are $(\chi_1,\chi_2,\ldots \chi_N)$. We choose this convention, which is at variance with the choice in \S\ref{flavour} since it simplifies some of the combinatorial formulae. With this choice, the multiplicity of any integer $i$ in $\vec{\chi}$ is given by the part $p_i$. 

It also follows that any non-decreasing sequence of integers $\vec{\chi}$ with multiplicities $p_i$ obeying the conditions
\begin{enumerate}[label=(\alph*)]
\item $p_1 \geq p_2 \geq p_3 \geq p_4, \ldots $ 
\item $\sum p_i = N-1$
\end{enumerate}
will correspond to the vanishing orders of some nilpotent orbit.  

In what follows, we will need the following combinatorial fact. If there is some entry $i$ in $\vec{\chi}$ occurring with multiplicity $p_i$, then conditions $(a),(b)$ imply that the minimum value of $k$ for which $\chi_k = i$ is given by 
\begin{equation}
k_{min} = p_i(i-1) + 1.
\label{kminimal}
\end{equation}
And the minimal $k$ is achieved when $p_1 = p_2 = p_3 = \cdots = p_{i-1} = p_i$.

In order to prove a statement like \textit{Th \ref{ohtheorem}(A)}, we need to show that the vanishing orders $\vec{\chi}_{nodal}$ of $\mathcal{L}'_{k,R}$ at the node $p$ (obtained from Eq \eqref{modifiedvanishing}) correspond to the vanishing orders arising from \textit{some} nilpotent orbit. Let us take multiplicity of an integer $i$ in $\vec{\chi}_{nodal}$ to be some $\alpha_i$. We need to show that the sequence $\vec{\chi}_{nodal}$ from Eq \eqref{modifiedvanishing} is non-decreasing and the multiplicities $\alpha_i$ obey the conditions $(a),(b)$ above. Henceforth, we will drop the `\textit{nodal}' subscript and refer to the sequence of nodal vanishing orders as just $\vec{\chi}$.



In this section, we want to restrict to the cases with $\textit{deg}(D_L)=2$. Let the residues of the Higgs field $\phi$ at the smooth points of $D_L$ live in conjugacy classes corresponding to partitions $[p_i]^{(1)}$ and $[p_i]^{(2)}$. Let the corresponding vanishing orders be $\chi_k^{(1)},\chi_k^{(2)}$ respectively. The component $C_R$ is taken to be sufficiently generic that we have  $h^0(\mathcal{L}_k\otimes \mathcal{O}_{C_{R}}) > 0$ for all $k$.  

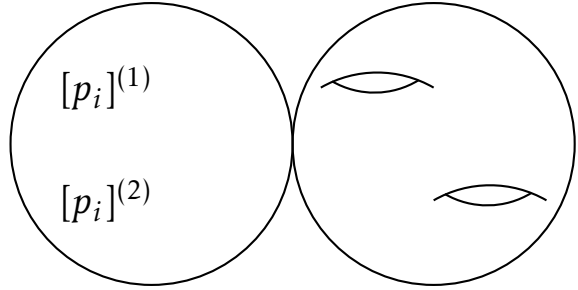
\begin{center}
\begin{tikzpicture}
\begin{scope}[scale=.75]
\node[scale=1.2] at (-.8,1){$[p_i]^{(1)}$};
\node[scale=1.2] at (-.8,-1){$[p_i]^{(2)}$};
\draw[thick] (3.0,1.0) arc (120:60:2);
\draw[thick] (3.2,1.112) arc (-120:-58.5:1.5);
\draw[thick] (5.0,-1.0) arc (120:60:2);
\draw[thick] (5.2,-0.888) arc (-120:-58.5:1.5);
\draw[thick] (0,0.0) circle [radius=2.5];
\draw[thick] (5,0.0) circle [radius=2.5];
\end{scope}
\end{tikzpicture}
\end{center}

Let us now assume that this degeneration leads to a restricted node with $n_k^L > 0$. This implies that $h^1(\mathcal{L}_k\otimes \mathcal{O}_{C_{L}}) > 0$ for some values of $k$.  From \eqref{h1equation}, we know that this happens iff $d_k^L\leq -2$. This reduces to (upon using \eqref{virtdimLR})  

\begin{equation}
\chi_k^{(1)} + \chi_k^{(2)} > k+1.
\label{h12puncture}
\end{equation}
The modified vanishing orders $\chi_k = 1 + n_k$ from Eq \eqref{modifiedvanishing} reduce to
\begin{equation}
\chi_k = \max(1, \chi_k^{(1)} + \chi_k^{(2)} - k ).
\label{defnofchik}
\end{equation}
We immediately see that $\chi_2 = \chi_3 = 1$. So, $\vec{\chi}$ is a sequence of the form
\begin{equation}
\vec{\chi} = (1,1,1,\ldots).
\end{equation}

\subsubsection*{Proof that $\vec{\chi}$ is non-decreasing }

Let us assume that there exists some $k$ for which 
\begin{equation}
\label{assumedecreasing}
\chi_k > \chi_{k+1} > 1 .
\end{equation}
This is possible iff
\begin{equation}
\begin{split}
\chi_{k+1}^{(1)} &= \chi_k^{(1)} \\
\chi_{k+1}^{(2)} &= \chi_k^{(2)}
\end{split}
\end{equation}
Using \eqref{kminimal}, this implies that
\begin{equation}
\begin{split}
k &\geq 2 (\chi_k^{(1)}  - 1) + 1 \\ 
k &\geq 2 (\chi_k^{(2)}  - 1) + 1
\end{split}
\end{equation}
Adding the two inequalities, we get $\chi_k^{(1)} + \chi_k^{(2)} - k \leq 1$ which contradicts our original assumption in \eqref{assumedecreasing} that $\chi_k > 1$. It follows that $\vec{\chi}$ is a non-decreasing sequence. 

\subsubsection*{Proof of conditions $(a),(b)$}

We still need to understand the multiplicities $\alpha_i$ of the integers occurring in $\vec{\chi}$. What are the allowed values of $\alpha_i$? We will study the possibilities case by case.

The cases $\alpha_1=N$ and $\alpha_1=N-1$ can clearly occur and they correspond to the regular orbit $[N]$ and the sub-regular orbit $[N-1,1]$ respectively. The case $\alpha_1=3,\alpha_2=1$ can also occur. This occurs for the extreme case where $\vec{\chi}^{\,(1),(2)}_k = (1,2,3,4,\ldots,N-1)$. This corresponds to the case where the two Hitchin orbits at $E_{1,2}$ are both the minimal nilpotent orbits and the corresponding $\vec{\chi} = (1,1,2,3,4,\ldots N-2)$. This is nothing but the vanishing orders for the Hitchin nilpotent $[3,1^{N-3}]$. This is the smallest nilpotent orbit that can occur at the node. The generalization to $\alpha_1 > 2, \alpha_2 =1$ is straightforward and leads to the vanishing orders for a hook type Hitchin orbit $[\alpha_1 + 1, 1^{N- \alpha_1 - 1}]$. This covers all instances with $\alpha_2 = 0,1$. We clearly get vanishing orders $\chi_k$ corresponding to a nilpotent orbit in each of these cases. 

Let us now turn to cases with $\alpha_2 > 1$. If `2' occurs exactly at the locations $k=l,l+1,l+2,l+\alpha_2$ in $\vec{\chi}$, then it follows that there is a repeated entry in either $\vec{\chi}^{\,(1)}$ or $\vec{\chi}^{\,(2)}$ (but not both) at the locations $k=l,l+1,l+2,\ldots,l+\alpha_2$. Let this repeated entry be the integer $i$ and let $\vec{\chi}^{\,(1)}$ contain these repeated entries. Now, $\max(\chi^{(2)}_l) = l-1$. So, if $\chi_l=2$, then $ \geq 3$. From Eq \eqref{kminimal}, we have $l_{min} = \alpha_2 (i-1)$. This implies that $\alpha_1 \geq l_{min} - 2 = \alpha_2 (i-1) - 2 $. When $i \geq 3, \alpha_2 > 1$, we see that $\alpha_1 > \alpha_2$. So, we have shown that $\vec{\chi}$ always satisfies condition (b). 

We are finally left with checking condition (a) for $\vec{\chi}$ in cases with $\alpha_2 > 1$. Let us say that (a) is violated. In other words, we have $\alpha_{j} > \alpha_{i}$ for some $j > i \geq 2$. Now, examining the possibilities (we omit the details), one can see that this is possible only if $\vec{\chi}^{\,(1)}$ or $\vec{\chi}^{\,(2)}$ itself were to violate condition (a). So, we arrive at a contradiction. Hence (a) always holds for $\chi_k$. This completes the proof of \textit{\ref{ohtheorem}(A)}.




\subsubsection{Proof of \ref{ohtheorem}(B)}
\label{proof1c}

Having deduced the Hitchin nilpotent $O$ at the node, we now turn to constraining the possible non-zero center parameters. We have already seen that $b_k^C = 0,1$ in \S\ref{Hitchinnodal}. Can any string of `0's and `1's occur as values of $b_k^C$? It turns out that the answer is \textit{no}. The allowed set of values are quite tightly constrained. Let us recall from Eq \eqref{bck0cond} the condition for a reduction in the number of center parameters at a reducible node :
\begin{equation}
\begin{split}
\sum_{D_L} \chi_k & >  k (deg(D_L) - 1)   \Longrightarrow b_k^C = 0  \label{center-left}\\
\sum_{D_R} \chi_k & >  k (deg(D_R) - 1)  \Longrightarrow b_k^C = 0 
\end{split}
\end{equation}

Specializing to the case of a one sided node with $deg(D_L)=2$, we get
\begin{equation}
 \chi_k^{(1)} + \chi_k^{(2)} >  k   \Longleftrightarrow b_k^C = 0 .
 \label{dkc02puncture}
\end{equation}

First, note that $\chi_2^{(1),(2)} = 1$ and hence $b_2^C = 1$ always. Next, we consider the two possibilities : (1) $b_3^C=1$ or (2) $b_3^C=0$.

\textit{Case 1 : $b_3^C=1$}

Let the first occurrence of a reduction in center parameters be for $k=l > 3$. This implies that we have $b_{l-2}^C=b_{l-1}^C=1$ and $b_l^C=0$,
\begin{equation}
\vec{b}^C = (0,1,\ldots,1,1,0,\ldots),
\end{equation}
where we have defined $b_1^C=0$.

This translates to the following conditions on $\chi_{k}^{(1),(2)}$ : 
\begin{equation}
\begin{split}
\chi_{l-1}^{(1)} &>  \chi_{l-2}^{(1)} \text{ (or) } \chi_{l-1}^{(2)} >  \chi_{l-2}^{(2)} \\
\chi_{l}^{(1)} &>  \chi_{l-1}^{(1)} \text{ (and) } \chi_{l}^{(2)} >  \chi_{l-1}^{(2)}
\end{split}
\end{equation}

In other words, if a repeated part were to occur in this piece of $\chi_{k}^{(1)},\chi_k^{(2)}$, then it can only occur for one among them and only for the entries at $k=l-2,l-1$. Let the vanishing orders without a repeated part be $\chi_k^{(1)} = (\ldots, i, i+1,i+2,\ldots )$ where we have set $\chi_{l-2}^{(1)}=i$. This clearly shows that the multiplicity $p_{i+1} = 1$ and it follows (from condition (a) in \S\ref{proof1aDL2}) that $p_{j} = 1$ for every $j \geq i+1$. 

Now, let us further \textit{assume} that $b_k^C = 1$ for some $k > l$. Let the smallest such $k$ be $m$. This implies that $b_{m-1}^C=0,b_m^C=1$ :

\begin{equation}
\vec{b}^C = (0,1,\ldots,1,1,0,\ldots,0,1,\ldots).
\end{equation}

From \eqref{dkc02puncture}, it follows that \textit{both} $\chi_{k}^{(1)},\chi_k^{(2)}$ have repeated parts at $k=m-1,k=m$. This, however, contradicts the statement that $p_{j}=1$ in $\chi_k^{(1)}$ for all $j \geq i+1$. So, our assumption that $b_k^C = 1$ for some $k > l$ is wrong. So, the only possible set of $b_k^C$ with $b_3^C=1$ are given by 

\begin{equation}
\vec{b}^C = (0,1,1,\ldots,1,0,0,\ldots,0)
\end{equation}

\textit{Case 2 : $b_3^C=0$}

Let us turn to the case where $b_3^C=0$. From \eqref{dkc02puncture}, this implies that $\chi_3^{(1)}= \chi_3^{(2)} = 2$. This forces $p_1^{(1)} = p_1^{(2)}=2$. From condition (b) in \S\ref{proof1aDL2}, we then have $p_2^{(1),(2)} = 1 \text{ (or) } 2$. Even if one among $p_2^{(1)}, p_2^{(2)}$ equals 1, then a simple calculation shows that $b_k^C=0$ for all $k > 3$. It remains to consider the case where both $p_2^{(1)}= p_2^{(2)}=2$. In this case, both $\chi_{k}^{(1)}$ and $\chi_{k}^{(2)}$ have the following form
\begin{equation}
\begin{split}
\vec{\chi}^{\,(1)} &= (1,1,2,2,3,\ldots), \\
\vec{\chi}^{\,(2)} &= (1,1,2,2,3,\ldots).
\end{split}
\end{equation}
and we have 
\begin{equation}
\vec{b}^C = (0,1,0,1,0,\ldots)
\end{equation}

Let us now assume that $b_6^C=0$, then at least one among $p_3^{(1)}$ or $p_3^{(2)}$ is equal to 1. Let us take $p_3^{(1)}=1$. This forces $p_j^{(1)} = 1$ for all $j \geq 3$. From this, it follows that $b_k^C=0$ for all $k \geq 6$.

On the other hand,  if $d_6^C=1$, then $p_3^{(1)}=p_3^{(1)}=2$. The multiplicities can't be bigger since $p_j \leq \alpha_2$ for $j > 2$ and we are in the case where $p_2^{(1)} = p_2^{(2)} =2$. So, we have

\begin{equation}
\begin{split}
\label{alternating-center1}
\vec{\chi}^{\,(1)} &= (1,1,2,2,3,3,4\ldots), \\
\vec{\chi}^{\,(2)} &= (1,1,2,2,3,3,4\ldots),\\
\vec{b}^C &= (0,1,0,1,0,1,0,\ldots).
\end{split}
\end{equation}

We then repeat the same procedure for the two cases $b_{8,C} = 0,1$ and find that $b_k^C$ is always of the form

\begin{equation}
\vec{b}^C = (0,1,0,1,0,1,0,\ldots,0,1,0,0,\ldots,0)
\end{equation}

This proves \textit{\ref{ohtheorem}(B)}.

Furthermore, we see that in each of these cases, $d_k^L$ is always of the form 

\begin{equation}
\label{alternating-degrees}
\begin{split}
d_k^L &= k - \chi^{\,(2)}_k - \chi^{\,(1)}_k  \\ 
 &= -\tfrac{1}{2}\bigl(1-(-1)^k\bigr) .
 \end{split}
\end{equation}

Consequently, we have that 
\begin{equation}
\label{alternatingLk}
\deg(\mathcal{L}_{k,L}) = -1- \tfrac{1}{2}\bigl(1-(-1)^k\bigr).
\end{equation}

We also see that we have $\vec{b}^C = (0,1,0,1,\ldots,0,1)$ for some $N = 2n$ iff
\begin{eqnarray}
\vec{\chi}^{(1)} &=& (1,1,2,2,3,3,\ldots n,n), \\
\vec{\chi}^{(2)} &=& (1,1,2,2,3,3,\ldots n,n).
\end{eqnarray}

This corresponds to the case where the two residues at the marked points on $C_L$ live in the nilpotent conjugacy class $[2^n]$. This proves Prop \ref{alternatingproposition}.

\subsubsection{Proof of \ref{ohtheorem}(C)}
\label{proof1d}

Given a non-regular nilpotent $O$ at the node, we would now like to understand the constraints on the allowed non-zero values of $b_k^C$. Let the vanishing orders at the node be of the form
\begin{equation}
\vec{\chi} = (1,1,1,\ldots,2,2,\ldots,2,3,\ldots).
\end{equation}

Let the (Hitchin) partition label of $O$ be $[p_i]$. Having proven Th \ref{ohtheorem}(A), we know that $\alpha_i=p_i$. 

We would like to arrive at a constraint on the total number of center parameters $b^C$ given a partition nilpotent $O$ at the node. Now, $b_k^C=0$ for every $k$ such that $\chi_{k}>1$ since the condition for $h^1(\mathcal{L}_k\otimes \mathcal{O}_{C_{L}})> 0$ \eqref{h12puncture} is stronger than the condition for the vanishing of $b_k^C$ \eqref{dkc02puncture}. So, it is straightforward that $b^C \leq p_1 - 1$. But, we will see that there is actually a stronger upper bound and that there is also a lower bound on $b^C$.

In the proof of Th \ref{ohtheorem}(B), we saw that $p_1 - p_2 > 1$ for every allowed nilpotent at the node. Assume $p_2 > 1$. We then have 

\begin{equation}
\chi_{k}^{(1)} + \chi_{k}^{(2)} = k + 2  \hspace{0.1in }\forall \hspace{0.1in } k \in (p_1, p_1 + p_2 - 1 ).
\end{equation}

This is possible iff one among $\chi_{k}^{(1)}, \chi_{k}^{(2)} $ were to also have repeated parts for this range of values of $k$. We take $\chi_{k}^{(1)}$ to have the repeated part $j$. Since $\max(\chi_{k}^{(2)})=k-1$, we see that $j \geq 3$. It follows that 

\begin{equation}
\label{p2constraint}
p_{j-1}^{(1)} \geq p_j^{(1)} = p_2. 
\end{equation}

This imposes a strong constraint on the allowed $O$ at the node. To understand this constraint, let us ask what are the allowed values of $p_1$ given that \eqref{p2constraint} is always true. The smallest possible value of $p_1$ will occur when we have $j=3$ and $p_{1}^{(1)} = p_2$ and $p_{2}^{(1)} = p_2 - 1$. In this case, we have $p_1 = 2p_2 - 1$. More generally, we always have
\begin{equation}
p_1 > 2 p_2. 
\end{equation} 

This condition rules out nilpotent orbits with Hitchin partition of type $[n^2]$ as possible nodal nilpotents. Recall that these nilpotents occurred in the list of allowed nodal nilpotents in analysis using constraints on flavor central charges in \S\ref{restricted_nodes}. If one had carried out an exhaustive enumeration of nodal nilpotents by brute force calculation for any fixed $N$, one would have seen that orbits of type $[(N/2)^2]$ do not occur. But, as we just showed, it is possible to give a general proof for all $N$ using the nodal Hitchin system.

Let us now try to understand the range of values of $k$ for which we could have $b_k^C=0$ but $h^1(\mathcal{L}_k\otimes \mathcal{O}_{C_{L}})=0$. This would be the range of values of $k$ for which the following relation holds :
\begin{equation}
\chi_{k}^{(1)} + \chi_{k}^{(2)} = k + 1.
\label{bkc2noh1}
\end{equation}

The smallest value of $b^C$ is reached when this range is the largest. And this range would be the largest when $p_{1}^{(1)}$ is minimal and $\chi_{k}^{(2)} = k-1$. In this case, every instance where $\chi_k = 1$ and $b_k^C=0$ arises from $\chi_{k}^{(1)} = 2$ and the only non-zero center parameters exist for those $k$ where $\chi_{k}^{(1)} = 1$. In this scenario, we have

\begin{equation}
p_{1}^{(1)} + p_{2}^{(1)} = p_1 - 1.
\end{equation} 

We have already argued that $\min(p_{1}^{(1)}) = p_{2}^{(1)}$. When this minimal value of $\alpha_{1}^{(1)}$ is reached, we have $p_{1}^{(1)} =  \lfloor (p_1 - 1)/2 \rfloor $. It follows that $b^C \geq \lfloor (p_1/2 - 1) \rfloor $.

 At the other end, the maximum allowed value of $b^C$ is reached when the range of $k$ for which \eqref{bkc2noh1} holds is the smallest. Note that we have already assumed that $\chi_{k}^{(1)} + \chi_{k}^{(2)} = k + 2$ for $p_2$ values of $k$ and that the corresponding repeated entry in $\chi_k^{(1)}$ is $j$ with $p_j^{(1)} = p_2$. This implies that $\alpha_{j-1} \geq \alpha_2 = p_2$. And in the range of $k$ for which $\chi_k^{(1)}=j-1$, if $\chi_k^{(2)}$ were to also have repeated parts, then it is easy to see that  the resulting $\vec{\chi}$ would have a decreasing sub-sequence. This violates Th \ref{ohtheorem}(A). So, $\chi_k^{(2)}$ does not have repeated parts for these values of $k$. If take $p_{j-1} = p_2$, then this implies that \eqref{bkc2noh1} holds exactly for $k \in (p_1 - p_2  , p_1 - 1 )$ and we have $b^C = p_1 - p_2 - 1$.  To summarize, we have

\begin{equation}
\lfloor (p_1)/2 \rfloor - 1 \leq b^C \leq p_1 - p_2 -1.
\end{equation}

which is equivalent to

\begin{equation}
p_1 - 2 \leq 2 b^C \leq 2(p_1 - p_2 -1).
\end{equation}

This proves \textit{Th \ref{ohtheorem}} for one sided nodes with $\deg(D_L)=2$.

\subsection{One sided nodes with $\deg(D_L)>2$}

When we have a one sided node with $\deg(D_L)>2$, we first pick a stable degeneration of $C_L$ which will be a tree of $\mathbb{P}^1$s with the end points of the tree being $\mathbb{P}^1$s with two punctures. We now normalize each node in $C_L$ starting from the the ends of tree. At every step, we use \textit{Th \ref{ohtheorem}(A)} for the $\deg(D_L)=2$ case and insert the nodal nilpotent on the right component of the normalized curve. 

The final result of this procedure will be a one sided node with $\deg(D_L)>2$ with some partitions $[p_i]^{(1)},[p_i]^{(2)}$.  The arguments in \S\ref{furtherdegeneration} ensure that the resulting $[p_i]^{(1)},[p_i]^{(2)}$ do not depend on the choice of the stable degeneration or the subsequent choice of the order in which we choose to do the normalizations. As a final step, we can now use \textit{Th \ref{ohtheorem}(A)-(C)} for $\deg(D_L)=2$.  This extends \textit{Th \ref{ohtheorem}} to all one sided nodes with $\deg(D_L)>2$.

\subsection{Two sided nodes} 

We now take up the case of two sided nodes. These are nodes in which the constraints on the space of center parameters arise from both the left and right components of a separating node. In other words, we have $h^0(\mathcal{L}_k\otimes \mathcal{O}_{C_{L}})=0$ and $h^0(\mathcal{L}_{k'}\otimes \mathcal{O}_{C_{R}})=0$ for some values of $(k,k')$ with $k \neq k'$. This can occur only in cases where both $C_L$ and $C_R$ are $\mathbb{P}^1$s. Let the first constraint on the left component occur at $k=l$ and the first constraint on the right occur at $k'=l'$.  Let us take $l' > l$. This implies that on $C_L$, we have  
\begin{equation}
\begin{split}
h^0(\mathcal{L}_l\otimes \mathcal{O}_{C_{L}})&=0  \\ 
h^0(\mathcal{L}_{l'}\otimes \mathcal{O}_{C_{L}}) &> 0 
\end{split}
\end{equation}

By our proof of Th \ref{ohtheorem}(B) for one sided nodes \eqref{alternating-degrees}, this is possible iff $d_k^L = (0,-1,0,-1,\ldots)$. And since we are only considering nodal degenerations arising from OK theories, this implies that $d_k^R \leq -1$ in any two sided node. So, we have $n_k^{L,R}=0 $ for all values of $k$. So, it follows that any two sided node necessarily has $\mathcal{O} = [2n]$.

The allowed set of center parameters arise from a combining the constraints from the left and the right components. On the left component, we have a pattern of constraints that is of the type that leads to a $b^C_k = (1,0,1,0,\ldots,1)$. On the right, we have a pattern of constraints that leads to a $b^C_k = (1,1,1,\dots,1,0)$. It is clear that the combined application of both sets of constraints also leads to a set of non-zero center parameters that obeys the conditions in $Th$ \ref{ohtheorem}(B)-(C).

This completes the proof of \textit{Th} \ref{ohtheorem}.


\begin{cor} [] 
The center parameters can be interpreted as the $H$-invariant polynomials for $H = SU(n)$ (or) $H=Sp(n)$ for some $n \leq N$ such that $\text{rank}(H) = b^C$.  Combined with the Hitchin nilpotent $O$, this completes the association of a pair $(O,H)$ for every restricted node where $H$ is always a simple Lie group.
\end{cor}

\textit{Proof} From the allowed possibilities for $b_k^C$ in Th \ref{ohtheorem}(B) and an inspection of the degrees of invariant polynomials in simple Lie algebras (see Appendix \ref{a1_appendix}), it is clear that the center parameters can always be interpreted as being the invariant polynomials for either a $SU(n)$ subgroup or a $Sp(n)$ subgroup\footnote{There is also the possibility that $b_k^C$ could be the invariant polynomials for an $H=SO(2n+1)$. This can be ruled out using the flavour considerations discussed in footnote \ref{SOodd}. In light of this, it is interesting to wonder if one could obtain a stronger version of Theorem \ref{ohtheorem}(B) which directly constrains the group $H$ from a Hitchin system point of view. The present version can be thought of as constraining the Weyl group $W(H)$, which is not sufficient to distinguish $Sp(n)$ from $SO(2n+1)$. } of $SU(N)$.

\begin{remark} []
It is interesting to ask if the converse of Th \ref{ohtheorem}(C) holds. In other words, for a given Lie algebra $\mathfrak{j}=\mathfrak{sl}(N)$, does every pair of $(O,H)$ obeying Th \ref{ohtheorem}(C) actually occur at a restricted node? The physics computations indicate that this is true, but we have not been able to provide a proof for arbitrary $N$.
\end{remark}

\subsection{Restricted Nodes vs Semistable Higgs bundles}
\label{stability}
As we have seen, the nodal limit of a symplectic Hitchin integrable system on $C$ gives rise, in a canonical way, to a symplectic integrable system on the normalization $\tilde{C}$.  Our main interest, in this paper, has been the appearance of restricted nodes. In that case, the symplectic integrable system that arises on one of the components ($C_L$) of the normalization is not a semistable $\mathfrak{j}$-Hitchin system.

But the ``generic" behaviour is that the degeneration leads to a standard node. In that case, the integrable system is a semistable Hitchin system  on $C_L$ with $n+1$ marked points (where $n=\deg(D_L)$), where the $(n+1)^{\text{st}}$ point is the pre-image of the node, and the conjugacy class of the residue there is the regular nilpotent. A necessary and sufficient condition for this is that $h^1(\mathcal{L}_{k,L})=0,\;\forall k$. As we explain in Appendix \ref{OKappendix}, this is also a necessary and sufficient condition for the corresponding irreducible character variety to exist. When $h^1(\mathcal{L}_{k,L})\neq0$ for some $k$, the irreducible character variety does not exist and hence there is no moduli space of semistable Higgs bundles on $C_L$.



\subsection{Compatibility of Coulomb and Higgs branch considerations}
\label{compatibility}

As we already discussed in \S\ref{restricted_nodes}, the existence of a superconformal $\mathcal{N}=2$ theory associated to a tame Hitchin system allows us to constrain the data $(O,H)$ from either the Coulomb branch or Higgs branch perspectives. While we used the \hyperlink{algorithm}{\textit{Higgs-Coulomb} algorithm} in \S\ref{restricted_nodes}, we  used purely Coulomb branch considerations to constrain $(O,H)$ in the present section. By this seemingly different route, we have arrived at the same set of allowed pairs $(O,H)$.

While the procedure for determining $O$ was essentially the same, the procedure for determining $H$ looks quite different. To see that they are the same, note that the vector $\vec{\pi}-\vec{\pi}_O$ of the  \hyperlink{algorithm}{algorithm} is just equal to the sum of the left and center  Hitchin base dimensions, $\vec{b}^L+\vec{b}^C$.  Proposition \ref{centervanishinglemma} relates $b^C_k=0$ to $h^1(C_L,\mathcal{L}'_{k,L})>0$. Since $C_L$ has genus-0, $b^L_k \coloneqq h^0(C_L,\mathcal{L}'_{k,L}) > 0$ and $h^1(C_L,\mathcal{L}'_{k,L})>0$ are mutually-exclusive. Hence if $b^L_k+b^C_k > 0$, we must have $b^C_k=1$. So $H$, as determined by Theorem \ref{ohtheorem}(B)  really \emph{is} the highest-rank subgroup of $F$ whose Casimirs are a subset of the positive entries of $\vec{\pi}-\vec{\pi}_O$. Of the allowed pairs obtained in \S\ref{restricted_nodes}, only $([n^2],SU(2n)_0)$ was excluded (in the untwisted $A_{2n-1}$ theory) by Theorem \ref{ohtheorem}(C). As already noted it does appear in the twisted version of the theory  \cite{Chacaltana:2012ch}. We also note that the constraints on $\text{rank}(H)$ obtained from \textit{Th} \ref{ohtheorem}(C) is exactly the same as the one obtained in (\ref{row1}) by flavor considerations.


In instances where the Higgs branch geometry, including the hyperK\"{a}hler metric, is known, one expects to see the smaller groups $H$ as the subgroups of $J$ that continue to act as isometries of the 4d Higgs branch \cite{Gaiotto:2011xs} appearing on the $C_L$ component of a restricted node. In the particular case of the restricted node arising in the $SU(3),\; N_f=6$ theory, the investigation of this question goes back to the work of \cite{Gaiotto:2008nz}. 

More generally, for every pants-decomposition of C, there is a different realization of the Higgs branch as a hyperK\"ahler quotient. When the corresponding boundary point of $\overline{\mathcal{M}}_{g,n}$ involves only $(3g-3+n)$ standard nodes, the quotient is by $J^{3g-3+n}$. When some of the nodes are restricted nodes $(O,H)$, the quotient is by $H$ rather than $J$ and one of the 3-punctured spheres which meet at the node has an insertion of $O$, rather than the regular nilpotent. Hiraku Nakajima has informed us that he has been able to provide a mathematical proof of the existence of these different realizations of the Higgs branch in certain cases \cite{Nakajima}.

%
%

\bigskip
\section*{Appendices}
\appendix

\section{Proof of Theorem \ref{hitchinglobalthm}}
\label{Proof}

As in the text, $\mathcal{L}_k= K_C ^{\otimes k}\otimes \mathcal{O}(-\sum_{p_i} \pi^{(i)}_k)$, where $\pi^{(i)}_k:=k-\chi^{(i)}_k$ satisfies $1\leq\pi^{(i)}_k\leq k-1$. Our OK condition is that $H^1(C,\mathcal{L}_k)=0$ for $C$ smooth. But we need to consider arbitrary nodal degenerations of $C$.

So let $C$ be a nodal curve with irreducible components $C_a$. Each $C_a$ has geometric genus $g_a$, $t_a\geq 1$ branches of nodes and a set of marked points $S_a\subset\{p_1,p_2,\dots p_n\}$.

We easily compute
\begin{equation}\label{multidegreeL}
\deg(K_C\otimes \mathcal{L}_k^{-1}\otimes \mathcal{O}_{C_a}) = -\bigl[(k-1)(2 g_a-2+t_a) + \sum_{p_i\in S_a} \pi^{(i)}_k\bigr]
\end{equation}
\begin{defn}
We will call a component $C_a$ \emph{blighted} if
\begin{itemize}
\item $g_a=0$
\item $t_a=1$
\item $ \sum_{p_i\in S_a} \pi^{(i)}_k < k-1$
\end{itemize}
For a blighted component $C_a$ define the positive integer
\begin{equation}
n^{a}_k \coloneqq k-1 - \sum_{p_i\in S_a} \pi^{(i)}_k
\end{equation}
\end{defn}

\noindent
From \eqref{multidegreeL},  $\deg(K_C\otimes \mathcal{L}_k^{-1}\otimes \mathcal{O}_{C_a})> 0$ if and only if $C_a$ is blighted.

\begin{lemma}
Restricted to a blighted component,
\begin{equation}
\mathcal{L}'_k\otimes \mathcal{O}_{C_a} = \mathcal{L}_k\otimes \mathcal{O}_{C_a} (n^{a}_k p)
\end{equation}
where $p$ is the node.
\end{lemma}

\begin{proof}
This follows from the fact that $\deg(\mathcal{O}_{C_a}(- \mathcal{C}_a))=1$ and our definition \eqref{globallprimedef} of $\mathcal{L}'_k$.
\end{proof}

\begin{lemma}\label{unblighted}
If $C$ has no blighted components, then
\begin{itemize}
\item[a)] $\mathcal{L}'_k=\mathcal{L}_k$
 \item[b)]$H^1(C,\mathcal{L}_k)=0$.
 \end{itemize}
\end{lemma}

\begin{proof}
For (a) , we note that if  $k-1- \sum_{p_i\in S} \pi^{(i)}_k > 0$, then $k-1- \sum_{p_i\in S_a} \pi^{(i)}_k > 0$ for any subset $S_a \subset S$. So if there's a (possibly reducible) genus-0 component of $C$ contributing to the twist  \eqref{globallprimedef}, then it has an irreducible subcomponent $C_a$ on which $n^a_k >0$.

For (b), we note that, by Serre duality, $h^1(\mathcal{L}_k) = h^0(K_C\otimes \mathcal{L}_k^{-1})$. In the absence of blighted components, the degree of $K_C\otimes \mathcal{L}_k^{-1}$ is non-positive when restricted to every component $C_a$ of $C$. The OK condition implies that the total degree is strictly negative; hence it must be negative on at least one component. Therefore any global section of $K_C\otimes \mathcal{L}_k^{-1}$ on $C$ vanishes.
\end{proof}

Our goal is now to reduce the problem of computing $H^*(C,\mathcal{L}'_k)$ on a curve with blighted components to the same computation on a simpler curve with no blighted components.

\begin{defn}
Consider a blighted component $C_a$. By definition, it intersects the rest of $C$ (which we will denote by $\check{C}$) at the node $p$. Our \emph{pruning operation} consists of removing the component $C_a$ and replacing the branch of the node on $\check{C}$ by a marked point with $\chi_k= 1+n^a_k$ (or $\pi_k=k-1-n^a_k$).
\end{defn}

The degree of the twist, $n^a_k$, was chosen precisely so that $H^0(C_a, \mathcal{L}'_k)=H^1(C_a,\mathcal{L}'_k)=0$. The long exact sequence associated to
\begin{equation}
0 \to \mathcal{L}'_k\otimes \mathcal{O}_{\check{C}}(-p) \to  \mathcal{L}'_k \to  \mathcal{L}'_k\otimes \mathcal{O}_{C_a} \to 0
\end{equation}
splits and we find 
\begin{equation}\label{thesame}
\begin{split}
    H^0(C, \mathcal{L}'_k) &= H^0(\check{C}, \mathcal{L}'_k(-p))\\
    H^1(C, \mathcal{L}'_k) &= H^1(\check{C}, \mathcal{L}'_k(-p))
\end{split}
\end{equation}
Let us denote the irreducible component of $\check{C}$ which contains $p$ as $C_b$. Note the following:
\begin{itemize}
\item Let $C_c$ be any \emph{other} component of $\check{C}$, \emph{except} $C_b$.
We have $\mathcal{L}'_k(-p)\otimes \mathcal{O}_{C_c} = \mathcal{L}'_k\otimes \mathcal{O}_{C_c}$.
\item Since $\mathcal{O}_{\check{C}}(-n^a_k\mathcal{C}_a)= \mathcal{O}_{\check{C}}(-n^a_k p)$, the pole order of $\mathcal{L}'_k(-p)\otimes\mathcal{O}_{C_b}$ at $p$ is $\pi_k=k-1-n^a_k$.
\item Let $S\supset S_a\cup S_b$ be a subset of the marked points on $C$. Let $\check{S} = \{p\}\cup (S\backslash S_a)$. Set $\pi^p_k= k-1-n^a_k$. Since
$$\pi^p_k +\sum_{p_i\in S_b} \pi^{(i)}_k = \sum_{p_i\in S_a\cup S_b} \pi^{(i)}_k$$
the coefficient, $-n^{\check{S}}_k$, of $\mathcal{C}_{\check{S}}$ in $\check{\mathcal{L}}'_k \coloneqq \mathcal{L}'_k\otimes \mathcal{O}_{\check{C}}(-p)$ is the \emph{same} as the coefficient, $-n^S_k$, of $\mathcal{C}_S$ in $\mathcal{L}'_k$.
\end{itemize}

The upshot is that the line bundle $\mathcal{L}'_k(-p)$ on $\check{C}$ is exactly the line  bundle ${\check{\mathcal{L}}}'_k$ that we would construct by the recipe \eqref{globallprimedef} for the curve $\check{C}$ with marked points $\check{S}$ and an additional marked point at $p$ with $\chi_k= 1+n^a_k$.

By the pruning procedure, we have constructed a new curve $\check{C}$ and a line bundle ${\check{\mathcal{L}}}'_k$ with exactly the same cohomology groups \eqref{thesame} as $(C,\mathcal{L}'_k)$. Now we drop the $\check{~}$s and repeat the pruning operation. Eventually, we arrive at a curve with no blighted components. We then apply Lemma \ref{unblighted} to conclude that $H^1(C,\mathcal{L}'_k)=0$.

Thus we have shown that $H^1(C,\mathcal{L}'_k)=0$ for every fiber of $\overline{\mathcal{C}}\to \overline{\mathcal{M}}_{g,n}$. Hence 
\begin{equation}
\mathcal{B}_k= \pi_*\mathcal{L}'_k
\end{equation}
is locally-free.

\section{OK theories and semistable Higgs bundles}
\label{OKappendix}

We begin by recalling the nonabelian Hodge theorem for tame Hitchin systems due to Simpson \cite{simpson1990harmonic}. The NAH theorem sets up a correspondence between the moduli space of semistable parabolic Higgs bundles and the character variety of irreducible representations $\rho : \pi_1(C_{g,n}) \rightarrow SL_N $ with parabolic structure.

Let us recall the local dictionary from Simpson \cite{simpson1990harmonic}. As in \cite{simpson1990harmonic}, let $(E,\Phi)$ be a filtered Higgs bundle, $(V,\nabla)$ be the flat connection with $V$ a filtered vector bundle,  $(L,\mu)$  a filtered local system with $\mu$ being an endomorphism of $L$. For any Lie algebra element $a$, we can write its Jordan decomposition as $a = a_N + a_s$ where $a_N$ is nilpotent, $a_s$ is semisimple and $[a_N,a_s]=0$. Since our residues are elements of the Lie algebra $\mathfrak{j}$, a similar Jordan decomposition exists for $\text{Res}(\cdots)$. We use $(\text{Res}(\cdots))_N$ to denote ``nilpotent part of the residue''. Then the local dictionary can be described in the following way: 

\begin{enumerate}
\item The weights and eigenvalues of the semi-simple parts of the residue are permuted according the following table (from p.~720 of \cite{simpson1990harmonic})

\begin{center}
\begin{tabular}{|c|c|c|c|}
\hline 
 & $(E,\Phi)$ & $(V,\nabla)$ & $(L,\mu)$ \\ 
\hline 
weights  &  $\alpha$ & $\alpha - 2 \beta $ & $-2 \beta$ \\ 
\hline 
eigenvalue(s) & $\beta + i\gamma $ & $\alpha + 2 i \gamma $ &  $\exp( - 2 \pi i \alpha + 4 \pi \gamma )$ \\ 
\hline 
\end{tabular} 
\end{center}

\item The fibers at each puncture of $E$, $V$ and $L$ have a refined decomposition  given by the triple $(\alpha,\beta,\gamma)$ at that puncture.  These decompositions are invariant under the respective operators $\Phi, \nabla, \mu$. On matching pieces of this decomposition,  the nilpotent parts of the endomorphisms coincide: \hbox{$(\text{Res}(\Phi))_N = (\text{Res}(\nabla))_N = (\text{Res}(\mu))_N$ .}
\end{enumerate}

Our arguments apply most directly to instances of NAH where $Res(\Phi)$ at each of the $n$ punctures is strictly nilpotent, i.e.~$\beta = \gamma = 0$. From the above table, it follows that the eigenvalues of the holonomy are $\exp { (- 2 \pi i \alpha)}$. The boundary conditions for the gauge field encode the \textit{parabolic weights} $\alpha$ at each puncture. The parabolic weights must be chosen in a way that is compatible with the nilpotent residues of the Higgs fields. The idea (see \cite{boden1996moduli} for an exposition) is that the fiber of $E$ at each puncture admits a filtration
\begin{equation}
E|_p  = F_l \supset F_{l-1} \supset F_{l-2} \supset \dots \supset F_1\supset 0
\end{equation}
where $F_j = \ker(Res(\Phi)^j)$. To this filtration, we assign a set of parabolic weights, $\alpha_j(p)\in [0,1)$ with $\alpha_j(p)<\alpha_{j-1}(p)$.  To each $\alpha_j(p)$ we assign a multiplicity\footnote{Since we are in $SL_N$, we further require $\sum_j q_j\alpha_j = 0 \mod{1}$.} $q_j = \dim(F_j/F_{j-1})$. The partition $[q_1,q_2,\dots,q_l]$ of $N$ is the Nahm partition which is the transpose of the Hitchin partition for the nilpotent orbit $O_H \ni Res(\Phi)$. The datum $(E,\Phi,\alpha)$ defines a strongly parabolic Higgs bundle. By NAH, the multiplicities of the eigenvalues of the holonomy $\mu$ are given by the same Nahm partition, $[q_1,q_2,\dots,q_l]$.
  
Let us further specialize to the case of a Higgs bundle on a genus zero curve $C_{0,n+1}$ with $n+1$ punctures such that $Res(\Phi)$ is regular nilpotent at (at least) one of the punctures (say the $(n+1)^{\text{st}}$ puncture). We will refer to these as the \textit{regular cases}. In these regular cases, Simpson has derived necessary and sufficient conditions for the irreducible character variety to be non-empty \cite{simpson1992products}. Let us define $
D := \sum_{a=1}^{n+1} \dim(C_a) - 2(N^2 -1)$ and $r_a := N - m_a$ where $C_a$ are the $SL_N$ conjugacy classes in which the local holonomies live, $a$ labels the punctures and  $m_a$ denotes the largest multiplicity for the eigenvalues of the holonomy matrix at the puncture $a$. 

In terms of these quantities, Simpson's conditions are: 
\begin{enumerate}[label={(\greekalpha*)}]
\item  $D \geq 0$,
\item $\sum_{a=1}^{n} r_a \geq N$. 
\end{enumerate} 

When the irreducible character variety is not empty, the quantity $D$ is equal to its complex dimension and it matches the dimension of the $Higgs$ moduli space computed using Riemann-Roch (as in \S\ref{2_tame_hitchin_systems_on_smooth_curves}). 

\begin{prop}
In the regular case, Simpson's two conditions ($\alpha$),($\beta$) above are equivalent to the ``OK'' condition which we introduced for the line bundles in $\mathcal{L}_k$  in \S\ref{13_hitchin_systems_corresponding_to_good_theories}. 
\end{prop}
\begin{proof}
As explained earlier in this section, $m_a$ can be identified with the first part of the Nahm partition $q_1$ at the puncture $a$.
From the algorithm (in \S\ref{21_nilpotent_orbits_and_spectral_curves_local_story_on_a_smooth_curve}) for the zero orders $\chi_k$, we see that the value of $\chi_N = q_1$. With this translation, we see that condition $(\beta)$ is the same as 
\begin{equation}
\sum_{a=1}^{n} (N- \chi^{(a)}_N) \geq N.
\end{equation}
By a simple rearrangement, this is equivalent to demanding
\begin{equation}
 (n-1) N - \sum_{a=1}^{n+1} \chi^{(a)}_N  \geq - 1,
\end{equation}
where we have used the fact that $\chi_N^{n+1} = 1$ since we have a regular nilpotent residue for the Higgs field at the $(n+1)^{\text{st}}$ puncture. From \eqref{degrees}, we recognize this to be exactly the condition that $\deg(\mathcal{L}_N) \geq -1$! It follows that demanding $(\beta)$ holds is the same as demanding that $h^1(\mathcal{L}_{N}) = 0$ which is one of our conditions for the tame Hitchin system to be OK.

What about condition $(\alpha)$? To study this, we first note that the quantity $D$ has a simple relationship to the indices of the line bundles $\mathcal{L}_k$,
\begin{equation}
D = 2 \sum_k \text{ind}(\mathcal{L}_k).
\end{equation}

If $h^1(\mathcal{L}_{k}) = 0$ for all $k$, it is straightforward that $(\alpha)$ holds. It is also clear that if $h^1(\mathcal{L}_k) > 0$ for all $k$, then both $(\alpha),(\beta)$ fail to hold. The interesting situations are the ones where $(\alpha)$ might be violated but $(\beta)$ holds. Such cases could occur if $h^1(\mathcal{L}_{k}) > 0$ for some $k < N$ but $h^1(\mathcal{L}_{N}) = 0$. What can we say about $D$ in such cases? 

To approach these cases, imagine we have a one sided separating node (see \S\ref{proofstrategy} for the definition) with $\mathcal{O}= [N]$ and $\deg(D_L)=n$. Now, the conditions $h^1(\mathcal{L}_{N}) = 0$ and $h^1(\mathcal{L}_{k}) > 0$ (for some $k < N$) are equivalent to demanding that $b_{N}^C = 1$ while $b_k^C =0$ for some $k < N$. From our proof of Th \ref{ohtheorem}(B), the only such possibilities occur when $\vec{b}^C_k = \tfrac{1}{2}\bigl(1+(-1)^k\bigr)$. In these cases, we showed that $\deg(\mathcal{L}_k)$ is necessarily of the form \eqref{alternatingLk}
\begin{equation}
\label{alternatingLkN}
\deg(\mathcal{L}_k) = -1-  \tfrac{1}{2}\bigl(1-(-1)^k\bigr).
\end{equation}

By Prop \ref{alternatingproposition}, such a scenario can occur when $N$ is even, $n=2$ and the residues of the Higgs field at the two marked points live in the nilpotent conjugacy class $[2^{N/2}]$. In these cases, an explicit calculation shows that $D < 0$. And Th \ref{ohtheorem}(C) implies that $[2^{N/2}]$ is not an allowed nodal nilpotent. This guarantees that there is no other scenario with $n > 2$ for which \eqref{alternatingLkN} could hold. This guarantees that $(\alpha)$ is violated whenever  $h^1(\mathcal{L}_{N}) = 0$ and $h^1(\mathcal{L}_{k}) > 0$ for some $k < N$. 
\end{proof}
 
In the regular case, demanding that the strongly parabolic Higgs bundle be OK is necessary and sufficient for the corresponding irreducible character variety to be non-empty. By the NAH theorem, this is the same as demanding the existence (i.e.~non-emptiness) of the corresponding moduli space of semistable Higgs bundles. The novel feature here is that semi-stability (in the Higgs sense) admits a translation to a condition on the line bundles $\mathcal{L}_k$ appearing in the description of the Hitchin base.

If we relax the assumption that one of the residues of the Higgs field is regular, then  Simpson's two conditions are known to be necessary, but not sufficient for the non-emptiness of the character variety (see, for instance, the discussion in Kostov's survey \cite{kostov2004deligne}). A natural guess is that the OK condition on the line bundles $\mathcal{L}_k$, which is \emph{stronger} than Simpson's conditions in this case, might be sufficient. To this end, we would like to propose two conjectures. The first is that the OK condition on strongly parabolic Higgs bundles is sufficient to ensure the non-emptiness of the corresponding character variety. The second, more optimistic, conjecture is that the OK condition is both necessary and sufficient.

We hope to study these conjectures further in a later work. We note here that some of the best known results towards the general problem of providing necessary and sufficient conditions are in Crawley-Boevey \cite{crawleyquiver} which follows the earlier work of \cite{crawley2006multiplicative}. A beautiful survey emphasizing the connection to Higgs bundles is in \cite{MR2508903}. Additional recent results are in \cite{soibelman2016parabolic}.

\section{Degrees of Invariant Polynomials} \label{a1_appendix}

We tabulate here the degrees of invariant polynomials of finite irreducible Coxeter systems.

\begin{center}
\begin{tabular}{c|c}
Coxeter type&Degrees of invariant polynomials\\
\hline 
$A_n$&$2,3,4,\ldots,n+1$\\
$B_n, C_n$&$2,4,6,\ldots,2n$\\
$D_n$&$2,4,6,\ldots,2n-2,n$\\
$E_6$&$2,5,6,8,9,12$\\
$E_7$&$2,6,8,10,12,14,18$\\
$E_8$&$2,8,12,14,18,20,24,30$\\
$F_4$&$2,6,8,12$\\
$G_2$&$2,6$\\
$H_3$&$2,6,10$\\
$H_4$&$2,12,20,30$\\
$I_2(m),m \geq 4$&$2,m$\\
\end{tabular}
\end{center}

Note that $I_2(6) \simeq G_2$ and $I_2(5)$ is sometimes denoted as $H_2$.

\subsection*{Acknowledgements}

We thank Florian Beck, Usha Bhosle, Marina Logares, Eyal Markman, Benedict Morrissey, Hiraku Nakajima, Carlos Simpson and Yan Soibelman for interesting discussions. This work grew out of discussions initiated during the discussion meeting \textit{``Quantum Fields, Geometry and Representation Theory 2018''} at ICTS-TIFR (\href{https://www.icts.res.in/discussion-meeting/qftgrt2018}{ICTS/qftgrt/2018/07}). AB's work is supported by the US Department of Energy under the grant DE{--}SC0010008. JD's work is supported by NSF grant PHY--1914679. JD would like to thank the Aspen Center for Physics (supported by NSF grant PHY--1607611) for hospitality while some of this work was conducted. During the preparation of this work, Ron Donagi was supported in part by NSF grant DMS 2001673 and by Simons HMS Collaboration grant \# 390287.

\bibliographystyle{utphys}
\bibliography{refs}
\end{document}